\documentclass[a4paper, 11pt, onecolumn]{article}

\usepackage[utf8]{inputenc} 
\usepackage[T1]{fontenc}
\usepackage{csquotes}
\usepackage[left=0.8in, right=0.8in, bottom=1.5in, top=1.5in]{geometry}

\usepackage{mathtools}	
\usepackage{amsthm, amssymb}
\usepackage{bm}
\usepackage{nicefrac}
\allowdisplaybreaks[1]

\usepackage[auth-sc]{authblk}

\usepackage{enumitem}
\usepackage{url}            
\usepackage{booktabs}       
\usepackage{microtype}      
\usepackage{hyperref}

\newcommand{\bQ}{{\mathbf{Q}}}
\newcommand{\bX}{{\mathbf{X}}}
\newcommand{\bx}{{\mathbf{x}}}
\newcommand{\bY}{{\mathbf{Y}}}
\newcommand{\by}{{\mathbf{y}}}
\newcommand{\bZ}{{\mathbf{Z}}}

\newcommand\smallO{
	\mathchoice
	{{\scriptstyle\mathcal{O}}}
	{{\scriptstyle\mathcal{O}}}
	{{\scriptscriptstyle\mathcal{O}}}
	{\scalebox{.7}{$\scriptscriptstyle\mathcal{O}$}}
}

\newcommand{\cN}{\mathcal{N}}
\newcommand{\cH}{\mathcal{H}}
\newcommand{\cL}{\mathcal{L}}
\newcommand{\cZ}{\mathcal{Z}}
\newcommand{\E}{\mathbb{E}}
\newcommand{\simiid}[1][0pt]{\mathrel{\raisebox{#1}{$\sim$}}}
\newcommand{\iid}{\overset{\text{\tiny i.i.d.}}{\simiid[-2pt]}}

\newtheorem{theorem}{Theorem}
\newtheorem{lemma}{Lemma}
\newtheorem{proposition}{Proposition}

\title{Mutual information for low-rank even-order symmetric tensor estimation}
\author[1]{\textsc{Cl\'{e}ment Luneau}}
\author[2]{Jean Barbier}
\author[1]{Nicolas Macris}
\affil[1]{Communication Theory Laboratory, \'{E}cole Polytechnique F\'{e}d\'{e}rale de Lausanne, Switzerland}
\affil[2]{The Abdus Salam International Center for Theoretical Physics, Trieste, Italy.}
\date{}

\begin{document}
	\maketitle
	
\begin{abstract}
We consider a statistical model for finite-rank symmetric tensor factorization and prove a single-letter 
variational expression for its asymptotic mutual information when the tensor is of even order.
The proof applies the adaptive interpolation method originally invented for rank-one factorization.
Here we show how to extend the adaptive interpolation to finite-rank and even-order tensors.
This requires new nontrivial ideas with respect to the current analysis in the literature. We also underline where the proof falls short when dealing with odd-order tensors.
\end{abstract}

\section{Introduction}
There exist well-known unsupervised algorithms to discover structure in a 2D dataset, e.g., singular value decomposition (SVD), principal component analysis (PCA) and other spectral methods \cite{Jolliffe2002Principal}.
Tensors naturally handle multidimensional data and their use becomes more and more beneficial with the emergence of big data, a strong incentive to go beyond the flat matrix world.
Tensor decompositions come with some advantages with respect to matrices, and 
have
numerous applications in signal processing and machine learning, e.g., data compression, data visualization, learning probabilistic latent variables 
models, etc. \cite{Cichocki2015Tensor,Rabanser2017Introduction}.
The canonical polyadic decomposition (CPD), also known as tensor rank decomposition or tensor factorization, is the most familiar one and represents a tensor 
as a minimum-length linear combination of rank-one tensors. 
This minimum-length defines the tensor rank. If instead the number $K$ of rank-one tensors forming the linear combination is not minimal, we talk of a $K$-term decomposition. 
Decompositions of tensors are also called tensor factorizations and this is the terminology we adopt in the rest of the paper.

One approach to explore computational and/or statistical limits of tensor factorization is to consider a statistical model, as done in \cite{Richard2014Statistical}. The model 
is the following: draw $K$ column vectors, evaluate for each of them their $p$\textsuperscript{th} tensor power and sum those $K$ symmetric order-$p$ tensors
(this sum is exactly a $K$-term polyadic decomposition).
Tensor factorization can then be studied as an inference problem, namely, to estimate the initial $K$ vectors from noisy observations of the tensor and to determine information theoretic limits for this task.
To do so, we focus on proving formulas for the asymptotic mutual information between the noisy observed tensor and the original $K$ vectors.
Such formulas were first rigorously derived for $p=2$ and $K=1$, i.e., rank-one matrix factorization:
see \cite{Korada2009Exact} for the case with a binary input vector,
\cite{Deshpande2016Asymptotic} for the restricted case in which no discontinuous phase transition occurs,
\cite{Krzakala2016Mutual} for a single-sided bound and, finally, \cite{Barbier2016Mutual} for the fully general case.
The proof in \cite{Barbier2016Mutual} combines interpolation techniques with spatial coupling and an analysis of the Approximate Message-Passing (AMP) algorithm.
Later, and still for $p=2$, \cite{Lelarge2019Fundamental} went beyond rank-one by using a rigorous version of the cavity method.
Reference \cite{Lesieur2017Statistical} applied the heuristic replica method to conjecture a formula for any $p$ and finite $K$, which is then proved for $p \geq 2$ and $K= 1$.
Reference \cite{Lesieur2017Statistical} also details the AMP algorithm for tensor factorization and shows how the single-letter variational expression for the mutual information allows one to give guarantees on AMP's performance.
Afterwards, \cite{Barbier2019Adaptivea,Barbier2019Adaptive} introduced the adaptive interpolation proof technique which they applied to the case $p \geq 2$, $K=1$.
Other proofs based on interpolations recently appeared, see \cite{Alaoui2018Estimation} ($p=2$, $K=1$) and \cite{Mourrat2018Hamilton}
($p\geq 2$, $K=1$).

In this work, we prove the conjectured replica formula for any finite rank $K$ and any even order $p$ using the adaptive interpolation method.
We also underline what is missing to extend the proof to odd orders.
The adaptive interpolation method was introduced in \cite{Barbier2019Adaptivea,Barbier2019Adaptive} as a powerful extension to the Guerra-Toninelli interpolation scheme \cite{Guerra2002Thermodynamic}.
Since then, it has been applied to many other inference problems in order to prove formulas for the mutual information, e.g., \cite{Barbier2017Layered,Barbier2019Optimal}.
While our proof outline is similar to \cite{Barbier2019Adaptive}, there are two important new ingredients.
First, to establish a tight lower bound on the asymptotic mutual information, we have to prove the regularity of a change of variable given by the solutions to an ordinary differential equation. This is nontrivial when the rank becomes greater than one.
Second, the same bound requires one to prove the concentration of the overlap (a quantity that fully characterizes the system in the high-dimensional limit).
When the rank is greater than one, this overlap is a matrix and a recent result \cite{Barbier2020Overlap} on the 
concentration of overlap matrices can be adapted to obtain the required concentration in our interpolation scheme.

The paper is organized as follows.
In Section \ref{sec:Low-rank_symmetric_tensor_factorization} we set up our precise statistical model and state the main theorems giving the single-letter variational expression for the asymptotic mutual information.
The adaptive interpolation method is formulated in Section~\ref{sec:Adaptive_path_interpolation} and the basic upper and lower bounds on the asymptotic mutual information are proved in Section~\ref{sec:matching_bounds}.
Sections \ref{sec:derivative_free_entropy} and \ref{sec:Barbier2020Overlap} contain the new and essential results which allow to go from rank-one to finite-rank tensors. 
Finally, the difficulties encountered for odd-order tensors are discussed in the last section, that is, Section~\ref{sec:discussion}.
The reader will find in Appendix~\ref{app:divergence} a technical calculation which is new and crucial to our proof, while Appendices \ref{app:properties_psi} and \ref{app:concentration_free_entropy} present more classical material.
\section{Low-rank symmetric tensor factorization}\label{sec:Low-rank_symmetric_tensor_factorization}
We study the following statistical model.
Let $n$ be a positive integer. $X_1,\dots,X_n$ are random column vectors in $\mathbb{R}^{K}$, independent and identically distributed (i.i.d.) with distribution $P_X$.
These vectors are not directly observed. Instead, for each $p$-tuple $i=(i_1,\dots,i_p) \in [n]^{p}$ with $i_1 \leq i_2 \leq \dots \leq i_p$, we observe
\begin{equation}\label{obs_model}
Y_{i}= \sqrt{\frac{\lambda (p-1)!}{n^{p-1}}} \sum_{k=1}^{K} X_{i_1 k}X_{i_2 k}\dots X_{i_p k}  + Z_{i}\;,
\end{equation}
where $\lambda$ is a known signal-to-noise ratio (SNR) and the noise $Z_{i}$ is i.i.d.\ with respect to the standard normal distribution $\cN(0,1)$.
Let $\bX$ be the $n \times K$ matrix whose $i$\textsuperscript{th} row is given by $X_i$.
All the observations \eqref{obs_model} are combined into the following symmetric order-$p$ tensor ($X_{\cdot,k} \in \mathbb{R}^n$ denotes the $k$\textsuperscript{th} column of $\bX$):
$$
\bY \coloneqq\sqrt{\nicefrac{\lambda (p-1)!}{n^{p-1}}} \sum_{k=1}^{K} X_{\cdot,k}^{\otimes p}  + \bZ\;.
$$
Our main result is the proof of a formula for the mutual information in the limit $n \to + \infty$ while the rank $K$ is kept fixed.
This formula is given as the optimization of a potential over the cone of $K \times K$ symmetric positive semidefinite matrices $\mathcal{S}_K^+$.
Let $\widetilde{Z} \sim \cN(0,I_K)$ and $X \sim P_X$.
Define the convex function (see Lemma~\ref{lemma:properties_psi} in Appendix~\ref{app:properties_psi})
\begin{equation}
\psi: S \in \mathcal{S}_K^+ \mapsto
\E \ln \int dP_X(x) \,\exp\Big(X^T S x + \widetilde{Z}^T\sqrt{S} x-\frac{1}{2}x^T S x\Big)\;;
\end{equation}
as well as the potential
\begin{equation}\label{potential_phi_p}
\phi_{p,\lambda}: S \in \mathcal{S}_K^+ \mapsto
\psi\big(\lambda S^{\circ (p-1)}\big)
-\frac{\lambda (p-1)}{2p}\sum_{\ell,\ell'=1}^K \big(S^{\circ p}\big)_{\ell \ell'} \;,
\end{equation}
where $S^{\circ k}$ is the $k$\textsuperscript{th} Hadamard power of $S$ (for two matrices $A$ and $B$ of the same dimension, the Hadamard product $A \circ B$ is the matrix of same dimension with elements given by $(A\circ B)_{ij}=A_{ij}B_{ij}$).
Note that, by the Schur Product Theorem \cite{Schur1911Bemerkungen}, the Hadamard product of two matrices in $\mathcal{S}_K^+$ is 
also in $\mathcal{S}_K^+$.
Let $\Sigma_{X} \coloneqq \E[X X^\intercal] \in \mathcal{S}_K^+$ the second moment matrix of a random vector $X \sim P_X$.
Our main result is the proof of the replica formula conjectured in \cite{Lesieur2017Statistical}, that is,
\begin{theorem}(Mutual information in the high-dimensional limit)\label{th:mutual_info}
Assume $p$ is even and $P_X$ is such that its first $2p$ moments are finite. Then
	\begin{equation}\label{formula_mutual_info}
	\lim_{n \to +\infty} \frac{1}{n} I(\bX ; \bY)
	= \frac{\lambda}{2p}\sum_{\ell,\ell'=1}^K \big(\Sigma_{X}^{\circ p}\big)_{\ell \ell'} - \sup_{S \in \mathcal{S}_K^+} \phi_{p,\lambda}(S) \,.
	\end{equation}
\end{theorem}
\noindent\textbf{Important remark:} We can reduce the proof of \eqref{formula_mutual_info} to the case $\lambda=1$ by rescaling properly $P_X$.
From now on, we set $\lambda=1$ and define $\phi_p \coloneqq \phi_{p,1}$.\\

Before proving Theorem~\ref{th:mutual_info}, we introduce important information theoretic quantities, adopting the statistical mechanics terminology.
Let $\mathcal{I} \coloneqq \{i \in [n]^p:i_a \leq i_{a+1}\}$.
Given the observations $\bY$, define the Hamiltonian for all $\bx \in \mathbb{R}^{n \times K}$:
\begin{equation}\label{hamiltonian}
\cH_n(\bx ; \bY)
\coloneqq
\sum_{i \in \mathcal{I}}
\frac{(p-1)!}{2n^{p-1}} \bigg(\sum_{\ell=1}^{K} \prod_{a=1}^{p}x_{i_a \ell}\bigg)^{\! 2}
-\sum_{i \in \mathcal{I}}\sqrt{\frac{ (p-1)!}{n^{p-1}}} Y_{i_1\dots i_p}  \sum_{\ell=1}^{K} \prod_{a=1}^{p}x_{i_a \ell}\:.
\end{equation}
Using Bayes' rule, the posterior probability density function is
\begin{equation*}\label{posterior}
dP(\bx \,\vert\, \bY)
\coloneqq \frac{1}{\cZ_n(\bY)} \, e^{-\cH_n(\bx; \bY)} \,\prod_{j=1}^n dP_{X}(x_j) \:,
\end{equation*}
with $\cZ_n(\bY) \coloneqq \int e^{-\cH_n(\bx; \bY)} \prod_{j} dP_{X}(x_j)$ the normalization factor.
Finally, the free entropy is the quantity
\begin{equation}\label{free_entropy}
 f_n  \coloneqq \frac{1}{n}\E\ln \cZ_{n}(\bY) \,,
\end{equation}
which is linked to the mutual information through the identity
\begin{equation}\label{link_in_fn}
\frac{1}{n} I(\bX ; \bY)
= \frac{1}{2p}\sum_{\ell,\ell'=1}^K \big(\Sigma_{X}^{\circ p}\big)_{\ell \ell'}
- f_n + \mathcal{O}(n^{-1}) \,.
\end{equation}
In \eqref{link_in_fn}, $\mathcal{O}(n^{-1})$ is a quantity such that $n\mathcal{O}(n^{-1})$ is bounded uniformly in $n$.
Thanks to \eqref{link_in_fn}, Theorem~\ref{th:mutual_info} will follow directly from the next two bounds on the asymptotic free entropy.
\begin{theorem}(Lower bound on the asymptotic free entropy)\label{th:lowerbound}
Assume $p$ is even and $P_X$ is such that its first $2p$ moments are finite. Then
\begin{equation}\label{liminf}
\liminf_{n \to +\infty} f_n \geq 
\sup_{S \in \mathcal{S}_K^+} \phi_p(S)\,.
\end{equation}
\end{theorem}
\begin{theorem}(Upper bound on the asymptotic free entropy)\label{th:upperbound}
	Assume $p$ is even and $P_X$ has bounded support.
	Then
	\begin{equation}\label{limsup}
	\limsup_{n \to +\infty} f_n \leq 
	\sup_{S \in \mathcal{S}_K^+} \phi_p(S)\,.
	\end{equation}
\end{theorem}
\noindent{\textbf{Important remark:}} Note that the assumption on $P_X$ in Theorem~\ref{th:upperbound} is stricter than the one in Theorem~\ref{th:mutual_info}.
Therefore, combining Theorem~\ref{th:lowerbound} and Theorem~\ref{th:upperbound} only proves the limit \eqref{formula_mutual_info} for a distribution $P_X$ which has bounded support.
The generalization to a distribution $P_X$ whose first $2p$ moments are finite is done by approaching $P_X$ with distributions having bounded support, much as it is done in \cite[Section 6.2.2]{Lelarge2019Fundamental}
\section{Adaptive path interpolation}\label{sec:Adaptive_path_interpolation}
We introduce a \textit{time} parameter $t \in [0,1]$.
The adaptive path interpolation interpolates from the original channel \eqref{obs_model} at $t=0$ to decoupled channels at $t=1$.
In between, we follow an interpolation path $R(\cdot,\epsilon): [0,1] \to \mathcal{S}_K^{+}$, which is a continuously differentiable function parametrized by a \textit{small perturbation} $\epsilon \in \mathcal{S}_K^{+}$ and such that $R(0,\epsilon)=\epsilon$.
More precisely, for $t \in [0,1]$, we observe:
\begin{align}\label{interpolation_model}
\begin{cases}
Y_{i}^{(t)} \;\;= \sqrt{\frac{(1-t)(p-1)!}{n^{p-1}}} \sum\limits_{k=1}^{K}\prod\limits_{a=1}^p X_{i_a k} + Z_{i}\:,\quad i \in \mathcal{I}\,;\\
\widetilde{Y}_{j}^{(t,\epsilon)} = \sqrt{R(t,\epsilon)}\, X_j + \widetilde{Z}_j\,, \qquad\qquad\quad\;\;\, j \in [n].
\end{cases}
\end{align}
The noise $\widetilde{Z}_j \iid \cN(0,I_K)$ is independent of both $\bX$ and $\bZ$.
Let $\widetilde{\bZ}$ be the $n \times K$ matrix whose $j$\textsuperscript{th} row is given by $\widetilde{Z}_j$ and $\widetilde{\bY}^{(t,\epsilon)} \coloneqq \sqrt{R(t,\epsilon)}\, \bX^T + \widetilde{\bZ}^T$.
The associated interpolating Hamiltonian reads:
\begin{multline}\label{interpolating_hamiltonian}
\cH_{t,\epsilon}(\bx ; \bY^{(t)}, \widetilde{\bY}^{(t,\epsilon)})
\coloneqq \sum_{i \in \mathcal{I}}
\frac{(1-t) (p-1)!}{2n^{p-1}} \bigg(\sum_{k=1}^{K}\prod_{a=1}^p x_{i_a k}\bigg)^{\! 2}
-\sum_{i \in \mathcal{I}}\sqrt{\frac{(1-t) (p-1)!}{n^{p-1}}} Y_{i}^{(t)}  \sum_{k=1}^{K} \prod_{a=1}^p x_{i_a k}\\
+\sum_{j=1}^{n} \frac{1}{2} x_j^T R(t,\epsilon) x_j- \big(\widetilde{Y}_j^{(t,\epsilon)}\big)^T  \sqrt{R(t,\epsilon)} x_j \:.
\end{multline}
The interpolating free entropy is defined similarly to the original free entropy \eqref{free_entropy}, that is,
\begin{equation}\label{interpolating_free_entropy}
f_n(t,\epsilon) \coloneqq \frac1n \E \ln \cZ_{t,\epsilon}(\bY^{(t)},\widetilde{\bY}^{(t,\epsilon)})
\end{equation}
with $\cZ_{t,\epsilon}(\bY^{(t)},\widetilde{\bY}^{(t,\epsilon)}) \coloneqq \int e^{-\cH_{t,\epsilon}(\bx ; \bY^{(t)},\widetilde{\bY}^{(t,\epsilon)})} \prod_{j = 1}^n dP_{X}(x_j)$.
Evaluating \eqref{interpolating_free_entropy} at both extremes gives:
\begin{align}\label{extremes_interpolating_free_entropy}
\begin{cases}
f_n(0,\epsilon) = f_n + \mathcal{O}(\Vert \epsilon \Vert) \,;\\
f_n(1,\epsilon) = \psi(R(1,\epsilon)) \;.
\end{cases}
\end{align}
$\Vert \cdot \Vert$ denotes the Frobenius norm and $\mathcal{O}(\Vert \epsilon \Vert)$ is a quantity such that $\vert \mathcal{O}(\Vert \epsilon \Vert)\vert \leq \nicefrac{\mathrm{Tr}(\Sigma_{X})\Vert \epsilon \Vert}{2}$.
In order to deal with future computations, it is useful to introduce the Gibbs brackets $\langle - \rangle_{t,\epsilon}$
that denote an expectation with respect to the posterior distribution, i.e., 
\begin{equation}
\langle g(\bx) \rangle_{t,\epsilon} = \int \!g(\bx)\, \frac{e^{-\cH_{t,\epsilon}(\bx ; \bY^{(t)},\widetilde{\bY}^{(t,\epsilon)})}}{\cZ_{t,\epsilon}(\bY^{(t)},\widetilde{\bY}^{(t,\epsilon)}) }
\,\prod_{j=1}^n dP_{X}(x_j) \;.
\end{equation}
Combining \eqref{extremes_interpolating_free_entropy} with the fundamental theorem of calculus
$f_n(0,\epsilon)=f_n(1,\epsilon)-\int_{0}^{1}f_n^{\prime}(t,\epsilon) dt$,
we obtain the sum-rule of the adaptive path interpolation.
\begin{proposition}[Sum-rule]\label{prop:sum_rule}
Assume $P_X$ has finite $(2p)$\textsuperscript{th}-order moments.
Denote $R^{\prime}(\cdot,\epsilon)$ the derivative of the interpolation path $R(\cdot,\epsilon)$.
Let $Q_{\ell \ell'} \coloneqq \frac{1}{n}\sum_{j=1}^{n} x_{j \ell} X_{j \ell'}$ be 
the entries of the $K \times K$ overlap matrix $\bQ \coloneqq \frac1n\bx^\intercal\bX$.
Then
\begin{equation}\label{sum_rule}
f_n = \mathcal{O}(\Vert \epsilon \Vert) + \mathcal{O}(n^{-1}) + \psi(R(1,\epsilon))
	+ \frac{1}{2p} \int_{0}^{1} dt \sum_{\ell,\ell'=1}^K
	\E \langle ( Q_{\ell \ell'})^{p} \rangle_{t,\epsilon} - p (R^{\prime}(t,\epsilon))_{\ell \ell'} \E \langle Q_{\ell \ell'}\rangle_{t,\epsilon} \,,
\end{equation}
where $\mathcal{O}(n^{-1})$ and $\mathcal{O}(\Vert \epsilon \Vert)$ are independent of $\epsilon$ and $n$, respectively.
\end{proposition}
\begin{proof}
	See Section~\ref{sec:derivative_free_entropy} for the computation of $f_n^{\prime}(t,\epsilon)$, that is, the $t$-derivative of $f_n(\cdot,\epsilon)$.
\end{proof}
\section{Matching bounds}\label{sec:matching_bounds}
In this section we prove both Theorems \ref{th:lowerbound} and \ref{th:upperbound} by plugging two different choices for $R(\cdot,\epsilon)$ in the sum-rule \eqref{sum_rule}.
\subsection{Lower bound: proof of Theorem~\ref{th:lowerbound}}
A lower bound on $f_n$ is obtained by choosing the interpolation function $R(t,0)=tS^{\circ(p-1)}$ with $S$ a $K \times K$ symmetric positive semidefinite matrix, i.e., $\epsilon=0$ and $R^{\prime}(t,\epsilon)=S^{\circ (p-1)}$. Then the sum-rule \eqref{sum_rule} reads
\begin{equation}\label{sum_rule_lowerbound}
f_n = \mathcal{O}(n^{-1}) + \phi_p(S)
+ \frac{1}{2p} \int_{0}^{1} dt \sum_{\ell,\ell'=1}^K \E \big\langle h_p(S_{\ell \ell'}, Q_{\ell \ell'}) \big\rangle_{t,0} \quad,
\end{equation}
where $h_p(r,q) \coloneqq q^p - p q r^{p-1} + (p-1)r^p$. If $p$ is even then $h_p$ is nonnegative on $\mathbb{R}^2$ and \eqref{sum_rule_lowerbound} directly implies $f_n \geq \phi_p(S) + \mathcal{O}(n^{-1})$.
Taking the inferior limit on both sides of this inequality, and bearing in mind that the inequality is valid for all $S \in \mathcal{S}_K^+$, ends the proof of Theorem~\ref{th:lowerbound}. \hfill\ensuremath{\square}

We have at our disposal a wealth of interpolation paths when considering any continuously differentiable $R(\cdot,\epsilon)$.
However, to establish the lower bound \eqref{liminf}, we have used a simple linear 
interpolation, i.e., $R'(t,\epsilon)=S^{\circ (p-1)}$.
Such an interpolation dates back to Guerra \cite{Guerra2002Thermodynamic} and was already used by \cite{Lesieur2017Statistical,Lelarge2019Fundamental} to derive the lower bound \eqref{liminf} for both cases $K=1$, any order $p$, and $p=2$, any finite rank $K$.
Now that we turn to the proof of the upper bound \eqref{limsup}, we will see how the flexibility in the choice of $R(\cdot,\epsilon)$ constitutes an improvement on the classical interpolation.

\subsection{Upper bound: proof of Theorem \ref{th:upperbound}}
\subsubsection{Interpolation determined by an ordinary differential equation (ODE)}
The sum-rule~\eqref{sum_rule} suggests to pick an interpolation path satisfying
\begin{equation}\label{choice_interpolation_upperbound}
\forall (\ell,\ell') \in \{1,\dots,K\}^2 \! : (R^{\prime}(t,\epsilon))_{\ell \ell'}=\E[\langle Q_{\ell \ell'}\rangle_{t,\epsilon}]^{p-1}.
\end{equation}
The integral in \eqref{sum_rule} can then be split in two terms:
one similar to the second summand in \eqref{potential_phi_p},
and one that will vanish in the high-dimensional limit if the overlap concentrates.
The next proposition states that \eqref{choice_interpolation_upperbound} indeed admits a solution, a fact which is not obvious because the Gibbs brackets $\langle - \rangle_{t,\epsilon}$ themselves depend on $R(\cdot,\epsilon)$.
Nontrivial properties required to show the upper bound \eqref{limsup} are also proved.
\begin{proposition}\label{prop:ode_and_properties}
For all $\epsilon \in \mathcal{S}_K^+$, there exists a unique global solution $R(\cdot,\epsilon): [0,1] \to \mathcal{S}_K^+$ to the  first-order ODE 
\begin{equation*}
\forall \,t \in [0,1]: \frac{d R(t)}{dt}=\E[\langle \bQ \rangle_{t,\epsilon}]^{\circ(p-1)}\,,\, R(0)=\epsilon\,.
\end{equation*}
This solution is continuously differentiable and bounded.
If $p$ is even then $\forall \,t \in [0,1]$, $R(t,\cdot)$ is a $\mathcal{C}^1$-diffeomorphism from $\mathcal{S}_K^{++}$ (the open cone of $K \times K$ symmetric positive definite matrices) into $R(t,\mathcal{S}_K^{++})$ whose Jacobian determinant is greater than one, i.e.,
\begin{equation}\label{jacobian_greater_one}
	\forall \,\epsilon \in \mathcal{S}_K^{++}: \det J_{R(t,\cdot)}(\epsilon) \geq 1 \,.
\end{equation}
Here $J_{R(t,\cdot)}$ denotes the Jacobian matrix of $R(t,\cdot)$.
\end{proposition}
\begin{proof}
We now rewrite \eqref{choice_interpolation_upperbound} explicitly as an ODE. Let $R$ be a matrix in $\mathcal{S}_K^{+}$. Consider the problem of inferring $\bX$ from the following observations:
\begin{align}\label{inference_problem_t_R}
\begin{cases}
Y_{i}^{(t)} &= \sqrt{\frac{(1-t)(p-1)!}{n^{p-1}}} \sum\limits_{k=1}^{K}\prod\limits_{a=1}^p X_{i_a k} + Z_{i}\:,\quad i \in \mathcal{I}\,;\\
\widetilde{Y}_{j}^{(t,R)} &= \sqrt{R}\, X_j + \widetilde{Z}_j\:, \qquad\qquad\qquad\quad\;\;\;\: j \in [n].
\end{cases}
\end{align}
It is reminiscent of the interpolating problem~\eqref{interpolation_model}.
We can form a Hamiltonian similar to~\eqref{interpolating_hamiltonian}, where $R(t,\epsilon)$ is simply replaced by $R$, and $\langle - \rangle_{t,R}$ are the Gibbs brackets associated to the posterior of this model.
We define the function
\begin{equation}
G_n:
\begin{array}{ccl}
		[0,1] \times \mathcal{S}_K^+ &\to& \mathcal{S}_K^+\\
(t,R) &\mapsto& \E[\langle \bQ \rangle_{t,R}]^{\circ(p-1)}
\end{array}\,.
\end{equation}
Note that $\E\langle \bQ \rangle_{t,R}$ is a symmetric positive semidefinite matrix.
Indeed, from the Nishimori identity\footnote{
The  Nishimori  identity  is  a  direct  consequence  of  the  Bayes  formula.
In our setting, it states $\E\langle g(\bx,\bX)\rangle_{t,R} = \E\langle g(\bx,\bx') \rangle_{t,R} = \E\langle g(\bX,\bx) \rangle_{t,R}$ where $\bx,\bx'$ are two samples drawn independently from the posterior distribution given $\bY^{(t)}$, $\widetilde{\bY}^{(t,R)}$.
Here $g$ can also explicitly depend on $\bY^{(t)}$, $\widetilde{\bY}^{(t,R)}$.
},
$\E\langle \bQ \rangle_{t,R} = n^{-1} \E[\langle \bx \rangle_{t,R}^\intercal \bX]
= n^{-1}\E[\langle \bx \rangle_{t,R}^\intercal \langle \bx \rangle_{t,R}]$.
By the Schur Product Theorem \cite{Schur1911Bemerkungen}, the Hadamard power $\E[\langle \bQ \rangle_{t,R}]^{\circ(p-1)}$ also belongs to $\mathcal{S}_K^+$, justifying that $G_n$ takes values in the cone of symmetric positive semidefinite matrices.
$G_n$ is continusouly differentiable on $[0,1] \times \mathcal{S}_K^+$.
By the Cauchy-Lipschitz theorem, there exists a unique global solution $R(\cdot,\epsilon)$ to the $K(K+1)/2$-dimensional ODE:
\begin{equation}\label{ODE}
		\forall t \in [0,1]:\;\frac{d R(t)}{dt} = G_n(t,R(t)) \;,\; R(0)=\epsilon \in S_K^{+} \;.
\end{equation}
Each initial condition $\epsilon \in S_K^{+}$ is tied to a unique solution $R(\cdot,\epsilon)$. This implies that the function $\epsilon \mapsto R(t,\epsilon)$ is injective. Its Jacobian determinant is given by Liouville's formula \cite{Hartman2002Ordinary}:
\begin{equation}\label{liouville_formula}
\det J_{R(t,\cdot)}(\epsilon)
= \exp \int_0^t ds \!\!\sum_{1\leq \ell \leq \ell' \leq K} \frac{\partial (G_n)_{\ell \ell'}}{\partial R_{\ell \ell'}}\bigg\vert_{s,R(s,\epsilon)} \,.
\end{equation}
Thanks to the identity \eqref{liouville_formula}, we can show that the Jacobian determinant is greater than (or equal to) one by proving that the divergence 
$$
\sum_{1\leq \ell \leq \ell' \leq K} \frac{\partial (G_n)_{\ell \ell'}}{\partial R_{\ell \ell'}}\bigg\vert_{s,R(s, \epsilon)}
$$
is nonnegative for all $(s,R) \in [0,1] \times \mathcal{S}_K^{+}$.
By Lemma~\ref{lemma:divergence_Fn} in Appendix \ref{app:divergence}, the divergence reads (we omit the subscripts of the Gibbs brackets $\langle - \rangle_{t,R}$):
\begin{equation}\label{divergence_Fn}
\sum_{\ell \leq \ell'} \frac{\partial (G_n)_{\ell \ell'}}{\partial R_{\ell \ell'}}\bigg\vert_{t,R}
= n(p-1)\sum_{\ell,\ell'} \E[\langle Q_{\ell \ell'} \rangle\,]^{p-2}\,\Delta_{\ell\ell'} \;,
\end{equation}
where
\begin{equation}\label{def_delta}
\Delta_{\ell \ell'} \coloneqq
\E\bigg[\bigg\langle \bigg(\frac{Q_{\ell \ell'} + Q_{\ell' \ell}}{2} - \bigg\langle \frac{Q_{\ell \ell'} + Q_{\ell' \ell}}{2} \bigg\rangle\bigg)^{\!\! 2} \,\bigg\rangle\bigg]
-\E\bigg[\bigg(\bigg\langle\frac{Q_{\ell \ell'} + Q_{\ell' \ell}}{2}\bigg\rangle
\!\!- \frac{(\langle\bx\rangle^T \langle\bx\rangle)_{\ell \ell'}}{n} \bigg)^{\!\! 2}\,\bigg].
\end{equation}
If $p$ is even then $\E[\langle Q_{\ell \ell'} \rangle_{t,R}]^{p-2}$ is nonnegative.
We show next that the $\Delta_{\ell \ell'}$'s are nonnegative, thus ending the proof of \eqref{jacobian_greater_one}.
The second expectation on the right-hand side (r.h.s.) of \eqref{def_delta} satisfies:
\begin{align*}
\E\,\bigg(\bigg\langle\frac{Q_{\ell\ell'} + Q_{\ell'\ell}}{2}\bigg\rangle - \frac{(\langle\bx\rangle^\intercal \langle\bx\rangle)_{\ell \ell'}}{n} \bigg)^{\!\! 2}
	&= \E\,\bigg\langle\frac{(\bx^\intercal\bX + \bX^\intercal\bx)_{\ell \ell'}}{2n} - \frac{(\langle\bx\rangle^\intercal \bx + \bx^\intercal \langle\bx\rangle)_{\ell \ell'}}{2n}\bigg\rangle^{\!\! 2}\\
	&\leq \E\,\bigg\langle\bigg(\frac{(\bx^\intercal\bX + \bX^\intercal\bx)_{\ell\ell'}}{2n} - \frac{(\langle\bx\rangle^\intercal \bx + \bx^\intercal \langle\bx\rangle)_{\ell\ell'}}{2n}\bigg)^{\!\! 2}\,\bigg\rangle\\
	&= \E\,\bigg\langle\Big(\frac{(\bX^\intercal\bx + \bx^\intercal\bX)_{\ell\ell'}}{2n}
	- \frac{(\langle\bx\rangle^\intercal \bX + \bX^\intercal \langle\bx\rangle)_{\ell \ell'}}{2n}\Big)^{\!\! 2}\,\bigg\rangle \\
	&= \E\,\bigg\langle\bigg(\frac{Q_{\ell' \ell}+Q_{\ell \ell'}}{2} - \bigg\langle\frac{ Q_{\ell \ell'}+Q_{\ell'\ell}}{2}\bigg\rangle\bigg)^{\!\! 2}\,\bigg\rangle\,.
\end{align*}
The inequality is a simple application of Jensen's inequality, while the equality that follows is an application of the Nishimori identity.
The final upper bound is nothing but the first expectation on the r.h.s. of \eqref{def_delta}.
Therefore, $\forall (\ell,\ell') \in \{1,\dots,K\}^2: \Delta_{\ell\ell'} \geq 0$.
\end{proof}
\subsubsection{Proof of Theorem~\ref{th:upperbound}}
Let $\epsilon$ be a symmetric positive definite matrix, i.e., $\epsilon \in \mathcal{S}_K^{++}$.
We interpolate with the unique solution $R(\cdot,\epsilon): [0,1] \mapsto \mathcal{S}_K^{++}$ to \eqref{choice_interpolation_upperbound}.
The sum-rule \eqref{sum_rule} then reads:
\begin{multline}
f_n
= \mathcal{O}(\Vert \epsilon \Vert) + \mathcal{O}(n^{-1}) + \psi(R(1,\epsilon))
-\frac{p-1}{2p}\sum_{\ell,\ell'=1}^K \int_{0}^{1} \!dt\, \E[\langle Q_{\ell \ell'} \rangle_{t,\epsilon}]^{p}\\
+ \int_{0}^{1} \frac{dt}{2p}\sum_{\ell,\ell'=1}^K 
\E \big\langle Q_{\ell \ell'}
\big((Q_{\ell \ell'})^{p-1} - \E[\langle Q_{\ell \ell'} \rangle_{t,\epsilon}]^{p-1}\big)\big\rangle_{t,\epsilon}
\:.\label{sum_rule_upperbound}
\end{multline}
Using first the Lipschitz continuity of $\psi$ and then its convexity (see Lemma~\ref{lemma:properties_psi}, Appendix~\ref{app:properties_psi}), it comes:
\begin{align}\label{inequality_psi_R(1,eps)}
\psi(R(1,\epsilon))
= \psi\Big(\epsilon + \int_{0}^{1}dt \,\E[\langle \bQ \rangle_{t,\epsilon}]^{\circ(p-1)}\Big)
&= \mathcal{O}(\Vert \epsilon\Vert) + \psi\Big(\int_{0}^{1}dt \,\E[\langle \bQ \rangle_{t,\epsilon}]^{\circ(p-1)}\Big)\nonumber\\
&\leq \mathcal{O}(\Vert \epsilon\Vert) + \int_{0}^{1}dt \:\psi\Big(\E[\langle \bQ \rangle_{t,\epsilon}]^{\circ(p-1)}\Big)\;,
\end{align}
with $\vert \mathcal{O}(\Vert \epsilon\Vert) \vert \leq \frac{\mathrm{Tr}\,\Sigma_X}{2} \Vert\epsilon\Vert$.
Combining both \eqref{sum_rule_upperbound} and \eqref{inequality_psi_R(1,eps)} directly gives:
\begin{align}
f_n
&\leq \mathcal{O}(n^{-1}) + \mathcal{O}(\Vert \epsilon\Vert) + \int_{0}^{1}dt \,\phi_p\big(\E\,\langle \bQ \rangle_{t,\epsilon}\big)
+ \int_{0}^{1} \frac{dt}{2p} \sum_{\ell,\ell'=1}^K 
\E \big\langle Q_{\ell \ell'}
\big((Q_{\ell \ell'})^{p-1} - \E[\langle Q_{\ell \ell'} \rangle_{t,\epsilon}]^{p-1}\big)\big\rangle_{t,\epsilon}\nonumber\\
&\leq \mathcal{O}(n^{-1}) + \mathcal{O}(\Vert \epsilon\Vert) + \sup_{S \in \mathcal{S}_K^+} \phi_p(S)
+ \int_{0}^{1} \frac{dt}{2p} \sum_{\ell,\ell'=1}^K 
\E \big\langle Q_{\ell \ell'}
\big((Q_{\ell \ell'})^{p-1} - \E[\langle Q_{\ell \ell'} \rangle_{t,\epsilon}]^{p-1}\big)\big\rangle_{t,\epsilon}. \label{fn_upperbound_remainder}
\end{align}
In order to end the proof of \eqref{limsup}, we must show that the last integral term in the upper bound \eqref{fn_upperbound_remainder} vanishes when $n$ goes to infinity.
This will be the case if the overlap matrix $\bQ$ concentrates around its expectation $\E \langle \bQ \rangle_{t,\epsilon}$.
Indeed, provided that the $(4p-4)$\textsuperscript{th}-order moments of $P_X$ are finite, there exists a constant $C_{X}$ depending only on $P_X$ such that
\begin{equation}
\bigg\vert \int_{0}^{1} \frac{dt}{2p} \sum_{\ell,\ell'} 
\E \big\langle Q_{\ell \ell'}
\big((Q_{\ell \ell'})^{p-1}-\E[\langle Q_{\ell \ell'} \rangle_{t,\epsilon}]^{p-1}\big)\big\rangle_{t,\epsilon}\bigg\vert
\leq \frac{C_{X}}{2}
\int_{0}^{1} dt\,
	\E\big[\big\langle \big\Vert \bQ-\E[\langle\bQ\rangle_{t,\epsilon}]\big\Vert^2\,\big\rangle_{ t,\epsilon}\,\big]^{\nicefrac{1}{2}}\,. \label{upperbound_remainder}
\end{equation}
However, proving that the r.h.s. of \eqref{upperbound_remainder} vanishes is only possible after integrating on a well-chosen set of \textit{perturbations} $\epsilon$ (that play the role of initial conditions in the ODE \eqref{ODE}).
In essence, the integration over $\epsilon$ smoothens the phase transitions that might appear for particular choices of $\epsilon$ when $n$ goes to infinity.
We now describe the set of perturbations on which to integrate.

Let $(s_n)_{n \in \mathbb{N}^*}$ be a decreasing sequence of real numbers in $(0,1)$ and define the sequence of subsets: 
\begin{equation}\label{definition_En}
\mathcal{E}_n
\coloneqq \left\{ \epsilon \in \mathbb{R}^{K \times K}\,\middle\vert\!\!
\begin{array}{l}
\forall\, \ell \neq \ell': \epsilon_{\ell \ell'}=\epsilon_{\ell' \ell} \in [s_n,2s_n]\\
\forall\, \ell: \,\epsilon_{\ell \ell} \in  [2Ks_n ,(2K+1)s_n]
\end{array}\!\!\right\}.
\end{equation}
Those are subsets of symmetric strictly diagonally dominant matrices with positive diagonal entries, hence they are included in $\mathcal{S}_K^{++}$ (see \cite[Corollary 7.2.3]{Horn2012Matrix}).
As $\mathcal{E}_n$ is a $K(K+1)/2$-dimensionnal hypercube whose side has length $s_n$, its volume is $V_{\mathcal{E}_n}=s_n^{\nicefrac{K(K+1)}{2}}$.

Remember that, as per Proposition \ref{prop:ode_and_properties}, for every $\epsilon \in \mathcal{E}_n$ the interpolation path is chosen as the unique solution $R(\cdot,\epsilon): [0,1] \mapsto \mathcal{S}_K^{++}$ to $R'(t,\epsilon) = \E[\langle \bQ \rangle_{t,\epsilon}]^{\circ (p-1)}$.
Then, for a fixed $t \in [0,1]$, using Cauchy-Schwarz inequality and the change of variable $\epsilon \to R \equiv R(t,\epsilon)$ -- which is justified because $\epsilon \mapsto R(t,\epsilon)$ is a $\mathcal{C}^1$-diffeomorphism (see Proposition \ref{prop:ode_and_properties}) --, we obtain:
\begin{align}
\frac{1}{V_{\mathcal{E}_n}}\int_{\mathcal{E}_n} d\epsilon\,
\E\big[\big\langle \big\Vert \bQ-\E[\langle\bQ\rangle_{t,\epsilon}]\big\Vert^2\,\big\rangle_{ t,\epsilon}\,\big]^{\!\nicefrac{1}{2}}
&\leq \frac{1}{V_{\mathcal{E}_n}^{\nicefrac{1}{2}}}\,
\bigg(\int_{\mathcal{E}_n} d\epsilon\,\E \big\langle \big\Vert \bQ-\E[\langle\bQ\rangle_{t,\epsilon}]\big\Vert^2\,\big\rangle_{ t,\epsilon} \bigg)^{\!\nicefrac{1}{2}}\nonumber\\
&= 
\bigg(\frac{1}{V_{\mathcal{E}_n}}\, \int_{\mathcal{R}_{n,t}} \frac{dR}{\vert \det J_{R(t,\cdot)}(\epsilon) \vert}
\E \big\langle \big\Vert \bQ-\E[\langle\bQ\rangle_{t,R}]\big\Vert^2\,\big\rangle_{ t,R}\bigg)^{\!\nicefrac{1}{2}}\nonumber\\
&\leq 
\bigg(\frac{1}{V_{\mathcal{E}_n}}\, \int_{\mathcal{R}_{n,t}} dR\:\E\big[\big\langle \big\Vert \bQ-\E[\langle\bQ\rangle_{t,R}]\big\Vert^2\,\big\rangle_{t,R}\,\big]\bigg)^{\! \nicefrac{1}{2}}\,.
\label{upperbound_integral_En/VEn}
\end{align}
We introduced the notation $\mathcal{R}_{n,t} \coloneqq R(t,\mathcal{E}_n)$ while $\langle - \rangle_{t,R}$ are still the Gibbs brackets associated to the posterior distribution of the inference problem \eqref{inference_problem_t_R}.
The last inequality follows from \eqref{jacobian_greater_one}.
It will be easier to work with the convex hulls of $\mathcal{R}_{n,t}$, denoted $\mathrm{C}(\mathcal{R}_{n,t})$.
These convex hulls are uniformly bounded compact sets of $\mathcal{S}_K^{++}$.
Indeed, every $\mathcal{R}_{n,t}$ is compact and included in the convex set
\begin{equation}
\mathcal{B}(\Sigma_{X},K,p) = \big\{S \in \mathcal{S}_K^{++}: \Vert S \Vert \leq 4K^{\nicefrac{3}{2}} + \mathrm{Tr}(\Sigma_X)^{p-1}\big\} \,,
\end{equation}
which does not depend on $n$ and $t$ (see Section~\ref{sec:Barbier2020Overlap}, property (i) of Lemma \ref{lemma:properties_convex_hulls}).
Note that the upper bound \eqref{upperbound_integral_En/VEn}, the inclusion $\mathcal{R}_{n,t} \subseteq \mathrm{C}(\mathcal{R}_{n,t})$ and the nonnegativity of the integrand directly imply:
\begin{equation}\label{upperbound_integral_En}
\frac{1}{V_{\mathcal{E}_n}}\int_{\mathcal{E}_n} d\epsilon\,
\E\big[\big\langle \big\Vert \bQ-\E[\langle\bQ\rangle_{t,\epsilon}]\big\Vert^2\,\big\rangle_{ t,\epsilon}\,\big]^{\!\nicefrac{1}{2}}
\leq 
\bigg(\frac{1}{V_{\mathcal{E}_n}}\, \int_{C(\mathcal{R}_{n,t})} dR\:\E\big[\big\langle \big\Vert \bQ-\E[\langle\bQ\rangle_{t,R}]\big\Vert^2\,\big\rangle_{t,R}\,\big]\bigg)^{\! \nicefrac{1}{2}} \,.
\end{equation}
By Theorem \ref{th:Barbier2020Overlap} in Section~\ref{sec:Barbier2020Overlap}, there exists a positive constant $C$ which depends \textit{only} on $P_X$, $K$ and $p$ such that:
\begin{equation}
\int_{\mathrm{C}(\mathcal{R}_{n,t})} \!\!dR\:\E\big[\big\langle \big\Vert \bQ-\E[\langle\bQ\rangle_{t,R}]\big\Vert^2\,\big\rangle_{t,R} \,\big]
\leq \frac{C}{s_n^{\nicefrac{3}{2}}n^{\nicefrac{1}{6}}}\quad.\label{upperbound_theorem}
\end{equation}
Combining~\eqref{upperbound_remainder},~\eqref{upperbound_integral_En} and \eqref{upperbound_theorem}, we finally get:
\begin{equation}
\bigg\vert \int_{\mathcal{E}_n} \frac{d\epsilon}{V_{\mathcal{E}_n}}
\int_{0}^{1} \frac{dt}{2p} \sum_{\ell,\ell'} 
\E \big\langle Q_{\ell \ell'}
\big((Q_{\ell \ell'})^{p-1}-\E[\langle Q_{\ell \ell'} \rangle_{t,\epsilon}]^{p-1}\big)\big\rangle_{\! t,\epsilon}\bigg\vert
\leq \frac{C_{X}}{2}
\sqrt{\frac{C}{\big(s_n^{9+3K(K+1)}n\big)^{\nicefrac{1}{6}}}}\,. \label{final_bound_remainder}
\end{equation}
To conclude the proof, we have to further constrain $s_n$ to satisfy both $s_n \to 0$ and $s_n^{9+3K(K+1)}n \to +\infty$ when $n \to +\infty$. E.g., $s_n = (0.99/n)^{\alpha}$ with $0<\alpha<(9 + 3K(K+1))^{-1}$ is a valid choice.
Under this constraint, the upper bound \eqref{final_bound_remainder} vanishes in the high-dimensional limit.
Integrating the inequality~\eqref{fn_upperbound_remainder} over $\epsilon \in \mathcal{E}_n$ and, then, making use of the vanishing upper bound \eqref{final_bound_remainder} as well as
\begin{equation*}
\frac1{V_{\mathcal{E}_n}}\int_{\mathcal{E}_n} \!d\epsilon\,\mathcal{O}(\Vert \epsilon\Vert)
\leq \mathcal{O}(1) \,\max_{\epsilon \in \mathcal{E}_n} \Vert \epsilon \Vert
= \mathcal{O}(1) \,s_n = \smallO_{n}(1)\;,
\end{equation*}
give the inequality
$f_n = V_{\mathcal{E}_n}^{-1} \int_{\mathcal{E}_n} d\epsilon\,f_n
\leq \sup_{S \in \mathcal{S}_K^+} \phi_p(S) + \smallO_{n}(1)$.
The upper bound \eqref{limsup} follows simply, thus ending the proof of Theorem~\ref{th:upperbound}.
\hfill\ensuremath{\square}
\section{Time-derivative of the average interpolating free entropy}\label{sec:derivative_free_entropy}
In order to prove the sum-rule in Proposition~\ref{prop:sum_rule}, we need to compute the derivative of the averaged interpolating free entropy \eqref{interpolating_free_entropy} with respect to $t$.
We recall that $R^{\prime}(\cdot,\epsilon)$ denotes the derivative of $R(\cdot,\epsilon)$ and that the overlap matrix is $\bQ = \frac{1}{n} \bx^T \bX \in \mathbb{R}^{K \times K}$, that is,
\begin{equation}
\forall (\ell,\ell') \in \{1,\dots,K\}^2:Q_{\ell \ell'}=\frac{1}{n}\sum_{j=1}^n x_{j\ell} X_{j\ell'}\,.
\end{equation}
\begin{proposition}[Derivative of the average interpolating free entropy]\label{prop:formula_derivative_free_entropy}
Assume that $P_X$ has finite $(2p)$\textsuperscript{th}-order moments.
Consider the average free entropy \eqref{interpolating_free_entropy}. Its derivative with respect to $t$ satisfies:
\begin{equation}\label{formula_derivative_free_entropy}
\forall t \in [0,1]:
f'_{n}(t,\epsilon)
= -\frac{1}{2p} \sum_{\ell,\ell'=1}^{K} 
\E\big[\big\langle \big( Q_{\ell \ell'}\big)^{p} \,\big\rangle_{t,\epsilon} \,\big]
+ \frac{1}{2}\mathrm{Tr}\big(R^{\prime}(t,\epsilon) \,\E\,\langle\bQ\rangle_{t,\epsilon} \big)
+ \mathcal{O}_{n}(n^{-1})\;.
\end{equation}
Here $\mathcal{O}_{n}(n^{-1})$ is a quantity such that $n\mathcal{O}_{n}(n^{-1})$ is bounded uniformly in $n$, $t$ and $\epsilon$.
\end{proposition}
\begin{proof}
Note that the conditional probability density function of $(\bY^{(t)},\widetilde{\bY}^{(t,\epsilon)})$ given $\bX=\bx^{\star}$ reads:
\begin{equation}
p_{\bY^{(t)},\widetilde{\bY}^{(t,\epsilon)}\vert\bX}(\by,\widetilde{\by}\,\vert\,\bx^{\star}) =\frac{1}{\sqrt{2\pi}^{nK + \vert \mathcal{I} \vert}}
\exp\bigg(-\sum_{i \in \mathcal{I}} \frac{y_i^2}{2} -\frac{\Vert \widetilde{\by} \Vert^2}{2}-\cH_{t,\epsilon}(\bx^{\star}; \by,\widetilde{\by})\bigg) \;.
\end{equation}
Therefore, the average interpolating free entropy satisfies:
\begin{equation}\label{eq:precise_formula_f}
f_n(t,\epsilon)
= \frac{1}{n} \E_{\bX}\bigg[\int d\by d\widetilde{\by}\:\frac{e^{-\sum_{i \in \mathcal{I}} \frac{y_i^2}{2} -\frac{\Vert \widetilde{\by} \Vert^2}{2}}}{\sqrt{2\pi}^{nK + \vert \mathcal{I} \vert}} e^{-\cH_{t,\epsilon}(\bX; \by,\widetilde{\by})}
\ln \cZ_{t,\epsilon}\big(\by,\widetilde{\by}\big) \bigg] \:.
\end{equation}
Taking the time-derivative of \eqref{eq:precise_formula_f}, we get:
\begin{align}
f'_n(t,\epsilon)
&= -\frac{1}{n} \E\Big[\cH'_{t,\epsilon}\big(\bX; \bY^{(t)},\widetilde{\bY}^{(t,\epsilon)}\big)\ln \cZ_{t,\epsilon}(\bY^{(t)},\widetilde{\bY}^{(t,\epsilon)})\Big]
-\frac{1}{n} \E\Big[\Big\langle \cH'_{t,\epsilon}\big(\bx; \bY^{(t)},\widetilde{\bY}^{(t,\epsilon)}\big)\Big\rangle_{\!t,\epsilon} \,\Big]
\nonumber\\
&= -\frac{1}{n} T_1 - \frac{1}{n} T_2 \;,
\end{align}
where $T_1$, $T_2$ are given by the two expectations $\mathbb{E}[-]$ and 
\begin{multline}\label{derivative_hamiltonian}
\cH'_{t,\epsilon}(\bx ; \bY,\widetilde{\bY})
= \sum_{i \in \mathcal{I}}
-\frac{(p-1)!}{2n^{p-1}} \Bigg(\sum_{\ell=1}^{K} \prod_{a=1}^{p} x_{i_a \ell} \Bigg)^2
+\frac{1}{2}\sqrt{\frac{(p-1)!}{(1-t) n^{p-1}}} Y_{i}  \sum_{\ell=1}^{K} \prod_{a=1}^{p} x_{i_a \ell}\\
+ \sum_{j=1}^{n} \frac{1}{2} x_j^T \frac{dR(t,\epsilon)}{dt} x_j - \big(\widetilde{Y}_j\big)^T  \frac{d\sqrt{R(t,\epsilon)}}{dt} x_j \,.
\end{multline}
Equation \eqref{derivative_hamiltonian} comes from differentiating the interpolating Hamiltonian \eqref{interpolating_hamiltonian}.
Before diving further, we remind two useful identities:
\begin{align}
 \frac{dR(t,\epsilon)}{dt} &= \sqrt{R(t,\epsilon)} \frac{d\sqrt{R(t,\epsilon)}}{dt} + \frac{d\sqrt{R(t,\epsilon)}}{dt} \sqrt{R(t,\epsilon)} \,;
 \label{eq:prop1_derivative_R}\\
 \forall v \in \mathbb{R}^K&: v^T \sqrt{R(t,\epsilon)} \frac{d\sqrt{R(t,\epsilon)}}{dt} v = v^T  \frac{d\sqrt{R(t,\epsilon)}}{dt} \sqrt{R(t,\epsilon)} v \,. \label{eq:prop2_derivative_R}
\end{align}
The identities \eqref{eq:prop1_derivative_R} and \eqref{eq:prop2_derivative_R} can further be combined to obtain
\begin{equation}
\forall v \in \mathbb{R}^K: v^T \sqrt{R(t,\epsilon)} \frac{d\sqrt{R(t,\epsilon)}}{dt} v
= \frac{1}{2} v^T \frac{dR(t,\epsilon)}{dt} v \;. \label{eq:prop3_derivative_R}
\end{equation}
Evaluating \eqref{derivative_hamiltonian} at $(\bx,\bY,\widetilde{\bY})=(\bX,\bY^{(t)},\widetilde{\bY}^{(t,\epsilon)})$ and then making use of \eqref{eq:prop3_derivative_R}, it comes:
\begin{align}\label{eq:H_timederivative}
\cH'_{t,\epsilon}(\bX ; \bY^{(t)},\widetilde{\bY}^{(t,\epsilon)})
&=
\sum_{i \in \mathcal{I}}
\frac{1}{2}\sqrt{\frac{(p-1)!}{(1-t) n^{p-1}}} Z_{i} \sum_{\ell=1}^{K} \prod_{a=1}^{p} X_{i_a \ell}\nonumber\\
&\qquad\qquad+\sum_{j=1}^{n} X_j^T \bigg(\frac{1}{2}\frac{dR(t,\epsilon)}{dt} - \sqrt{R(t,\epsilon)} \frac{d\sqrt{R(t,\epsilon)}}{dt} \bigg) X_j  - \widetilde{Z}_j^T  \frac{d\sqrt{R(t,\epsilon)}}{dt} X_j\nonumber\\
&= \sum_{i \in \mathcal{I}}
\frac{1}{2}\sqrt{\frac{(p-1)!}{(1-t) n^{p-1}}} Z_{i}  \sum_{\ell=1}^{K}\prod_{a=1}^{p} X_{i_a \ell}
\;-\;\sum_{j=1}^{n} \widetilde{Z}_j^T  \frac{d\sqrt{R(t,\epsilon)}}{dt} X_j \,.
\end{align}
Thanks to the Nishimori identity,
$$
T_2
= \E\,\big\langle \cH'_{t,\epsilon}\big(\bx; \bY^{(t)},\widetilde{\bY}^{(t,\epsilon)}\big)\big\rangle_{t,\epsilon}
= \E\,\cH'_{t,\epsilon}\big(\bX; \bY^{(t)},\widetilde{\bY}^{(t,\epsilon)}\big)\,.
$$
It follows that
\begin{equation*}
T_2
=\sum_{i \in \mathcal{I}}
\frac{1}{2}\sqrt{\frac{ (p-1)!}{(1-t) n^{p-1}}} \E[Z_{i}  ] \sum_{\ell=1}^{K} \E\Bigg[\prod_{a=1}^{p} X_{i_a \ell}\Bigg]
-\sum_{j=1}^{n} {\E[\widetilde{Z}_j]}^T \frac{d\sqrt{R(t,\epsilon)}}{dt} \E[X_j]
=0\;,
\end{equation*}
where we used $\E[Z_{i}  ] = \E[\widetilde{Z}_j] = 0$ to get the last equality.
Therefore, $f'_n(t,\epsilon) = -\nicefrac{T_1}{n}$. Plugging~\eqref{eq:H_timederivative} in the expression for $T_1$, we obtain:
\begin{multline}\label{eq:non_final_formula_derivative}
f'_{n}(t,\epsilon)
= -\frac{1}{2n}\sqrt{\frac{(p-1)!}{(1-t) n^{p-1}}}\sum_{i \in \mathcal{I}} \sum_{\ell=1}^{K}
\E\Bigg[Z_{i} \prod_{a=1}^{p} X_{i_a \ell} \ln \cZ_{t,\epsilon}(\bY^{(t)},\widetilde{\bY}^{(t,\epsilon)})\Bigg]\\
+\frac{1}{n} \sum_{j=1}^{n} \E\bigg[\widetilde{Z}_j^T  \frac{d\sqrt{R(t,\epsilon)}}{dt} X_j \ln \cZ_{t,\epsilon}(\bY^{(t)},\widetilde{\bY}^{(t,\epsilon)})\bigg]\,.
\end{multline}
The two kind of expectations appearing on the r.h.s.\ of~\eqref{eq:non_final_formula_derivative} are simplified in the paragraphs {\bf a)} and {\bf b)}.

\noindent{\bf a)} Integrating by parts with respect to the Gaussian random variable $Z_{i}$, we get:
\begin{multline*}
\E\Bigg[Z_i\prod_{a=1}^{p} X_{i_a \ell} \,\ln \cZ_{t,\epsilon}(\bY^{(t)},\widetilde{\bY}^{(t,\epsilon)})\Bigg]\\
= \E\Bigg[\prod_{a=1}^{p} X_{i_a \ell} \, \frac{\partial \ln \cZ_{t,\epsilon}(\widetilde{\bY}^{(t,\epsilon)})}{\partial Z_{i}} \Bigg]
= \sqrt{\frac{(1-t) (p-1)!}{n^{p-1}}} \sum_{\ell'=1}^{K}
\E\bigg[\bigg\langle \prod_{a=1}^{p} x_{i_a \ell'} X_{i_a \ell}\bigg\rangle_{\!\!\! t,\epsilon}\,\bigg]\;.
\end{multline*}
Summing the latter identity over $\ell \in \{1,\dots,K\}$ and $i \in \mathcal{I} = \{i \in [n]^p: i_a \leq i_{a+1}\}$, we obtain:
\begin{multline*}
-\frac{1}{2n}\sqrt{\frac{(p-1)!}{(1-t) n^{p-1}}}\sum_{i \in \mathcal{I}} \sum_{\ell=1}^{K}
\E\bigg[Z_{i} \prod_{a=1}^{p} X_{i_a \ell} \ln \cZ_{t,\epsilon}(\bY^{(t)},\widetilde{\bY}^{(t,\epsilon)})\bigg]\\
=-\frac{(p-1)!}{2n^{p}}\sum_{i \in \mathcal{I}} \sum_{\ell,\ell'=1}^{K}
\E\,\bigg\langle \prod_{a=1}^{p} x_{i_a \ell'} X_{i_a \ell}\bigg\rangle_{\!\!\! t,\epsilon} \;.
\end{multline*}
This last equality can be further simplified by replacing the sum over tuples $i \in [n]^p$ such that $i_1 < \dots < i_p$ by a sum over any $p$-tuple whose elements are distinct divided by $p!$ (the cardinality of the symmetric group of degree $p$).
This is possible because the summand is symmetric with respect to any permutation of the indices $(i_1,\dots,i_p)$.
We also need to account for the terms corresponding to $p$-tuples having common elements (that is, $i_a = i_{a'}$ for some $a \neq a'$).
There are $\mathcal{O}_{n}(n^{p-1})$ such terms and each summand is bounded under the assumption that $P_X$ has finite $(2p)$\textsuperscript{th} order moments.
Hence the term $\mathcal{O}_{n}(n^{-1})$ appearing in the final equalities:
\begin{multline}\label{eq:final_part_a}
-\frac{1}{2n}\sqrt{\frac{(p-1)!}{(1-t) n^{p-1}}}\sum_{i \in \mathcal{I}} \sum_{\ell=1}^{K}
\E\bigg[Z_{i} \prod_{a=1}^{p} X_{i_a \ell} \ln \cZ_{t,\epsilon}(\bY^{(t)},\widetilde{\bY}^{(t,\epsilon)})\bigg]\\
= \mathcal{O}_{n}(n^{-1}) -\frac{(p-1)!}{2n^{p}p!} \sum_{i \in [n]^p}^{n} \sum_{\ell,\ell'=1}^{K}
\E\,\bigg\langle \prod_{a=1}^{p} x_{i_a \ell'} X_{i_a \ell}\bigg\rangle_{\!\!\! t,\epsilon}
=\mathcal{O}_{n}(n^{-1}) -\frac{1}{2p} \sum_{\ell, \ell'=1}^{K} 
\E\,\big\langle (Q_{\ell \ell'})^{p} \big\rangle_{t,\epsilon} \,.
\end{multline}
{\bf b)} Now we look at the second expectation and integrate by parts with respect to the Gaussian random vector $\widetilde{Z}_{j}$:
\begin{align}
\E\bigg[ \widetilde{Z}_j^T  \frac{d\sqrt{R(t,\epsilon)}}{dt} X_j \ln \cZ_{t,\epsilon}(\bY^{(t)},\widetilde{\bY}^{(t,\epsilon)})\bigg]
&= \sum_{\ell=1}^{K} \E\bigg[ \bigg(\frac{d\sqrt{R(t,\epsilon)}}{dt} X_i\bigg)_{\!\! \ell} \frac{\partial \ln \cZ_{t,\epsilon}(\bY_t,\widetilde{\bY}_{t,\epsilon})}{\partial \widetilde{Z}_{j\ell}}\bigg]\nonumber\\
&= \sum_{\ell=1}^{K}
\E\bigg[ \bigg(\frac{d\sqrt{ R(t,\epsilon)}}{dt} X_j\bigg)_{\! \ell}
\Big\langle \big(\sqrt{R(t,\epsilon)}x_j\big)_{\ell} \Big\rangle_{\! t,\epsilon}\,\bigg]\nonumber\\
&=\E\bigg[ X_j^T \frac{d\sqrt{R(t,\epsilon)}}{dt}\sqrt{R(t,\epsilon)}
\,\big\langle x_j \big\rangle_{\! t,\epsilon}\,\bigg]\,. \label{eq:part_b_stein}
\end{align}
Equation \eqref{eq:part_b_stein} can be further simplified thanks to the Nishimory identity (for the first and last equalities) and the identity \eqref{eq:prop3_derivative_R} (for the second equality):
\begin{multline}
\E\bigg[ \widetilde{Z}_j^T  \frac{d\sqrt{R(t,\epsilon)}}{dt} X_j \ln \cZ_{t,\epsilon}(\bY^{(t)},\widetilde{\bY}^{(t,\epsilon)})\bigg]
=\E\bigg[ \big\langle x_j \big\rangle_{\! t,\epsilon}^T \frac{d\sqrt{R(t,\epsilon)}}{dt}\sqrt{R(t,\epsilon)}
\big\langle x_j \big\rangle_{\! t,\epsilon}\,\bigg]\\
=\frac{1}{2}\E\bigg[ \big\langle x_j \big\rangle_{\! t,\epsilon}^T \frac{d R(t,\epsilon)}{dt}
\big\langle x_j \big\rangle_{\! t,\epsilon}\,\bigg]
=\frac{1}{2}\E\bigg[ X_j^T \frac{d R(t,\epsilon)}{dt}
\big\langle x_j \big\rangle_{\! t,\epsilon}\,\bigg]\,.
\end{multline}
Summing the latter over $j \in \{1,\dots,n\}$, we obtain:
\begin{multline}\label{eq:final_part_b}
\frac{1}{n} \sum_{j=1}^n \E\bigg[ \widetilde{Z}_j^T  \frac{d\sqrt{R(t,\epsilon)}}{dt} X_j \ln \cZ_{t,\epsilon}(\bY_t,\widetilde{\bY}_{t,\epsilon})\bigg]\\
= \frac{1}{2n}\sum_{j=1}^n \E\bigg[ X_j^T \frac{d R(t,\epsilon)}{dt}
\big\langle x_j \big\rangle_{\! t,\epsilon}\,\bigg]
= \frac{1}{2}\E\,\bigg\langle\mathrm{Tr}\bigg(\frac{d R(t,\epsilon)}{dt} \frac{\bx^T\bX}{n} \bigg) \bigg\rangle_{\! t,\epsilon}
= \frac{1}{2}\mathrm{Tr}\big(R^{\prime}(t,\epsilon) \,\E\,\langle\bQ\rangle_{t,\epsilon}\big)\,.
\end{multline}
Summing the final expressions in \eqref{eq:final_part_a} and \eqref{eq:final_part_b} ends the proof of Proposition~\ref{prop:formula_derivative_free_entropy}.
\end{proof}
\section{Concentration of the overlap matrix}\label{sec:Barbier2020Overlap}
The proof of Theorem \ref{th:upperbound} requires that, up to an integral over a small volume of perturbations $\epsilon \in \mathcal{S}_K^{++}$, the overlap matrix $\bQ$ concentrates around its expectation $\E \langle \bQ \rangle_{t,\epsilon}$.
We chose to integrate the perturbation over the hypercube $\mathcal{E}_n \subseteq \mathcal{S}_K^{++}$ which is defined by \eqref{definition_En} and depends on a sequence $(s_n)_{n \in \mathbb{N}^*}$ of decreasing numbers in $(0,1)$.
Remember that, for all $\epsilon \in \mathcal{E}_n$, $R(\cdot,\epsilon): [0,1] \mapsto \mathcal{S}_K^{++}$ is the unique solution to $R'(t,\epsilon) = \E[\langle \bQ \rangle_{t,\epsilon}]^{\circ (p-1)}$ and, for all $t \in [0,1]$, $C(\mathcal{R}_{n,t})$ is the convex hull of the image $\mathcal{R}_{n,t} \coloneqq R(t,\mathcal{E}_n)$.
We also remind that, in Proposition \ref{prop:ode_and_properties}, we introduced the inference problem \eqref{inference_problem_t_R} whose associated posterior distribution reads
\begin{equation}\label{def:posterior_dist_t,R}
	dP(\bx \,\vert\,\bY^{(t)},\widetilde{\bY}^{(t,R)}) = \frac{1}{\cZ_{t,R}(\bY^{(t)},\widetilde{\bY}^{(t,R)})}
	\, e^{-\cH_{t,R}(\bx ; \bY^{(t)},\widetilde{\bY}^{(t,R)})}\, \prod_{j=1}^n dP_{X}(x_j) \quad;
\end{equation}
where $\forall \bx \in \mathbb{R}^{n \times K}$:
\begin{multline}
	\cH_{t,R}(\bx ; \bY, \widetilde{\bY})
	= \sum_{i \in \mathcal{I}}
	\frac{(1-t) (p-1)!}{2n^{p-1}} \Bigg(\sum_{k=1}^{K}\prod_{a=1}^p x_{i_a k}\Bigg)^{\! 2} -\sqrt{\frac{(1-t) (p-1)!}{n^{p-1}}} Y_{i}  \sum_{k=1}^{K} \prod_{a=1}^p x_{i_a k}\\
	+\sum_{j=1}^{n} \frac{1}{2} x_j^T R x_j- \widetilde{Y}_j^T  \sqrt{R} x_j \,.
\end{multline}
Let $\langle - \rangle_{t,R}=\int - \,dP(\bx \,\vert\,\bY^{(t)},\widetilde{\bY}^{(t,R)})$ be the Gibbs brackets associated to the posterior distribution \eqref{def:posterior_dist_t,R}.
Thanks to a change of variables (see the upper bound \eqref{upperbound_integral_En}), we showed that the following theorem is enough to prove Theorem \ref{th:upperbound}.
\begin{theorem}[Concentration of the overlap matrix around its expectation]\label{th:Barbier2020Overlap}
Assume $P_X$ has bounded support. There exists a positive constant $C$ depending only on $P_X$, $K$ and $p$ such that
\begin{equation}\label{eq:Barbier2020Overlap}
\int_{C(\mathcal{R}_{n,t})} dR\:\E\big[\big\langle \big\Vert \bQ-\E[\langle\bQ\rangle_{t,R}]\big\Vert^2\,\big\rangle_{t,R}\,\big] \leq \frac{C}{s_n^{\nicefrac{3}{2}} n^{\nicefrac{1}{6}}} \,.
\end{equation}
\end{theorem}
The proof of Theorem \ref{th:Barbier2020Overlap} relies on the one of \cite[Theorem 3]{Barbier2020Overlap}.
In the later reference, the concentration result is given for an integral over a hypercube $\mathcal{E}_n$.
In our case, the integral on the left-hand side of \eqref{eq:Barbier2020Overlap} is over the convex hull of $\mathcal{E}_n$'s image by the function $R(t,\cdot)$.
It is likely not a hypercube, even less one whose form is similar to $\eqref{definition_En}$.
Therefore, we first show that the convex hulls $\mathrm{C}(\mathcal{R}_{n,t})$ have properties allowing us to carry out a proof similar to \cite{Barbier2020Overlap}.
\subsection{Properties of \texorpdfstring{$\mathcal{R}_{n,t}$'s}{R\_\{n,t\}'s} convex hull}
For $(\ell,\ell') \in \{1,\dots,K\}^2$, we will denote $E^{(\ell,\ell')}$ the $K \times K$ symmetric matrix whose entries are:
\begin{equation}\label{def:E^(l,l')}
E_{ij}^{(\ell,\ell')} =
\begin{cases}
1 \text{ if } (i,j) \in \{(\ell,\ell'),(\ell',\ell)\} \,;\\
0 \text{ otherwise}.
\end{cases}
\end{equation}
\begin{lemma}[Properties of $\mathcal{R}_{n,t}$'s convex hull]\label{lemma:properties_convex_hulls}
For every $R \in \mathrm{C}(\mathcal{R}_{n,t})$:
\begin{enumerate}[label=(\roman*)]
	\item $\Vert R \Vert \leq 4K^{\nicefrac{3}{2}} + \mathrm{Tr}(\Sigma_{X})^{p-1}$;
	\item there exists $\epsilon \in \mathcal{E}_n$ such that $R \succcurlyeq \epsilon$;
	\item for every pair $(\ell,\ell') \in \{1,\dots,K\}^2$ and real number $\delta \in (-s_n,s_n)$, $R + \delta E^{(\ell,\ell')}$ is a symmetric positive definite matrix;
	\item the $1$\textsuperscript{st}-order Fr\'{e}chet derivative $\frac{\partial \sqrt{R}}{\partial R_{\ell\ell'}}$ and the $2$\textsuperscript{nd}-order Fr\'{e}chet derivative $\frac{\partial^2 \sqrt{R}}{\partial R_{\ell\ell'}^2}$ satisfy
	\begin{align}
	\bigg\Vert \frac{\partial \sqrt{R}}{\partial R_{\ell\ell'}} \bigg\Vert &\leq 	\frac{1}{\sqrt{2 s_n}}\,;\label{upperbound_norm_derivative_sqrtR}\\
	\bigg\Vert \frac{\partial^2 \sqrt{R}}{\partial R_{\ell\ell'}^2} \bigg\Vert &\leq \frac{\sqrt{K}}{(2s_n)^{\nicefrac{3}{2}}} \,.
	\label{upperbound_norm_2ndderivative_sqrtR}
	\end{align}
\end{enumerate}
\textbf{Remark:} Note that (i) does not depend on $n$ and $t$, while (ii-iv) do not depend on $t$.
\end{lemma}
\begin{proof}
We start by proving (i).
If $R \in \mathcal{R}_{n,t}$ then there exists $\epsilon \in \mathcal{E}_n$ such that $R=R(t,\epsilon)$, i.e.,
\begin{equation}\label{eq:R_in_Rn,t}
R = \epsilon + \int_0^t \E[\langle \bQ \rangle_{s,\epsilon}]^{\circ (p-1)} \, ds\,.
\end{equation}
Thus, $\Vert R \Vert
\leq \Vert \epsilon \Vert + \int_0^t \Vert \E[\langle \bQ \rangle_{s,\epsilon}]^{\circ (p-1)}\Vert \, ds
\leq 4K^{\nicefrac{3}{2}} + \int_0^t \Vert \E[\langle \bQ \rangle_{s,\epsilon}]\Vert^{p-1} \, ds$.
We have:
\begin{equation}\label{eq:upperbound_norm_ E<Q>}
\Vert \E[\langle \bQ \rangle_{s,\epsilon}]\Vert
\leq \frac{\E[ \langle \Vert \bx \Vert \rangle_{s,\epsilon} \Vert \bX \Vert]}{n} 
\leq \sqrt{\frac{\E[\Vert \bX \Vert^2]}{n} \frac{\E[ \langle \Vert\bx\Vert \rangle_{s,\epsilon}^2]}{n}}
= \frac{\E[\Vert \bX \Vert^2]}{n} 
= \mathrm{Tr}(\Sigma_{X})\,.
\end{equation}
The second inequality follows from Cauchy-Schwarz inequality and the first equality from the Nishimori identity.
Hence the upper bound
$\Vert R \Vert \leq 4K^{\nicefrac{3}{2}} + \mathrm{Tr}(\Sigma_{X})^{p-1}$
for all $R \in \mathcal{R}_{n,t}$, which directly extends to $C(\mathcal{R}_{n,t})$ by definition of a convex hull.

Now to prove (ii). If $R \in \mathcal{R}_{n,t}$, note that \eqref{eq:R_in_Rn,t} directly implies $R - \epsilon \succcurlyeq 0$ as -- by the Nishimori identity and the Schur Product theorem -- $\E[\langle \bQ \rangle_{s,\epsilon}]^{\circ (p-1)}$ is symmetric positive semidefinite for all $s \in [0,1]$. More generally, if $R \in C(\mathcal{R}_{n,t})$, there exist $m \in \mathbb{N}^*$, $(\alpha_1,\alpha_2,\dots,\alpha_m) \in [0,1]^m$ and $(R_1,\dots,R_m) \in (\mathcal{R}_{n,t})^M$ such that
$\sum_{j=1}^m \alpha_j = 1$ and $R = \sum_{j=1}^m \alpha_j R_j$.
It follows direcly that
$R \succcurlyeq \sum_{j=1}^m  \alpha_j \epsilon_j$
where $\forall j \in \{1,\dots,m\}: \mathcal{E}_n \ni \epsilon_j \preccurlyeq R_j$.
As $\mathcal{E}_n$ is convex, it concludes the proof of (ii).

We now show (ii) $\Rightarrow$ (iii).
Let $R \in \mathrm{C}(\mathcal{R}_{n,t})$ and pick $\epsilon \in \mathcal{E}_n$ such that $R \succcurlyeq \epsilon$.
For all $(\ell,\ell') \in \{1,\dots,K\}^2$ and $\delta \in (-s_n,s_n)$,
$\epsilon + \delta E^{(\ell,\ell')}$ is a symmetric strictly diagonally dominant matrix with positive diagonal entries.
Therefore, $\epsilon + \delta E^{(\ell,\ell')}$ belongs to $\mathcal{S}_K^{++}$ and $R + \delta E^{(\ell,\ell')} \succcurlyeq  \epsilon + \delta E^{(\ell,\ell')} \succ 0$.

Finally, we prove (iv).
Let $R \in \mathrm{C}(\mathcal{R}_{n,t})$ and denote $\lambda_{\min}(R)$ its minimum eigenvalue.
Applying \cite[Theorem 1.1]{Moral2018Taylor} (the first upper bound in (6) to be more precise), we obtain:
\begin{equation}\label{upperbound_Freichet_derivatives}
	\bigg\Vert \frac{\partial \sqrt{R}}{\partial R_{\ell\ell'}} \bigg\Vert \leq
	\frac{\Vert E^{(\ell,\ell')}\Vert}{2\sqrt{\lambda_{\min}(R)}}\:;\;
	\bigg\Vert \frac{\partial^2 \sqrt{R}}{\partial R_{\ell\ell'}^2} \bigg\Vert \leq
	\frac{\sqrt{K} \,\Vert E^{(\ell,\ell')}\Vert}{4\lambda_{\min}(R)^{\nicefrac{3}{2}}}\,. 
\end{equation}
Using (ii), pick $\epsilon \in \mathcal{E}_n$ such that $R \succcurlyeq \epsilon$.
By \cite[Corollary 2]{Varah1975Lower}, the minimum eigenvalue of $\epsilon$ is greater than $\sqrt{\alpha \beta}$ where
\begin{equation*}
\alpha = \min_{1 \leq k \leq K} \bigg\{\vert \epsilon_{kk} \vert - \sum_{j \neq k} \vert \epsilon_{kj} \vert \bigg\} \geq s_n\
\quad\text{and}\quad
\beta = \min_{1 \leq k \leq K} \bigg\{\vert \epsilon_{kk} \vert - \sum_{j \neq k} \vert \epsilon_{jk} \vert \bigg\} \geq s_n \,.
\end{equation*}
Hence $\lambda_{\min}(R) \geq \sqrt{\alpha \beta} \geq s_n$. 
Combining this lower bound with \eqref{upperbound_Freichet_derivatives} ends the proof of (iv).
\end{proof}
%
%
\subsection{Concentration of \texorpdfstring{$\pmb{\cL}$}{L} around its expectation}
As in \cite{Barbier2020Overlap}, the concentration of the overlap matrix around its expectation will follow from the concentration of the $K \times K$ symmetric matrix $\pmb{\cL} \equiv \pmb{\cL}(R)$ whose entries are:
\begin{equation}\label{def_L}
\forall (\ell,\ell') \in \{1,\dots,K\}^2: \cL_{\ell\ell'}
\coloneqq \frac{1}{n} \sum_{j=1}^n \frac{1}{2}x_j^T \frac{\partial R}{\partial R_{\ell\ell'}}x_j
-X_j^T \frac{\partial R}{\partial R_{\ell\ell'}}x_j
-x_j^T \frac{\partial \sqrt{R}}{\partial R_{\ell\ell'}} \widetilde{Z}_j \,.
\end{equation}
This is well-defined as long as $R \in \mathcal{S}_K^{++}$.
To prove concentration results on $\pmb{\cL}$, it will be useful to work with the free entropy $\frac{1}{n} \ln \cZ_{t,R}(\bY^{(t)},\widetilde{\bY}^{(t,R)})$ where $\cZ_{t,R}(\bY^{(t)},\widetilde{\bY}^{(t,R)})$ is the normalization factor of the posterior distribution \eqref{posterior_dist_t,R}.
In Appendix \ref{app:concentration_free_entropy}, we prove that this free entropy concentrates around its expectation when $n \to +\infty$. 
In order to shorten notations, we define:
\begin{equation}
	F_n(t,R) = \frac{1}{n} \ln \cZ_{t,R}\big(\bY^{(t)},\widetilde{\bY}^{(t,R)}\big)\;;\quad
	f_n(t,R) = \frac{1}{n} \E\big[\ln \cZ_{t,R}\big(\bY^{(t)},\widetilde{\bY}^{(t,R)}\big)\big] = \E\,F_n(t,R) \,.
\end{equation}
\begin{proposition}[Thermal fluctuations of $\pmb{\cL}$]\label{prop:concentration_L_on_<L>}
Assume $P_X$ has finite fourth-order moments. There exists a positive constant $C$, depending only on $\Sigma_X$, $K$ and $p$, such that for all $(n,t) \in \mathbb{N}^{*} \times [0,1]$:
\begin{equation}
\int_{C(\mathcal{R}_{n,t})} dR\:\E\big[\big\langle \big\Vert \pmb{\cL}-\langle\pmb{\cL}\rangle_{t,R}\big\Vert^2\,\big\rangle_{t,R}\,\big]
\leq \frac{C}{n s_n} \,.
\end{equation}
\end{proposition}
\begin{proof}
Fix $(n,t) \in \mathbb{N}^* \times [0,1]$.
Note that $\forall R \in \mathcal{S}_K^{++}$, $\forall (\ell,\ell') \in \{1,\dots,K\}^2$:
\begin{equation}
\frac{\partial f_n}{\partial R_{\ell\ell'}}\bigg\vert_{t,R}
= -\frac{1}{n}\E\Bigg[\Bigg\langle \frac{\partial \cH_{t,R}(\bx;\bY^{(t)},\widetilde{\bY}^{(t,R)})}{\partial R_{\ell\ell'}} \Bigg\rangle_{\!\! t,R}\,\Bigg]
=-\E\big[\big\langle \cL_{\ell \ell'}\big\rangle_{t,R}\,\big] \,.
\end{equation}
Further differentiating, we obtain:
\begin{align}
\frac{\partial^2 f_n}{\partial R_{\ell\ell'}^2}\bigg\vert_{t,R}
&= \E\bigg[\bigg\langle \cL_{\ell\ell'}\frac{\partial \cH_{t,R}}{\partial R_{\ell\ell'}} \bigg\rangle_{\!\! t,R}\,\bigg]
-\E\bigg[\big\langle \cL_{\ell \ell'}\big\rangle_{t,R}\,
\bigg\langle \frac{\partial \cH_{t,R}}{\partial R_{\ell\ell'}} \bigg\rangle_{\!\! t,R}\,\bigg]
-\E\,\bigg\langle \frac{\partial \cL_{\ell\ell'}}{\partial R_{\ell\ell'}} \bigg\rangle_{\!\! t,R}\\
&=n\E\,\big\langle \big(\cL_{\ell \ell'} -\big\langle \cL_{\ell \ell'}\big\rangle_{t,R}\big)^{2}\,\big\rangle_{t,R}
-\E\,\bigg\langle \frac{\partial \cL_{\ell\ell'}}{\partial R_{\ell\ell'}} \bigg\rangle_{\!\! t,R}\quad.\label{second_derivative_f_n_R}
\end{align}
Combining \eqref{second_derivative_f_n_R} and \eqref{formula_E<dL/dR>} for $\E\,\langle \nicefrac{\partial \cL_{\ell\ell'}}{\partial R_{\ell\ell'}} \rangle_{t,R}$ (see Lemma \ref{lemma:computation_E<L>_and_others} following this proof), it comes:
\begin{equation}\label{thermal_fluctuation_Lll'}
\E\,\big\langle \big(\cL_{\ell \ell'} -\big\langle \cL_{\ell \ell'}\big\rangle_{t,R}\big)^{2}\,\big\rangle_{t,R}
= \frac{1}{n}\frac{\partial^2 f_n}{\partial R_{\ell\ell'}^2}\bigg\vert_{t,R}
+\frac{1}{n}\mathrm{Tr}\bigg(\frac{\partial \sqrt{R}}{\partial R_{\ell\ell'}} 
\Big(\Sigma_X - \E\,\langle \bQ \rangle_{t,R}\Big) \frac{\partial \sqrt{R}}{\partial R_{\ell\ell'}} \bigg) \,.
\end{equation}
We start with upper bounding the integral over $C(\mathcal{R}_{n,t})$ of the second summand on the right-hand side of \eqref{thermal_fluctuation_Lll'}.
Thanks to the Nishimory identity, we can see that $\Sigma_X \succcurlyeq \E\,\langle \bQ \rangle_{t,R}$.
Indeed:
\begin{multline}
\Sigma_X - \E\,\langle \bQ \big\rangle_{t,R}
= \frac{1}{n} \E[\bX^T\bX] - \frac{1}{n} \E[\langle \bx \rangle_{t,R}^T \bX ]
= \frac{1}{n} \E[\langle \bx^T\bx \rangle_{t,R}] - \frac{1}{n} \E[\langle \bx \rangle_{t,R}^T\langle \bx \rangle_{t,R}]\\
= \frac{1}{n} \E[\langle (\bx-\langle \bx \rangle_{t,R})^T (\bx-\langle \bx \rangle_{t,R})\rangle_{t,R}] \succcurlyeq 0\,.
\end{multline}
Therefore,
$(\nicefrac{\partial \sqrt{R}}{\partial R_{\ell\ell'}})(\Sigma_X - \E\,\langle \bQ \rangle_{t,R})(\nicefrac{\partial \sqrt{R}}{\partial R_{\ell\ell'}})$
is symmetric positive semidefinite and the second term on the right-hand side of \eqref{thermal_fluctuation_Lll'} satisfies:
\begin{equation}
0 \leq \frac{1}{n}\mathrm{Tr}\bigg(\frac{\partial \sqrt{R}}{\partial R_{\ell\ell'}} 
\Big(\Sigma_X - \E\,\langle \bQ \rangle_{t,R}\Big) \frac{\partial \sqrt{R}}{\partial R_{\ell\ell'}} \bigg)
\leq \frac{1}{n}\mathrm{Tr}\bigg(\frac{\partial \sqrt{R}}{\partial R_{\ell\ell'}} 
\Sigma_X \frac{\partial \sqrt{R}}{\partial R_{\ell\ell'}} \bigg)
\leq \bigg\Vert \frac{\partial \sqrt{R}}{\partial R_{\ell\ell'}} \bigg\Vert^2
\frac{\Vert \Sigma_X \Vert}{n}
\leq \frac{\Vert \Sigma_X \Vert}{2n s_n}\:.
\end{equation}
The last inequality follows from the upper bound \eqref{upperbound_norm_derivative_sqrtR} in Lemma \ref{lemma:properties_convex_hulls}.
Therefore, keeping in mind that $C(\mathcal{R}_{n,t})$ is included in the ball $\mathcal{B}(\Sigma_{X},K,p)$, there exists a positive constant $C_1$ depending only on $\Sigma_{X}$, $K$ and $p$ such that:
\begin{equation}\label{upperbound_int_derivative_Lll'_Rll'}
\int_{C(\mathcal{R}_{n,t})} \frac{dR}{n}\,\mathrm{Tr}\bigg(\frac{\partial \sqrt{R}}{\partial R_{\ell\ell'}} 
\Big(\Sigma_X - \E\,\langle \bQ \rangle_{t,R}\Big) \frac{\partial \sqrt{R}}{\partial R_{\ell\ell'}} \bigg)
\leq \frac{C_1}{n s_n} \,.
\end{equation}
Now we turn to upper bounding $\int_{C(\mathcal{R}_{n,t})} \frac{dR}{n}\frac{\partial^2 f_n}{\partial R_{\ell\ell'}^2}\Big\vert_{t,R}$.
Define the closed convex set
\begin{equation}\label{def_Cll'}
C^{(\ell,\ell')} = \big\{ \{R_{kk'}\}_{(k,k')\neq (\ell,\ell'),(\ell',\ell)} \big\vert R \in C(\mathcal{R}_{n,t}) \big\} \,.
\end{equation}
For every pair $(\widetilde{R},r) \in C^{(\ell,\ell')}  \times \mathbb{R}$, we denote $\widetilde{R} \cup \{r\}$ the symmetric matrix whose entries are given by:
\begin{equation}\label{def_tildeRcupr}
\big(\widetilde{R} \cup \{r\}\big)_{k k'} =
\begin{cases}
\widetilde{R}_{kk'} \quad\text{if } (k,k') \neq (\ell,\ell'), (\ell',\ell) \,;\\
\;\;\, r \quad\;\;\:\text{otherwise}. 
\end{cases}
\end{equation}
Because $C(\mathcal{R}_{n,t})$ is a closed convex,
there exist two functions $a,b: C^{(\ell,\ell')} \to \mathbb{R}$ such that $\forall \widetilde{R} \in C^{(\ell,\ell')}$:
\begin{enumerate}[label=(\roman*)]
	\item $a(\widetilde{R}) \leq b(\widetilde{R})$;
	\item $\forall r \in [a(\widetilde{R}),b(\widetilde{R})]: \widetilde{R} \cup \{r\} \in C(\mathcal{R}_{n,t})$;
	\item $\forall r \in \mathbb{R}\setminus[a(\widetilde{R}),b(\widetilde{R})]: \widetilde{R} \cup \{r\} \notin C(\mathcal{R}_{n,t})$.
\end{enumerate}
Therefore,
\begin{multline}\label{upperbound_int_secondDerivative_fn_Rll'}
\int_{C(\mathcal{R}_{n,t})} \frac{dR}{n}\frac{\partial^2 f_n}{\partial R_{\ell\ell'}^2}\Big\vert_{t,R}
= \int_{C^{(\ell,\ell')}} \frac{d\widetilde{R}}{n}
\int_{a(\widetilde{R})}^{b(\widetilde{R})} dr\,\frac{\partial^2 f_n}{\partial R_{\ell\ell'}^2}\bigg\vert_{t,\widetilde{R} \cup \{r\}}\\
= \int_{C^{(\ell,\ell')}} \frac{d\widetilde{R}}{n}
\Bigg(\frac{\partial f_n}{\partial R_{\ell\ell'}}\bigg\vert_{t,\widetilde{R} \cup \{b(\widetilde{R})\}} - \frac{\partial f_n}{\partial R_{\ell\ell'}}\bigg\vert_{t,\widetilde{R} \cup \{a(\widetilde{R})\}}\Bigg)\,.
\end{multline}
Note that $\forall R \in \mathcal{S}_K^{++}$:
\begin{equation}\label{upperbound_derivative_fn_Rll'}
\Bigg\vert \frac{\partial f_n}{\partial R_{\ell\ell'}}\bigg\vert_{t,R} \Bigg\vert
= \big\vert \E\,\langle \cL_{\ell \ell'}\rangle_{t,R} \big\vert
\leq \big\vert \E\,\langle Q_{\ell \ell'}\rangle_{t,R}  \big\vert
\leq \mathrm{Tr}\,\Sigma_{X} \,,
\end{equation}
where the second and third inequalities follow from the identity \eqref{formula_E<L>} (see Lemma \ref{lemma:computation_E<L>_and_others} following this proof) and the inequality \eqref{eq:upperbound_norm_ E<Q>}, respectively.
Combining both \eqref{upperbound_int_secondDerivative_fn_Rll'} and \eqref{upperbound_derivative_fn_Rll'}, we finally get
\begin{equation}\label{final_upperbound_int_secondDerivative_fn_Rll'}
\Bigg\vert \int_{C(\mathcal{R}_{n,t})} \frac{dR}{n}\frac{\partial^2 f_n}{\partial R_{\ell\ell'}^2}\bigg\vert_{t,R} \Bigg\vert
\leq \frac{C_2}{n}\;,
\end{equation}
where $C_2$ is a positive constant that depends only on $\Sigma_{X}$, $K$ and $p$.
Integrating \eqref{thermal_fluctuation_Lll'} over $C(\mathcal{R}_{n,t})$,
making use of the upper bounds \eqref{upperbound_int_derivative_Lll'_Rll'} and \eqref{final_upperbound_int_secondDerivative_fn_Rll'} and, finally, summing over $(\ell,\ell') \in \{1,\dots,K\}^2$ end the proof.
\end{proof}
We relied on the following lemma for the proof of Proposition~\ref{prop:concentration_L_on_<L>}.
\begin{lemma}\label{lemma:computation_E<L>_and_others}
Assume $P_X$ has finite second-order moments.
Let $\delta_{\ell\ell'} = 0$ if $\ell \neq \ell'$ and $\delta_{\ell\ell'}=1$ otherwise. 
Then, $\forall (t,R) \in [0,1] \times \mathcal{S}_K^{++}$, $\forall (\ell,\ell') \in \{1,\dots,K\}^2$:
\begin{align}
\E\,\langle \cL_{\ell\ell'} \rangle_{t,R} &= -(1 - \nicefrac{\delta_{\ell\ell'}}{2})\E\,\langle Q_{\ell\ell'}\rangle_{t,R}\;;\label{formula_E<L>}\\
\E\,\bigg\langle \frac{\partial \cL_{\ell\ell'}}{\partial R_{\ell\ell'}} \bigg\rangle_{\!\! t,R}
&= \mathrm{Tr}\bigg(\frac{\partial \sqrt{R}}{\partial R_{\ell\ell'}} 
\Big(\Sigma_X - \E\,\langle \bQ \rangle_{t,R}\Big) \frac{\partial \sqrt{R}}{\partial R_{\ell\ell'}} \bigg)\;.\label{formula_E<dL/dR>}
\end{align}
\end{lemma}
\begin{proof}
Fix $(t,R) \in [0,1] \times \mathcal{S}_K^{++}$.
By the definition \eqref{def_L} of $\pmb{\cL}$, we have $\forall (\ell,\ell') \in \{1,\dots,K\}^2$:
\begin{equation}\label{def_E<L>}
\E\,\langle \cL_{\ell\ell'} \rangle_{t,R}
= \frac{1}{n} \sum_{j=1}^n \frac{1}{2} \E\bigg[\bigg\langle x_j^T \frac{\partial R}{\partial R_{\ell\ell'}}x_j \bigg\rangle_{\!\! t,R}\,\bigg]
- \E\bigg[X_j^T \frac{\partial R}{\partial R_{\ell\ell'}}\big\langle x_j \big\rangle_{t,R}\bigg]
-\E\bigg[ \big\langle x_j \big\rangle_{t,R}^T \frac{\partial \sqrt{R}}{\partial R_{\ell\ell'}} \widetilde{Z}_j \bigg]\,.
\end{equation}
Integrating by parts with respect to the Gaussian random vectors $\widetilde{Z}_j$, $j \in [n]$, the last expectation on the right-hand side of \eqref{def_E<L>} reads:
\begin{align}
\E\bigg[ \big\langle x_j \big\rangle_{t,R}^T \frac{\partial \sqrt{R}}{\partial R_{\ell\ell'}} \widetilde{Z}_j \bigg]
&= \E\,\bigg\langle x_j^T \frac{\partial \sqrt{R}}{\partial R_{\ell\ell'}} \sqrt{R} \,x_j
\bigg\rangle_{\!\! t,R}
-\;\E\bigg[ \big\langle x_j \big\rangle_{t,R}^T \frac{\partial \sqrt{R}}{\partial R_{\ell\ell'}} \sqrt{R} \,\big\langle x_j \big\rangle_{t,R} \bigg] \nonumber\\
&= \frac{1}{2}\E\,\bigg\langle x_j^T \frac{\partial R}{\partial R_{\ell\ell'}} x_j
\bigg\rangle_{\!\! t,R}
-\;\frac{1}{2}\E\bigg[ \big\langle x_j \big\rangle_{t,R}^T \frac{\partial R}{\partial R_{\ell\ell'}} \big\langle x_j \big\rangle_{t,R} \bigg] \nonumber\\
&= \frac{1}{2}\E\,\bigg\langle x_j^T \frac{\partial R}{\partial R_{\ell\ell'}} x_j
\bigg\rangle_{\!\! t,R}
-\;\frac{1}{2}\E\bigg[ X_j^T \frac{\partial R}{\partial R_{\ell\ell'}} \big\langle x_j \big\rangle_{t,R} \bigg] \;\;.\label{stein_lemma_last_term_E<L>}
\end{align}
The second and third equalities follow from \eqref{eq:identity_partial_diff_Rll'} and the Nishimori identity, respectively.
Plugging \eqref{stein_lemma_last_term_E<L>} in \eqref{def_E<L>} and, then, making use of the identity $\nicefrac{\partial R}{\partial R_{\ell\ell'}} = E^{(\ell,\ell')}$ end the proof of \eqref{formula_E<L>}:
\begin{equation*}
\E\,\langle \cL_{\ell\ell'} \rangle_{t,R}
= -\frac{1}{2n} \sum_{j=1}^n \E\bigg[X_j^T \frac{\partial R}{\partial R_{\ell\ell'}}\big\langle x_j \big\rangle_{t,R}\bigg]
= -\frac{1}{2} \mathrm{Tr}\bigg(\frac{\partial R}{\partial R_{\ell\ell'}} \E\,\langle \bQ \rangle_{t,R}\bigg)
= -(1 - \nicefrac{\delta_{\ell\ell'}}{2})\E\,\langle Q_{\ell\ell'}\rangle_{t,R} \;\;.
\end{equation*}
We now turn to the proof of \eqref{formula_E<dL/dR>}. We have:
\begin{align}\label{derivative_L_Rll'}
\E\,\bigg\langle \frac{\partial \cL_{\ell\ell'}}{\partial R_{\ell\ell'}} \bigg\rangle_{\!\! t,R}
&= -\frac{1}{n}\sum_{j=1}^{n} \E\bigg[\big\langle x_j\big\rangle_{t,R}^T \frac{\partial^2 \sqrt{R}}{\partial R_{\ell\ell'}^2} \widetilde{Z}_j\bigg]\nonumber\\
&= \frac{1}{n}\sum_{j=1}^{n} \E\bigg[\big\langle x_j \big\rangle_{t,R}^T \frac{\partial^2 \sqrt{R}}{\partial R_{\ell\ell'}^2} \sqrt{R} \,\big\langle x_j \big\rangle_{t,R}\bigg]
	-\E\,\bigg\langle x_j^T \frac{\partial^2 \sqrt{R}}{\partial R_{\ell\ell'}^2} \sqrt{R} \,x_j \bigg\rangle_{\!\! t,R}\quad.
\end{align}
The second equality follows once again from a Gaussian integration by parts with respect to $\widetilde{Z}_j$, $j \in [n]$.
Note that for all $v \in \mathbb{R}^K$:
\begin{equation}\label{scalar_prod_derivative_sqrtR_Rll'}
v^T \frac{\partial^2 \sqrt{R}}{\partial R_{\ell\ell'}^2} \sqrt{R} v
= \frac{1}{2} v^T\bigg(\sqrt{R}\frac{\partial^2 \sqrt{R}}{\partial R_{\ell\ell'}^2} + \frac{\partial^2 \sqrt{R}}{\partial R_{\ell\ell'}^2}\sqrt{R}\bigg) v
= -v^T \bigg(\frac{\partial \sqrt{R}}{\partial R_{\ell\ell'}}\bigg)^{\! 2} v
\end{equation}
because of the identity
\begin{equation*}
0 
= \frac{\partial^2 R}{\partial R_{\ell\ell'}^2} 
=\frac{\partial}{\partial R_{\ell\ell'}}\bigg(\sqrt{R}\frac{\partial \sqrt{R}}{\partial R_{\ell\ell'}} + \frac{\partial \sqrt{R}}{\partial R_{\ell\ell'}}\sqrt{R}\bigg)
= 2\bigg(\frac{\partial \sqrt{R}}{\partial R_{\ell\ell'}}\bigg)^2
+ \sqrt{R}\frac{\partial^2 \sqrt{R}}{\partial R_{\ell\ell'}^2} + \frac{\partial^2 \sqrt{R}}{\partial R_{\ell\ell'}^2}\sqrt{R} \;\;.
\end{equation*}
Plugging \eqref{scalar_prod_derivative_sqrtR_Rll'} in \eqref{derivative_L_Rll'},
$\E\,\langle \nicefrac{\partial \cL_{\ell\ell'}}{\partial R_{\ell\ell'}} \rangle_{t,R}$ further simplifies:
\begin{align}\label{derivative_L_Rll'_final}
	\E\,\bigg\langle \frac{\partial \cL_{\ell\ell'}}{\partial R_{\ell\ell'}} \bigg\rangle_{\!\! t,R}
	&= \frac{1}{n}\sum_{j=1}^{n} \E\,\bigg\langle x_j^T \bigg(\frac{\partial \sqrt{R}}{\partial R_{\ell\ell'}}\bigg)^{\! 2} x_j \bigg\rangle_{\!\! t,R}
	-\E\bigg[\big\langle x_j \big\rangle_{t,R}^T \bigg(\frac{\partial \sqrt{R}}{\partial R_{\ell\ell'}}\bigg)^{\! 2} \big\langle x_j \big\rangle_{t,R}\bigg]\nonumber\\
	&= \frac{1}{n}\sum_{j=1}^{n} \E\bigg[ X_j^T \bigg(\frac{\partial \sqrt{R}}{\partial R_{\ell\ell'}}\bigg)^{\! 2} X_j \bigg]
	-\E\bigg[ X_j^T \bigg(\frac{\partial \sqrt{R}}{\partial R_{\ell\ell'}}\bigg)^{\! 2} \big\langle x_j \big\rangle_{t,R}\bigg]\nonumber\\
	&= \mathrm{Tr}\bigg(\bigg(\frac{\partial \sqrt{R}}{\partial R_{\ell\ell'}}\bigg)^{\! 2}\,\Sigma_X\bigg)
	- \mathrm{Tr}\bigg(\bigg(\frac{\partial \sqrt{R}}{\partial R_{\ell\ell'}}\bigg)^{\! 2} \,\E\,\langle \bQ \rangle_{t,R}\bigg)\nonumber\\
	&= \mathrm{Tr}\bigg(\frac{\partial \sqrt{R}}{\partial R_{\ell\ell'}} 
	\Big(\Sigma_X - \E\,\langle \bQ \rangle_{t,R}\Big) \frac{\partial \sqrt{R}}{\partial R_{\ell\ell'}} \bigg)\;\;.
\end{align}
The second equality follows from the Nishimori identity.
\end{proof}
\begin{proposition}[Quenched fluctuations of $\pmb{\cL}$]\label{prop:concentration_<L>_on_E<L>}
	Assume $P_X$ has bounded support.
	There exists a positive constant $C$, depending only on $P_X$, $K$ and $p$, such that for all $(n,t) \in \mathbb{N}^{*} \times [0,1]$:
	\begin{equation}
	\int_{C(\mathcal{R}_{n,t})} dR\:\E\big[\big\langle \big\Vert \langle\pmb{\cL}\rangle_{t,R}-\E\,\langle\pmb{\cL}\rangle_{t,R}\,\big\Vert^2\,\big\rangle_{t,R}\,\big]
	\leq \frac{C}{s_n^3 n^{\nicefrac{1}{3}}}\,.
	\end{equation}
\end{proposition}
\begin{proof}
Fix $(n,t) \in \mathbb{N}^* \times [0,1]$. For all $R \in \mathcal{S}_K^{++}$ and $(\ell,\ell') \in \{1,\dots,K\}^2$, we have:
\begin{align}
\frac{\partial F_n}{\partial R_{\ell\ell'}}\bigg\vert_{t,R}
&=-\big\langle \cL_{\ell \ell'}\big\rangle_{t,R} \;;\\
\frac{\partial^2 F_n}{\partial R_{\ell\ell'}^2}\bigg\vert_{t,R}
&=n \big\langle \big(\cL_{\ell \ell'} -\big\langle \cL_{\ell \ell'}\big\rangle_{t,R}\big)^{2}\,\big\rangle_{t,R}
+\frac{1}{n}\sum_{j=1}^{n}\big\langle x_j\big\rangle_{t,R}^T \frac{\partial^2 \sqrt{R}}{\partial R_{\ell\ell'}^2} \widetilde{Z}_j \;;\label{quenched:2ndDeriv_Fn}\\
\frac{\partial f_n}{\partial R_{\ell\ell'}}\bigg\vert_{t,R}
&=-\E\,\big\langle \cL_{\ell \ell'}\big\rangle_{t,R}\;;\\
\frac{\partial^2 f_n}{\partial R_{\ell\ell'}^2}\bigg\vert_{t,R}
&=n \E\,\big\langle \big(\cL_{\ell \ell'} - \langle \cL_{\ell \ell'}\rangle_{t,R}\big)^{2}\,\big\rangle_{t,R}
+\frac{1}{n}\sum_{j=1}^{n}\E\bigg[\big\langle x_j\big\rangle_{t,R}^T \frac{\partial^2 \sqrt{R}}{\partial R_{\ell\ell'}^2} \widetilde{Z}_j\bigg]\,.
\end{align}
By assumption there exists a nonnegative real number $B_X$ such that $X \sim P_X \Rightarrow \Vert X \Vert \leq B_X$ almost surely.
Using the upper bound \eqref{upperbound_norm_2ndderivative_sqrtR} in Lemma \ref{lemma:properties_convex_hulls},
the second term on the right-hand side of \eqref{quenched:2ndDeriv_Fn} can be upper bounded:
\begin{equation}\label{upperbound_2ndterm_2ndDeriv_Fn}
\Bigg\vert\frac{1}{n}\sum_{j=1}^{n}\big\langle x_j\big\rangle_{t,R}^T \frac{\partial^2 \sqrt{R}}{\partial R_{\ell\ell'}^2} \widetilde{Z}_j\Bigg\vert
\leq \frac{1}{n}\sum_{j=1}^{n}\big\Vert\big\langle x_j\big\rangle_{t,R}\big\Vert
\big\Vert\widetilde{Z}_j\big\Vert
\bigg\Vert \frac{\partial^2 \sqrt{R}}{\partial R_{\ell\ell'}^2} \bigg\Vert
\leq \frac{B_X\sqrt{K}}{(2s_n)^{\nicefrac{3}{2}}n}\sum_{j=1}^{n}
\big\Vert\widetilde{Z}_j\big\Vert \,.
\end{equation}
From now on, we also fix $(\ell,\ell') \in \{1,\dots,K\}^2$ as well as $\widetilde{R} \in C^{(\ell,\ell')}$. The closed convex set $C^{(\ell,\ell')}$ is as defined by \eqref{def_Cll'}.
Remember that, for every real number $r$, $\widetilde{R}\cup\{r\}$ is the matrix defined by \eqref{def_tildeRcupr}, and that there exist two functions $a,b: C^{(\ell,\ell')} \to \mathbb{R}$ such that $\forall \widetilde{R} \in C^{(\ell,\ell')}$:
\begin{enumerate}[label=(\roman*)]
	\item $a(\widetilde{R}) \leq b(\widetilde{R})$;
	\item $\forall r \in [a(\widetilde{R}),b(\widetilde{R})]: \widetilde{R} \cup \{r\} \in C(\mathcal{R}_{n,t})$;
	\item $\forall r \in \mathbb{R}\setminus[a(\widetilde{R}),b(\widetilde{R})]: \widetilde{R} \cup \{r\} \notin C(\mathcal{R}_{n,t})$.
\end{enumerate}
Besides, by property (iii) in Lemma \ref{lemma:properties_convex_hulls}, for every $r \in (a(\widetilde{R})-s_n,b(\widetilde{R})+s_n)$ the matrix  $\widetilde{R} \cup \{r\}$ is in $\mathcal{S}_K^{++}$.
Thus, we can define for all $r \in (a(\widetilde{R})-s_n,b(\widetilde{R})+s_n)$:
\begin{align}
F(r) &= F_n(t,\widetilde{R}\cup\{r\}) + \frac{r^2}{2}\frac{B_X\sqrt{K}}{(2s_n)^{\nicefrac{3}{2}}n}\sum_{j=1}^{n}
\big\Vert\widetilde{Z}_j\big\Vert \:;\\
f(r) &= f_n(t,\widetilde{R}\cup\{r\}) + \frac{r^2}{2}\frac{B_X\sqrt{K}}{(2s_n)^{\nicefrac{3}{2}}n}\sum_{j=1}^{n}
\E\,\big\Vert\widetilde{Z}_j\big\Vert \,.
\end{align}
$F$ is convex on $(a(\widetilde{R})-s_n,b(\widetilde{R})+s_n)$ as it is twice differentiable with a nonnegative second derivative by \eqref{quenched:2ndDeriv_Fn} and \eqref{upperbound_2ndterm_2ndDeriv_Fn}.
The same holds for $f$.
We will apply the following standard to these two convex functions (see \cite{Barbier2019Adaptivea} for a proof):
\begin{lemma}[An upper bound for differentiable convex functions]\label{lemma:diff_convex_functions}
Let $g$ and $G$ be two differentiable convex functions defined on an interval $I\subseteq \mathbb{R}$.
Let $r \in I$ and $\delta > 0$ such that $r \pm \delta \in I$. Then
\begin{equation}
\vert G'(r) - g'(r) \vert
\leq C_{\delta}(r) + \frac{1}{\delta}\sum_{u \in \{-\delta,0,\delta\}} \vert G(r+u) - g(r+u) \vert \,,
\end{equation}
where $C_{\delta}(r) = g'(r+\delta) - g'(r-\delta) \geq 0$.
\end{lemma}
\noindent For all $r \in (a(\widetilde{R})-s_n,b(\widetilde{R})+s_n)$, we have:
\begin{align}
F(r) - f(r)
= F_n(t,\widetilde{R}\cup\{r\}) - f_n(t,\widetilde{R}\cup\{r\})
+ \frac{r^2}{2}\frac{B_X\sqrt{K}}{(2s_n)^{\nicefrac{3}{2}}n}\sum_{j=1}^{n}
\big\Vert\widetilde{Z}_j\big\Vert- \E\,\big\Vert\widetilde{Z}_j\big\Vert \:;\label{F_minus_f}\\
F'(r) - f'(r)
= -\Big(\langle \cL_{\ell \ell'}\rangle_{t,\widetilde{R}\cup\{r\}} - \E\,\langle \cL_{\ell \ell'}\rangle_{t,\widetilde{R}\cup\{r\}}\Big)
+ r\frac{B_X\sqrt{K}}{(2s_n)^{\nicefrac{3}{2}}n}\sum_{j=1}^{n}
\big\Vert\widetilde{Z}_j\big\Vert- \E\,\big\Vert\widetilde{Z}_j\big\Vert \;. \label{F'_minus_f'}
\end{align}
Let $C_{\delta}(r) = f'(r+\delta) - f'(r-\delta)$, which is nonnegative by convexity of $f$.
It follows from Lemma \ref{lemma:diff_convex_functions} and the two identities \eqref{F_minus_f} and \eqref{F'_minus_f'} that $\forall (r,\delta) \in [a(\widetilde{R}),b(\widetilde{R})] \times (0,s_n)$:
\begin{align*}
&\big\vert \langle \cL_{\ell \ell'}\rangle_{t,\widetilde{R}\cup\{r\}} - \E\,\langle\cL_{\ell \ell'}\rangle_{t,\widetilde{R}\cup\{r\}}\big\vert\\
&\qquad\qquad\qquad\qquad\leq \vert r \vert \frac{B_X\sqrt{K}}{(2s_n)^{\nicefrac{3}{2}}n}
\Bigg\vert\sum_{j=1}^{n} \big\Vert\widetilde{Z}_j\big\Vert- \E\,\big\Vert\widetilde{Z}_j\big\Vert \Bigg\vert
+ C_{\delta}(r)
+ \frac{1}{\delta}\sum_{u \in \{-\delta,0,\delta\}} \vert F(r+u) - f(r+u)\vert\\
&\qquad\qquad\qquad\qquad\leq \Big(\vert r \vert + \frac{3}{2}r^2 + \delta^2 \Big)\frac{B_X\sqrt{K}}{(2s_n)^{\nicefrac{3}{2}}n}
\Bigg\vert \sum_{j=1}^{n} \big\Vert\widetilde{Z}_j\big\Vert- \E\,\big\Vert\widetilde{Z}_j\big\Vert\Bigg\vert
+ C_{\delta}(r)\\
&\qquad\qquad\qquad\qquad\qquad\qquad\qquad\qquad\qquad\;+ \frac{1}{\delta}\sum_{u \in \{-\delta,0,\delta\}} \vert F_n(t,\widetilde{R}\cup\{r+u\}) - f_n(t,\widetilde{R}\cup\{r+u\})\vert \;.
\end{align*}
Thanks to the inequality $(\sum_{i=1}^{m} v_i)^2 \leq m \sum_{i=1}^{m} v_i^2$, this directly implies $\forall (r,\delta) \in [a(\widetilde{R}),b(\widetilde{R})] \times (0,s_n)$:
\begin{multline}\label{upperbound_variance_thermal_Lll'}
\E\Big[\Big(\langle \cL_{\ell \ell'}\rangle_{t,\widetilde{R}\cup\{r\}} - \E\,\langle \cL_{\ell \ell'}\rangle_{t,\widetilde{R}\cup\{r\}}\Big)^{\! 2}\,\Big]\\
\leq 5\Big(\vert r \vert + \frac{3}{2}r^2 + \delta^2 \Big)^2\frac{B_X^2 K}{2s_n^{3}n^2}
{\mathbb{V}\mathrm{ar}}\Bigg(\sum_{j=1}^{n} \big\Vert\widetilde{Z}_j\big\Vert\Bigg)
+ 5C_{\delta}(r)^2\\
+ \frac{5}{\delta^2}\sum_{u \in \{-\delta,0,\delta\}} \E\big[\big(F_n(t,\widetilde{R}\cup\{r+u\}) - f_n(t,\widetilde{R}\cup\{r+u\})\big)^2\big]\,.
\end{multline}
The next step is to bound the integral of the three summands on the right-hand side of \eqref{upperbound_variance_thermal_Lll'}.
Remember that $\forall r \in [a(\widetilde{R}),b(\widetilde{R})]: \widetilde{R}\cup\{r\} \in C(\mathcal{R}_{n,t})$.
By property (i) in Lemma \ref{lemma:properties_convex_hulls}, we have:
\begin{equation}\label{upperbound_r_in_[a,b]}
\forall r \in [a(\widetilde{R}),b(\widetilde{R})]: \vert r \vert \leq \Vert \widetilde{R}\cup\{r\} \Vert \leq 4K^{\nicefrac{3}{2}} + \mathrm{Tr}(\Sigma_X)^{p-1}\,.
\end{equation}
Besides, by independence of the Gaussian random vectors $\widetilde{Z}_j$, ${\mathbb{V}\mathrm{ar}}\big(\sum_{j=1}^{n} \Vert\widetilde{Z}_j\Vert\big) = n{\mathbb{V}\mathrm{ar}}\,\Vert\widetilde{Z}_1\Vert \leq nK$.
We conclude that there exists a positive constant $C_1$ depending only on $P_X$, $K$ and $p$ such that $\forall \delta \in (0,s_n)$:
\begin{equation}\label{upperbound_1st_summand}
\int_{a(\widetilde{R})}^{b(\widetilde{R})} \! dr\, 5\Big(\vert r \vert + \frac{3}{2}r^2 + \delta^2 \Big)^2\frac{B_X^2 K}{2s_n^{3}n^2}
{\mathbb{V}\mathrm{ar}}\Bigg(\sum_{j=1}^{n} \big\Vert\widetilde{Z}_j\big\Vert\Bigg)
\leq \frac{C_1}{s_n^3n} \,.
\end{equation}
Note that $C_{\delta}(r) = \vert C_{\delta}(r)\vert \leq \vert f'(r+\delta) \vert + \vert f'(r-\delta) \vert$. For all $q \in (a(\widetilde{R})-s_n,b(\widetilde{R})+s_n)$, we have:
\begin{multline}\label{upperbound_f'}
\vert f'(q) \vert
\leq
\Big\vert \E\,\langle \cL_{\ell \ell'}\rangle_{t,\widetilde{R}\cup\{q\}} \Big\vert
+ \vert q \vert \frac{B_X\sqrt{K}}{(2s_n)^{\nicefrac{3}{2}}n}\sum_{j=1}^{n} \E \Vert\widetilde{Z}_j \Vert\\
\leq \mathrm{Tr}(\Sigma_{X})
+ \big(s_n + 4K^{\nicefrac{3}{2}} + \mathrm{Tr}(\Sigma_X)^{p-1}\big) \frac{B_X K}{(2s_n)^{\nicefrac{3}{2}}}
\leq \frac{\widetilde{C}_2}{s_n^{\nicefrac{3}{2}}} \,,
\end{multline}
where $\widetilde{C}_2$ is a positive constant depending only on $P_X$, $K$ and $p$.
The second inequality in \eqref{upperbound_f'} follows from
the upper bounds
$\vert \E\,\langle \cL_{\ell \ell'}\rangle_{t,\widetilde{R}\cup\{q\}} \vert \leq \mathrm{Tr}(\Sigma_{X})$ (see \eqref{upperbound_derivative_fn_Rll'}), \eqref{upperbound_r_in_[a,b]} and
$\E \Vert\widetilde{Z}_j\Vert \leq \E[\Vert\widetilde{Z}_j\Vert^2]^{\nicefrac{1}{2}} = \sqrt{K}$.
Thus, for the second summand, we obtain $\forall \delta \in (0,s_n)$:
\begin{multline}\label{upperbound_2nd_summand}
\int_{a(\widetilde{R})}^{b(\widetilde{R})} \! dr\, C_{\delta}(r)^2
\leq \int_{a(\widetilde{R})}^{b(\widetilde{R})} \! dr\, \frac{2\widetilde{C}_2}{s_n^{\nicefrac{3}{2}}} C_{\delta}(r)\\
= \frac{2\widetilde{C}_2}{s_n^{\nicefrac{3}{2}}}
\big[\big(f(b(\widetilde{R})+\delta) - f(b(\widetilde{R})-\delta)\big) - \big(f(a(\widetilde{R})+\delta) - f(a(\widetilde{R})-\delta)\big)\big]
\leq \frac{8\delta\widetilde{C}_2^2}{s_n^{3}}\,.
\end{multline}
The last inequality is a simple application of the mean value theorem.
We finally turn to the third summand.
For every $\widetilde{R} \in C^{(\ell,\ell')}$ and pair $(r,\delta) \in [a(\widetilde{R}),b(\widetilde{R})] \times (-s_n,s_n)$, we have:
$$
\Vert \widetilde{R} \cup \{r+\delta\}\Vert
= \Vert \widetilde{R} \cup \{r\} + \delta E^{(\ell,\ell')}\Vert
\leq \Vert \widetilde{R} \cup \{r\} \Vert + \vert \delta \vert \, \Vert E^{(\ell,\ell')}\Vert
\leq 4K^{\nicefrac{3}{2}} + \mathrm{Tr}(\Sigma_X)^{p-1} + 2 \,.
$$
This upper bound is uniform in $n$ and $t$.
Hence, by Theorem \ref{th:concentration_free_entropy} of Appendix \ref{app:concentration_free_entropy}, there exists a positive constant $C_3$ depending only on $P_X$, $K$ and $p$ such that $\forall \widetilde{R} \in C^{(\ell,\ell')}$, $\forall(r,\delta) \in [a(\widetilde{R}),b(\widetilde{R})] \times (-s_n,s_n)$:
\begin{equation}\label{upperbound_variance_free_entropy}
\E\big[\big(F_n(t,\widetilde{R}\cup\{r+\delta\}) - f_n(t,\widetilde{R}\cup\{r+\delta\})\big)^2\big] \leq \frac{C_3}{n}\,.
\end{equation}
Using first \eqref{upperbound_variance_free_entropy} and then \eqref{upperbound_r_in_[a,b]}, we see that the third summand satisfies $\forall \delta \in (0,s_n)$:
\begin{multline}\label{upperbound_3rd_summand}
\int_{a(\widetilde{R})}^{b(\widetilde{R})} \! dr\,
\frac{5}{\delta^2}\sum_{u \in \{-\delta,0,\delta\}} \E\big[\big(F_n(t,\widetilde{R}\cup\{r+u\}) - f_n(t,\widetilde{R}\cup\{r+u\})\big)^2\big]\\
\leq \frac{15C_3}{\delta^2 n} (b(\widetilde{R}) - a(\widetilde{R}))
\leq \frac{30C_3}{\delta^2 n} \big(4K^{\nicefrac{3}{2}} + \mathrm{Tr}(\Sigma_X)^{p-1} + 2\big)\,.
\end{multline}
We now choose $\delta = \nicefrac{s_n^{\nicefrac{3}{2}}}{n^{\nicefrac{1}{3}}}$.
As $s_n \in (0,1)$, this choice satisfies $\delta \in (0,s_n)$.
The combination of \eqref{upperbound_variance_thermal_Lll'} with the three upper bounds \eqref{upperbound_1st_summand}, \eqref{upperbound_2nd_summand} and \eqref{upperbound_3rd_summand} shows the existence of a positive constant $C$ depending only on $P_X$, $K$ and $p$ such that:
\begin{equation}
\int_{a(\widetilde{R})}^{b(\widetilde{R})} \! dr\,
\E\Big[\Big(\langle \cL_{\ell \ell'}\rangle_{t,\widetilde{R}\cup\{r\}} - \E\,\langle \cL_{\ell \ell'}\rangle_{t,\widetilde{R}\cup\{r\}}\Big)^{\! 2}\,\Big]
\leq \frac{C}{s_n^3 n^{\nicefrac{1}{3}}}\,.
\end{equation}
One important fact following from our analysis is that $C$ can be chosen independently of both $(\ell,\ell') \in \{1,\dots,K\}^2$ and $\widetilde{R} \in C^{(\ell,\ell')}$.
Therefore, for all $(\ell,\ell') \in \{1,\dots,K\}^2$, we have
\begin{multline}\label{final_upperbound_quenched_fluctuations}
\int_{C(\mathcal{R}_{n,t})} dR\:
\E\Big[\Big(\langle \cL_{\ell \ell'}\rangle_{t,\widetilde{R}\cup\{r\}} - \E\,\langle \cL_{\ell \ell'}\rangle_{t,\widetilde{R}\cup\{r\}}\Big)^{\! 2}\,\Big]\\
= \int_{C^{(\ell,\ell')}} d\widetilde{R}\:
\int_{a(\widetilde{R})}^{b(\widetilde{R})} \! dr\,
\E\Big[\Big(\langle \cL_{\ell \ell'}\rangle_{t,\widetilde{R}\cup\{r\}} - \E\,\langle \cL_{\ell \ell'}\rangle_{t,\widetilde{R}\cup\{r\}}\Big)^{\! 2}\,\Big]
\leq \frac{C}{s_n^3 n^{\nicefrac{1}{3}}} V_{C^{(\ell,\ell')}}\,,
\end{multline}
where $V_{C^{(\ell,\ell')}}$ denotes the volume of $C^{(\ell,\ell')}$.
As each of the $K(K+1)/2$ sets $C^{(\ell,\ell')}$ is uniformly bounded in $n$ and $t$, the theorem follows from summing \eqref{final_upperbound_quenched_fluctuations} over $(\ell,\ell')$.
\end{proof}
\subsection{Concentration of \texorpdfstring{$\bQ$}{Q} around its expectation}
We forthwith use the concentration results for $\pmb{\cL}$, that is, Propositions~\ref{prop:concentration_L_on_<L>} and~\ref{prop:concentration_<L>_on_E<L>}, to prove Theorem~\ref{th:Barbier2020Overlap}. First an intermediary result on the thermal fluctuations of $\bQ$:
\begin{proposition}[Concentration of the overlap matrix around its expectation]\label{prop:concentration_Q_on_<Q>}
	Assume $P_X$ has finite fourth-order moments. There exists a positive constant $C$ depending only on $P_X$, $K$ and $p$ such that
	\begin{align}
	\int_{C(\mathcal{R}_{n,t})} dR\;\E\,\big\langle \big\Vert \bQ-\langle\bQ\rangle_{t,R}\big\Vert^2\,\big\rangle_{t,R}
	&\leq \frac{C}{\sqrt{s_n n}} \quad; \label{upperbound_int_(Q-<Q>)^2}\\
	\int_{C(\mathcal{R}_{n,t})} dR\;\E\,\bigg\langle \bigg\Vert \bQ-\frac{\langle\bx\rangle_{t,R}^T\langle\bx\rangle_{t,R}}{n}\bigg\Vert^2\,\bigg\rangle_{\!\! t,R}
	&\leq \frac{C}{\sqrt{s_n n}}  \quad.\label{upperbound_int_(Q-<x><x>/n)^2}
	\end{align}
\end{proposition}
\begin{proof}
Fix $(\ell,\ell') \in \{1,\dots,K\}^2$. Note that $\forall (t,R) \in [0,1] \times \mathcal{S}_K^{++}$:
\begin{multline}\label{eq:upperbound_(Q-<Q>)^2}
\E\,\langle (Q_{\ell\ell'} - \langle Q_{\ell\ell'}\rangle_{t,R})^2\rangle_{t,R}
= \frac{1}{n^2}\sum_{i,j=1}^{n} \E\big[X_{i\ell'}X_{j\ell'}(\langle x_{i\ell}x_{j\ell} \rangle_{t,R} - \langle x_{i\ell}\rangle_{ t,R}\langle x_{j\ell}\rangle_{t,R})\big]\\
\leq
\sqrt{M_{X}}
\bigg(\frac{1}{n^2}\sum_{i,j=1}^{n} \E\big[(\langle x_{i\ell}x_{j\ell} \rangle_{t,R} - \langle x_{i\ell}\rangle_{ t,R}\langle x_{j\ell}\rangle_{t,R})^2\big]\bigg)^{\!\!\nicefrac{1}{2}} \;,
\end{multline}
where $M_X \coloneqq \frac{1}{n^2}\sum_{i,j=1}^{n} \E[X_{i\ell'}^2X_{j\ell'}^2]$ is finite thanks to the assumption.
Differentiating with respect to $R_{\ell\ell}$ the identity $\E\,\langle \cL_{\ell\ell}\rangle_{t,R} = -\frac{1}{2}\E\,\langle Q_{\ell\ell}\rangle_{t,R}$ (see Lemma~\ref{lemma:computation_E<L>_and_others}), we obtain (see \cite{Barbier2020Overlap} for the detailed computation):
\begin{equation}
\frac{1}{n^2}\sum_{i,j=1}^{n} \E\big[(\langle x_{i\ell}x_{j\ell} \rangle_{t,R} - \langle x_{i\ell}\rangle_{ t,R}\langle x_{j\ell}\rangle_{t,R})^2\big]
= 2\E\,\big\langle \big(\cL_{\ell \ell'} - \langle \cL_{\ell \ell'}\rangle_{t,R}\big)^{2}\,\big\rangle_{t,R}
-\frac{2}{n}\E\,\bigg\langle \frac{\partial \cL_{\ell\ell}}{\partial R_{\ell\ell}}\bigg\rangle_{t,R} \;.
\end{equation}
By Proposition~\ref{prop:concentration_L_on_<L>} and the inequality~\eqref{upperbound_int_derivative_Lll'_Rll'} combined with~\eqref{formula_E<dL/dR>}, there exists a positive constant $C$ depending only on $P_X$, $K$ and $p$ such that:
\begin{equation}\label{eq:upperbound_int_<xx>-<x><x>}
\int_{C(\mathcal{R}_{n,t})}\frac{dR}{n^2}\sum_{i,j=1}^{n} \E\big[(\langle x_{i\ell}x_{j\ell} \rangle_{t,R} - \langle x_{i\ell}\rangle_{ t,R}\langle x_{j\ell}\rangle_{t,R})^2\big]\leq \frac{C}{n s_n} \;.
\end{equation}
Combining both inequalities \eqref{eq:upperbound_(Q-<Q>)^2} and \eqref{eq:upperbound_int_<xx>-<x><x>} with Cauchy-Schwarz inequality, it comes:
\begin{equation}
\int_{C(\mathcal{R}_{n,t})} dR\:\E\,\langle (Q_{\ell\ell'} - \langle Q_{\ell\ell'}\rangle_{t,R})^2\rangle_{t,R}
\leq \sqrt{\frac{M_X V_{C(\mathcal{R}_{n,t})}C}{ns_n}}\;,
\end{equation}
with $V_{C(\mathcal{R}_{n,t})}$ the volume of $C(\mathcal{R}_{n,t})$ which is bounded uniformly in $n$ and $t$ by (i) of Lemma~\ref{lemma:properties_convex_hulls}.
This ends the proof of \eqref{upperbound_int_(Q-<Q>)^2}.
The inequality \eqref{upperbound_int_(Q-<x><x>/n)^2} is proved in a similar way (see \cite{Barbier2020Overlap}).
\end{proof}
Finally we conclude this section with the proof of Theorem~\eqref{th:Barbier2020Overlap}.
\begin{proof}[Proof of Theorem~\ref{th:Barbier2020Overlap}]
To lighten notations we drop the subscripts of the Gibbs brackets $\langle - \rangle_{t,R}$.
The concentration of $\bQ$ can be linked to the concentration of $\pmb{\cL}$ by rewriting $\mathrm{Tr}\,\E\,\langle \bQ(\pmb{\cL} - \E\langle \pmb{\cL} \rangle)\rangle$ properly.
Thanks to the identity \eqref{formula_E<L>}, we have:
\begin{equation}\label{eq:Tr_E<Q>E<L>}
\mathrm{Tr}(\E[\langle \bQ \rangle]\E[\langle \pmb{\cL} \rangle])
= -\frac{1}{2}\sum_{\ell, \ell'=1}^K \E \langle Q_{\ell\ell'} \rangle \mathrm{Tr}\bigg(\frac{\partial R}{\partial R_{\ell\ell'}} \E \langle \bQ \rangle\bigg)\;.
\end{equation}
Plugging $\pmb{\cL}$'s definition \eqref{def_L} in $\mathrm{Tr}\,\E\,\langle \bQ \pmb{\cL}\rangle$ and integrating by parts with respect to the Gaussian random vectors $\widetilde{Z}_j$, $j \in [n]$, we find:
\begin{equation}\label{eq:Tr_E<QL>}
\mathrm{Tr}\,\E\,\langle \bQ \pmb{\cL}\rangle
= \frac{1}{n} \sum_{\ell,\ell'=1}^K \sum_{j=1}^n\E\bigg[\langle Q_{\ell\ell'} x_j^T\rangle \frac{\partial \sqrt{R}}{\partial R_{\ell\ell'}}\sqrt{R}\langle x_j\rangle\bigg]
-\E\bigg[\bigg\langle Q_{\ell\ell'} X_j^T \frac{\partial R}{\partial R_{\ell\ell'}} x_j\bigg\rangle\bigg]
\end{equation}
Note that $\forall (\ell,\ell') \in \{1,\dots,K\}^2,\forall j \in \{1,\dots,n\}$:
\begin{align}
\E\bigg[\langle Q_{\ell\ell'} x_j^T\rangle \frac{\partial \sqrt{R}}{\partial R_{\ell\ell'}}\sqrt{R}\langle x_j\rangle\bigg]
&= \E\bigg[\langle Q_{\ell\ell'} \rangle \langle x_j \rangle^T \frac{\partial \sqrt{R}}{\partial R_{\ell\ell'}}\sqrt{R}\langle x_j\rangle\bigg]
\!\!+ \E\bigg[\big\langle Q_{\ell\ell'} (x_j-\langle x_j \rangle)^T\big\rangle \frac{\partial \sqrt{R}}{\partial R_{\ell\ell'}}\sqrt{R}\langle x_j\rangle\bigg]\nonumber\\
&= \E\bigg[\frac{\langle Q_{\ell\ell'}\rangle}{2} \langle x_j \rangle^T\! \frac{\partial R}{\partial R_{\ell\ell'}}\langle x_j\rangle\bigg]
\!\!+ \E\bigg[\langle Q_{\ell'\ell} \rangle (X_j-\langle x_j \rangle)^T \frac{\partial \sqrt{R}}{\partial R_{\ell\ell'}}\sqrt{R}\langle x_j\rangle\bigg].\label{eq:identity_1st_sumand_TrE<QL>}
\end{align}
The second equality follows from \eqref{eq:identity_partial_diff_Rll'}, for the first expectation, and the Nishimori identity, for the second expectation. Plugging \eqref{eq:identity_1st_sumand_TrE<QL>} in \eqref{eq:Tr_E<QL>}, it comes:
\begin{multline}\label{eq:Tr_E<QL>_simplified}
\mathrm{Tr}\,\E\,\langle \bQ \pmb{\cL}\rangle
= -\sum_{\ell, \ell'=1}^K \E \bigg\langle Q_{\ell\ell'}  \mathrm{Tr}\bigg(\frac{\partial R}{\partial R_{\ell\ell'}} \bQ \bigg)\bigg\rangle
+\frac{1}{2}\sum_{\ell, \ell'=1}^K 
\E\bigg[\langle Q_{\ell\ell'}\rangle \mathrm{Tr}\bigg(\frac{\partial R}{\partial R_{\ell\ell'}}\frac{\langle \bx\rangle^T\langle \bx\rangle}{n}\bigg)\bigg]\\
+ \sum_{\ell, \ell'=1}^K \E\bigg[\langle Q_{\ell'\ell} \rangle \mathrm{Tr}\bigg(\frac{\partial \sqrt{R}}{\partial R_{\ell\ell'}}\sqrt{R}\bigg(\langle\bQ\rangle - \frac{\langle \bx\rangle^T\langle \bx\rangle}{n}\bigg)\bigg) \bigg]\;.
\end{multline}
Subtracting \eqref{eq:Tr_E<Q>E<L>} to \eqref{eq:Tr_E<QL>_simplified}, we obtain:
\begin{multline}\label{formula_TrE<Q(L-E<L>)>}
\mathrm{Tr}\,\E\,\langle \bQ(\pmb{\cL} - \E\langle \pmb{\cL} \rangle)\rangle
= -\sum_{\ell, \ell'=1}^K \E \bigg\langle Q_{\ell\ell'}  \mathrm{Tr}\bigg(\frac{\partial R}{\partial R_{\ell\ell'}} (\bQ- \E\langle\bQ\rangle)\bigg)\bigg\rangle\\
\qquad\qquad\qquad\qquad+\frac{1}{2}\sum_{\ell, \ell'=1}^K 
\E\bigg[\langle Q_{\ell\ell'}\rangle \mathrm{Tr}\bigg(\frac{\partial R}{\partial R_{\ell\ell'}}
\bigg(\frac{\langle \bx\rangle^T\langle \bx\rangle}{n}-\E\langle\bQ\rangle\bigg)\bigg)\bigg]\\
+ \sum_{\ell, \ell'=1}^K \E\bigg[\langle Q_{\ell'\ell} \rangle \mathrm{Tr}\bigg(\frac{\partial \sqrt{R}}{\partial R_{\ell\ell'}}\sqrt{R}\bigg(\langle\bQ\rangle - \frac{\langle \bx\rangle^T\langle \bx\rangle}{n}\bigg)\bigg) \bigg]\;.
\end{multline}
Remember the matrices $E^{(\ell,\ell')}$ defined by \eqref{def:E^(l,l')}.
As $\nicefrac{\partial R}{\partial R_{\ell\ell'}}=E^{(\ell,\ell')}$, we have:
\begin{align}
&\sum_{\ell, \ell'=1}^K \E \bigg\langle Q_{\ell\ell'}  \mathrm{Tr}\bigg(\frac{\partial R}{\partial R_{\ell\ell'}} (\bQ- \E\langle\bQ\rangle)\bigg)\bigg\rangle\nonumber\\
&\qquad\qquad\quad= \E\,\big\langle \mathrm{Tr}\big( \bQ (\bQ- \E\langle\bQ\rangle)^T\big)\big\rangle
+ \E\,\big\langle \mathrm{Tr}\big( \bQ (\bQ- \E\langle\bQ\rangle)\big)\big\rangle
- \sum_{\ell=1}^K \E \big\langle Q_{\ell\ell}  (Q_{\ell\ell}- \E\langle Q_{\ell\ell}\rangle)\big\rangle\nonumber\\
&\qquad\qquad\quad= \E\,\big\langle \big\Vert \bQ- \E\langle\bQ\rangle \big\Vert^2\big\rangle
+ \E\,\big\langle \mathrm{Tr}\big( \bQ (\bQ- \E\langle\bQ\rangle)\big)\big\rangle
- \sum_{\ell=1}^K \E \big\langle (Q_{\ell\ell}- \E\langle Q_{\ell\ell}\rangle)^2\big\rangle\;\;;\label{eq:formula_1stterm_TrE<Q(L-E<L>)>}\\
&\frac{1}{2}\sum_{\ell, \ell'=1}^K 
\E\bigg[\langle Q_{\ell\ell'}\rangle \mathrm{Tr}\bigg(\frac{\partial R}{\partial R_{\ell\ell'}}
\bigg(\frac{\langle \bx\rangle^T\langle \bx\rangle}{n}-\E\langle\bQ\rangle\bigg)\bigg)\bigg]\nonumber\\*
&\qquad\qquad\quad= \E\,\bigg\langle \mathrm{Tr}\bigg( \bQ \bigg(\frac{\langle \bx\rangle^T\langle \bx\rangle}{n}-\E\langle\bQ\rangle\bigg)\bigg)\bigg\rangle
- \frac{1}{2}\sum_{\ell=1}^K \E\bigg[\langle Q_{\ell\ell}\rangle  \bigg(\frac{\langle \bx\rangle^T\langle \bx\rangle}{n}- \E\langle \bQ\rangle\bigg)_{\!\!\!\ell\ell}\,\bigg]\,.\label{eq:formula_2ndterm_TrE<Q(L-E<L>)>}
\end{align}
Subtracting \eqref{eq:formula_1stterm_TrE<Q(L-E<L>)>} to \eqref{eq:formula_2ndterm_TrE<Q(L-E<L>)>} yields:
\begin{align}
&\frac{1}{2}\sum_{\ell, \ell'=1}^K 
\E\bigg[\langle Q_{\ell\ell'}\rangle \mathrm{Tr}\bigg(\frac{\partial R}{\partial R_{\ell\ell'}}
\bigg(\frac{\langle \bx\rangle^T\langle \bx\rangle}{n}-\E\langle\bQ\rangle\bigg)\bigg)\bigg]
-\sum_{\ell, \ell'=1}^K \E \bigg\langle Q_{\ell\ell'}  \mathrm{Tr}\bigg(\frac{\partial R}{\partial R_{\ell\ell'}} (\bQ- \E\langle\bQ\rangle)\bigg)\bigg\rangle\nonumber\\
&\qquad\qquad=-\E\,\big\langle \big\Vert \bQ- \E\langle\bQ\rangle \big\Vert^2\big\rangle
-\E\,\bigg\langle \mathrm{Tr}\bigg( \bQ \bigg(\bQ-\frac{\langle \bx\rangle^T\langle \bx\rangle}{n}\bigg)\bigg)\bigg\rangle\nonumber\\
&\qquad\qquad\qquad\qquad\qquad+\sum_{\ell=1}^K \E \big\langle (Q_{\ell\ell}- \E\langle Q_{\ell\ell}\rangle)^2\big\rangle
- \frac{1}{2}\sum_{\ell=1}^K \E\bigg[\langle Q_{\ell\ell}\rangle  \bigg(\frac{\langle \bx\rangle^T\langle \bx\rangle}{n}- \E\langle \bQ\rangle\bigg)_{\!\!\!\ell\ell}\,\bigg]\nonumber\\
&\qquad\qquad=-\E\,\big\langle \big\Vert \bQ- \E\langle\bQ\rangle \big\Vert^2\big\rangle
-\E\,\bigg\langle \mathrm{Tr}\bigg( \bQ \bigg(\bQ-\frac{\langle \bx\rangle^T\langle \bx\rangle}{n}\bigg)\bigg)\bigg\rangle\nonumber\\
&\qquad\qquad\qquad\qquad\qquad+\frac{1}{2}\sum_{\ell=1}^K \E \big\langle (Q_{\ell\ell}- \E\langle Q_{\ell\ell}\rangle)^2\big\rangle
+ \frac{1}{2}\sum_{\ell=1}^K \E\bigg[\bigg\langle Q_{\ell\ell} \bigg(\bQ-\frac{\langle \bx\rangle^T\langle \bx\rangle}{n}\bigg)_{\!\!\!\ell\ell}\bigg\rangle\bigg]\nonumber\\
&\qquad\qquad=-\E\,\big\langle \big\Vert \bQ- \E\langle\bQ\rangle \big\Vert^2\big\rangle
-\E\,\bigg\langle \mathrm{Tr}\,\bigg(\bQ - \frac{\langle \bx\rangle^T\langle \bx\rangle}{n}\bigg)^{\!\! 2}\,\bigg\rangle\nonumber\\
&\qquad\qquad\qquad\qquad\qquad
+\frac{1}{2}\sum_{\ell=1}^K \E \big\langle (Q_{\ell\ell}- \E\langle Q_{\ell\ell}\rangle)^2\big\rangle
+ \frac{1}{2}\sum_{\ell=1}^K \E\,\bigg\langle\! \bigg(Q_{\ell\ell}-\frac{\langle \bx\rangle^T\langle \bx\rangle}{n}\bigg\vert_{\ell\ell}\bigg)^{\!\!\! 2} \,\bigg\rangle
\label{eq:simplification_first_terms_TrE<Q(L-E<L>)>}\quad.
\end{align}
Plugging \eqref{eq:simplification_first_terms_TrE<Q(L-E<L>)>} in \eqref{formula_TrE<Q(L-E<L>)>} gives the equality:
\begin{align}
&\E\,\big\langle \big\Vert \bQ- \E\langle\bQ\rangle \big\Vert^2\big\rangle
-\frac{1}{2}\sum_{\ell=1}^K \E \big\langle (Q_{\ell\ell}- \E\langle Q_{\ell\ell}\rangle)^2\big\rangle\nonumber\\
&\qquad=-\mathrm{Tr}\,\E\,\langle \bQ(\pmb{\cL} - \E\langle \pmb{\cL} \rangle)\rangle
-\E\,\bigg\langle \mathrm{Tr}\:\bigg(\bQ - \frac{\langle \bx\rangle^T\langle \bx\rangle}{n}\bigg)^{\!\!\!  2}\,\bigg\rangle
+ \frac{1}{2}\sum_{\ell=1}^K \E\,\bigg\langle\! \bigg(Q_{\ell\ell}-\frac{\langle \bx\rangle^T\langle \bx\rangle}{n}\bigg\vert_{\ell\ell}\bigg)^{\!\!\! 2}\,\bigg\rangle\nonumber\\
&\qquad\qquad\qquad\qquad\qquad\qquad\qquad\qquad
+\sum_{\ell, \ell'=1}^K \E\bigg[\langle Q_{\ell'\ell} \rangle \mathrm{Tr}\bigg(\frac{\partial \sqrt{R}}{\partial R_{\ell\ell'}}\sqrt{R}\bigg(\langle\bQ\rangle - \frac{\langle \bx\rangle^T\langle \bx\rangle}{n}\bigg)\bigg) \bigg]\;.
\label{final_equality_Barbier2020Overlap}
\end{align}
On one hand,
\begin{equation}\label{lowerbound_E<(Q-<Q>)^2>}
\frac{1}{2}\E\,\langle \Vert \bQ- \E\langle\bQ\rangle \Vert^2\rangle
\leq \E\,\langle \Vert \bQ- \E\langle\bQ\rangle\Vert^2\rangle
-\frac{1}{2}\sum_{\ell=1}^K \E \langle (Q_{\ell\ell}- \E\langle Q_{\ell\ell}\rangle)^2\rangle\;.
\end{equation}
On the other hand, using exclusively Cauchy-Schwarz inequality, we have:
\begin{align}
-\mathrm{Tr}\,\E\,\langle \bQ(\pmb{\cL} - \E\langle \pmb{\cL} \rangle)\rangle
&\leq \sqrt{\E\,\langle \Vert \bQ \Vert^2 \rangle
\E\,\langle \Vert \pmb{\cL} - \E\langle \pmb{\cL} \rangle\Vert^2\rangle}\quad;\label{upperbound_1st_term_E<(Q-<Q>)^2>}\\
-\E\,\bigg\langle \mathrm{Tr}\:\bigg(\bQ - \frac{\langle \bx\rangle^T\langle \bx\rangle}{n}\bigg)^{\!\!\!  2}\,\bigg\rangle
&\leq \E\,\bigg\langle \bigg\Vert\bQ - \frac{\langle \bx\rangle^T\langle \bx\rangle}{n}\bigg\Vert^2\,\bigg\rangle\quad.\label{upperbound_2nd_term_E<(Q-<Q>)^2>}
\end{align}
Note that $\Vert \sqrt{R} \Vert = \sqrt{\mathrm{Tr}\, R}
\leq \sqrt{\Vert R \Vert \Vert I_K\Vert}
\leq \sqrt{4K^2 + K^{\nicefrac{1}{2}} \mathrm{Tr}(\Sigma_X)^{p-1}} \coloneqq B$ where the last inequality follows from (i) in Lemma~\ref{lemma:properties_convex_hulls}.
Therefore,
$\big\Vert\frac{\partial \sqrt{R}}{\partial R_{\ell\ell'}}\sqrt{R}\big\Vert
\leq \big\Vert\frac{\partial \sqrt{R}}{\partial R_{\ell\ell'}}\big\Vert \Vert \sqrt{R} \Vert \leq \nicefrac{B}{\sqrt{2s_n}}$
(remember \eqref{upperbound_norm_derivative_sqrtR}). By Cauchy-Schwarz inequality:
\begin{align*}
\sum_{\ell, \ell'=1}^K \E\bigg[\langle Q_{\ell'\ell} \rangle \mathrm{Tr}\bigg(\frac{\partial \sqrt{R}}{\partial R_{\ell\ell'}}\sqrt{R}\bigg(\langle\bQ\rangle - \frac{\langle \bx\rangle^T\langle \bx\rangle}{n}\bigg)\bigg) \bigg]
&\leq
\sum_{\ell, \ell'=1}^K \E\bigg[\big\vert\langle Q_{\ell'\ell} \rangle\big\vert \bigg\Vert\frac{\partial \sqrt{R}}{\partial R_{\ell\ell'}}\sqrt{R}\bigg\Vert \bigg\Vert\langle\bQ\rangle - \frac{\langle \bx\rangle^T\langle \bx\rangle}{n}\bigg\Vert \bigg]\\
&\leq
\frac{B}{\sqrt{2 s_n}}\sum_{\ell, \ell'=1}^K \E\bigg[\big\vert\langle Q_{\ell'\ell} \rangle\big\vert \bigg\Vert\langle\bQ\rangle - \frac{\langle \bx\rangle^T\langle \bx\rangle}{n}\bigg\Vert \bigg]\\
&\leq
\frac{B K}{\sqrt{2 s_n}} \bigg( \E\,\Vert \langle\bQ\rangle \Vert^2 \; \E\,\bigg\Vert\langle\bQ\rangle - \frac{\langle \bx\rangle^T\langle \bx\rangle}{n}\bigg\Vert^2\, \bigg)^{\nicefrac{1}{2}}\;.
\end{align*}
Further upperbounding, we obtain
\begin{equation}\label{upperbound_3rd_term_E<(Q-<Q>)^2>}
\sum_{\ell, \ell'=1}^K \E\bigg[\langle Q_{\ell'\ell} \rangle \mathrm{Tr}\bigg(\frac{\partial \sqrt{R}}{\partial R_{\ell\ell'}}\sqrt{R}\bigg(\langle\bQ\rangle - \frac{\langle \bx\rangle^T\langle \bx\rangle}{n}\bigg)\bigg) \bigg]
\leq
\frac{B K}{\sqrt{2 s_n}} \sqrt{\E\,\langle\Vert\bQ\Vert^2\rangle \;
\E\,\big\langle \big\Vert \bQ - \langle \bQ\rangle \big\Vert^2\big\rangle}\;;
\end{equation}
as, by Jensen's inequality and Nishimori identity, we have:
\begin{equation*}
\E\,\bigg\Vert\langle\bQ\rangle - \frac{\langle \bx\rangle^T\langle \bx\rangle}{n}\bigg\Vert^2
\leq \E\,\bigg\langle \bigg\Vert \bQ - \frac{\bx^T\langle \bx\rangle}{n}\bigg\Vert^2\bigg\rangle
= \E\,\big\langle \big\Vert \bQ^T - \langle \bQ\rangle^T \big\Vert^2\big\rangle
= \E\,\big\langle \big\Vert \bQ - \langle \bQ\rangle \big\Vert^2\big\rangle \;.
\end{equation*}
Putting together the equality \eqref{final_equality_Barbier2020Overlap}, the lower bound \eqref{lowerbound_E<(Q-<Q>)^2>} and the upper bounds \eqref{upperbound_1st_term_E<(Q-<Q>)^2>}, \eqref{upperbound_2nd_term_E<(Q-<Q>)^2>}, \eqref{upperbound_3rd_term_E<(Q-<Q>)^2>}, there exists a positive constant $C$ depending only on $P_X$, $K$ and $p$ such that:
\begin{multline}\label{final_upperbound_E<(Q-<Q>)^2>}
\frac{1}{2}\E\,\langle \Vert \bQ- \E\langle\bQ\rangle \Vert^2\rangle\\
\leq
C \Bigg(
\sqrt{\E\,\langle \Vert \pmb{\cL} - \E\langle \pmb{\cL} \rangle\Vert^2\rangle}
+\E\,\bigg\langle \bigg\Vert\bQ - \frac{\langle \bx\rangle^T\langle \bx\rangle}{n}\bigg\Vert^2\,\bigg\rangle
+ \sqrt{\frac{\E\,\big\langle \big\Vert \bQ - \langle \bQ\rangle \big\Vert^2\big\rangle}{s_n}}\:\Bigg)\;.
\end{multline}
To end the proof of Theorem~\ref{th:Barbier2020Overlap}, it remains to integrate both sides of \eqref{final_upperbound_E<(Q-<Q>)^2>} over $C(\mathcal{R}_{n,t})$ and apply Propositions~\ref{prop:concentration_L_on_<L>}, \ref{prop:concentration_<L>_on_E<L>}, \ref{prop:concentration_Q_on_<Q>}.
\end{proof}
\section{Conclusion and discussion for odd-order tensors}\label{sec:discussion}
In this work, we have proved the conjectured replica formula for even-order symmetric tensors.
It would be desirable to extend both Theorem~\ref{th:lowerbound} and Theorem~\ref{th:upperbound} to the odd-order case.
For the case $K=1$ we refer to \cite{Lesieur2017Statistical}.
For $K>1$, this is still an open problem and we now briefly discuss where our proofs fall short in this case.

Ideally, to extend Theorem~\ref{th:lowerbound} to an odd order $p$, we would show that the integral on the r.h.s.\ of \eqref{sum_rule_lowerbound}, i.e.,
$\int_{0}^{1} dt \,\sum_{\ell,\ell'} \E\,\langle h_p(S_{\ell \ell'}, Q_{\ell \ell'}) \rangle_{t,0}$
with $h_p(r,q)=q^p - p q r^{p-1} + (p-1)r^p$, is nonnegative.
However, when $p$ is odd, $h_p$ is not nonnegative on its whole domain of definition.
To be able to say something about the integral, we have to take a Gibbs average of $Q_{\ell \ell'}$ before applying $h_p$. This requires rewriting the integral as follows:
\begin{multline}\label{remainder_rewritten}
\int_{0}^{1} dt \sum_{\ell,\ell'=1}^K \E\,\big\langle h_p(S_{\ell \ell'}, Q_{\ell \ell'}) \big\rangle_{t,0}
=
\int_{0}^{1} dt \E\,\bigg\langle \mathrm{Tr}\:\bQ^T\bigg(\bQ^{\circ(p-1)}-\bigg(\frac{\langle \bx \rangle_{t,0}^T\langle \bx \rangle_{t,0}}{n}\bigg)^{\!\!\!\circ(p-1)}\,\bigg)\bigg\rangle_{\!\!t,0}\\
+\int_{0}^{1} dt \sum_{\ell,\ell'=1}^K \E\,h_p\bigg(S_{\ell \ell'}, \frac{\langle \bx \rangle_{t,0}^T\langle \bx \rangle_{t,0}}{n}\bigg\vert_{\ell\ell'}\bigg) \;.
\end{multline}
When $K=1$, both $\langle \bx \rangle_{t,0}^T\langle \bx \rangle_{t,0}$ and $S$ are nonnegative real numbers.
The nonnegativity of $h_p(r,q)$ for $r,q \geq 0$ then ensures that the second integral on the r.h.s.\ of \eqref{remainder_rewritten} is nonnegative and, 
by introducing a small perturbation $\epsilon$ on which we integrate, we can cancel the first integral as was done in the proof of Theorem~\ref{th:upperbound}.
This is how the lower bound is proved in \cite{Lesieur2017Statistical}.
When $K>1$, we only know that $\langle \bx \rangle_{t,0}^T\langle \bx \rangle_{t,0}$ and $S$ are symmetric positive semidefinite matrices: \textit{a priori} nothing can be said on the sign of their individual entries.
The problem remains if we further rewrite:
\begin{multline}\label{remainder_rewritten_2}
\int_{0}^{1} dt \sum_{\ell,\ell'=1}^K \E\,\big\langle h_p(S_{\ell \ell'}, Q_{\ell \ell'}) \big\rangle_{t,0}
=
\int_{0}^{1} dt \E\,\big\langle \mathrm{Tr}\big(\bQ^T\big(\bQ^{\circ(p-1)}
-\E[\langle \bQ\rangle_{t,0}]^{\circ(p-1)}\,\big)\big)\big\rangle_{\!t,0}\\
+\int_{0}^{1} dt \sum_{\ell,\ell'=1}^K h_p\big(S_{\ell \ell'}, \E\,\langle Q_{\ell\ell'}\rangle_{t,0}\big) \;.
\end{multline}
While $\E\,\langle \bQ\rangle_{t,0}$ and $S$ are positive semidefinite, nothing can be said on the sign of their individual entries.
Most probably, it should be
the full sum over $(\ell,\ell')$ that one should consider to conclude on the sign of the second integral on the r.h.s.\ of \eqref{remainder_rewritten_2}. 
Indeed, using $A \succcurlyeq B \succcurlyeq 0 \Rightarrow \forall k \in \mathbb{N}: A^{\circ k} \succcurlyeq B^{\circ k} \succcurlyeq 0$, we can show that
$\sum_{\ell,\ell'=1}^K h_p(S_{\ell \ell'}, \E\,\langle Q_{\ell\ell'}\rangle_{t,0})$
is nonnegative if $S \succcurlyeq \E\,\langle \bQ \rangle_{t,0}$ or $\E\,\langle \bQ \rangle_{t,0} \succcurlyeq S$.
As far as we can tell, it is not clear why such partial ordering between $S$ and $\E\,\langle \bQ \rangle_{t,0}$ (which itself depends on $S$) holds.

Regarding Theorem~\ref{th:upperbound}, the whole proof would directly apply to $p$ odd if we could show that the divergence~\eqref{divergence_Fn} is nonnegative.
However this is more difficult than for $p$ even. Indeed, while the $\Delta_{\ell \ell'}$'s are still $\ge 0$, it is not necessarily 
the case of $\E[\langle Q_{\ell \ell'} \rangle_{t,R}]^{p-2}$ when $p-2$ is odd.
\section*{Funding}
This work was supported by the Swiss National Science Foundation [200021E-175541 to C. L].
%

\newpage
\ifx\undefined\BySame
\newcommand{\BySame}{\leavevmode\rule[.5ex]{3em}{.5pt}\ }
\fi
\ifx\undefined\textsc
\newcommand{\textsc}[1]{{\sc #1}}
\newcommand{\emph}[1]{{\em #1\/}}
\let\tmpsmall\small
\renewcommand{\small}{\tmpsmall\sc}
\fi

%
\newpage
\appendix
\section{Properties of the function \texorpdfstring{$\psi$}{psi}}\label{app:properties_psi}
\begin{lemma}\label{lemma:properties_psi}
	Let $X \in \mathbb{R}^K \sim P_X$ and $\widetilde{Z} \in \mathbb{R}^K \sim \cN(0,I_K)$. The function $\psi: \mathcal{S}_K^{+} \to \mathbb{R}$, defined as
	\begin{equation}
	\psi(R)
	= \E_{X,\widetilde{Z}}\bigg[\ln \int dP_X(x)
	\exp\bigg(\big(R X + \sqrt{R}\widetilde{Z}\big)^T x -\frac{1}{2}x^T R x\bigg)\bigg]\,,
	\end{equation}
	is Lipschitz continuous with Lipschitz constant $\nicefrac{\mathrm{Tr}(\Sigma_X)}{2}$ and convex.
\end{lemma}
	\begin{proof}
		Consider the inference problem in which one observes the $K$-dimensional vector $Y = \sqrt{R}X + \widetilde{Z}$, where $R \in \mathcal{S}_K^+$ is known, and one wants to recover $X$. The posterior of X given Y is
		\begin{equation}\label{eq:posterior_psi}
		dP(x \,\vert\, Y,R) = \frac{1}{\cZ_{R}(Y)}\exp\bigg(Y^T \sqrt{R}x -\frac{1}{2}x^T R x\bigg)\,,
		\end{equation}
		where $\cZ_{R}(Y) = \int dP_X(x)\exp\big(Y^T \sqrt{R}x -\frac{1}{2}x^T R x\big)$.
		We denote $\langle - \rangle_R = \int -\,dP(x \,\vert\, Y,R)$ the Gibbs brackets associated to the latter posterior distribution.
		Clearly, $\psi(R)=\E_{X,\widetilde{Z}}[\ln \cZ_{R}(Y)]$.
		
		Now	fix $R,Q \in \mathcal{S}_K^{++}$.
		We will prove that the function $h: t \in [0,1] \mapsto \psi(tR + (1-t)Q)$ is convex, thus proving that $\psi$ is convex on $\mathcal{S}_K^{++}$. The convexity on the whole cone $\mathcal{S}_K^+$ will then follow from the continuity of $\psi$ (which is clear from its definition).
		$h$ is twice differentiable.
		Its derivative reads:
		\begin{align}
		h'(t) &= \E\bigg[\bigg\langle X^T(R-Q)x-\frac{1}{2}x^T(R-Q)x +\widetilde{Z}^T\frac{d\sqrt{tR+(1-t)Q}}{dt}x\bigg\rangle_{\!\! tR + (1-t)Q} \,\bigg]\nonumber\\
		&= \frac{1}{2}\E\big[X^T(R-Q) \langle x \rangle_{tR + (1-t)Q}\,\big]\,.\label{eq:first_derivative_h}
		\end{align}
		To get the second equality, first integrate by parts with respect to the Gaussian random variables $\widetilde{Z}_i$, $i \in [n]$.
		Then make use of the identity
		\begin{equation}
		\forall v \in \mathbb{R}^K:\;\;
		v^T\sqrt{tR+(1-t)Q}\frac{d\sqrt{tR+(1-t)Q}}{dt}v =\frac{1}{2}v^T(R-Q)v \;,\label{eq:identity_t_derivative_sqrt_R}
		\end{equation}
		which follows from
		$\sqrt{tR+(1-t)Q}\frac{d\sqrt{tR+(1-t)Q}}{dt} + \frac{d\sqrt{tR+(1-t)Q}}{dt}\sqrt{tR+(1-t)Q}=\frac{d(tR+(1-t)Q)}{dt}$.
		Differentiating~\eqref{eq:first_derivative_h} further, we find (the subscript of $\langle - \rangle_{tR+(1-t)Q}$ is omitted):
		\begin{align*}
		h''(t)
		&= \frac{1}{2}\E\bigg[X^T(R-Q) \bigg\langle x \bigg(X^T(R-Q)x-\frac{1}{2}x^T(R-Q)x +\widetilde{Z}^T\frac{d\sqrt{tR+(1-t)Q}}{dt}x\bigg)\bigg\rangle\bigg]\\
		&\qquad\qquad\quad
		-\frac{1}{2}\E\bigg[X^T(R-Q) \langle x \rangle \bigg\langle X^T(R-Q)x-\frac{1}{2}x^T(R-Q)x +\widetilde{Z}^T\frac{d\sqrt{tR+(1-t)Q}}{dt}x\bigg\rangle\bigg]\\
		&= \frac{1}{2}\E\Big[\Big\langle (X^T(R-Q)x)^2\Big\rangle\,\Big]
		-\E\Big[\Big\langle X^T(R-Q) x \Big\rangle^{\! 2}\,\Big]
		+ \frac{1}{2}\E\Big[\Big(\langle x \rangle^T(R-Q) \langle x \rangle\Big)^{\! 2}\,\Big]\\
		&= \frac{1}{2}\E\Big[\mathrm{Tr}\Big(XX^T(R-Q) \langle x  x^T \rangle (R-Q)\Big)\,\Big]
		-\E\Big[\mathrm{Tr}\Big(XX^T(R-Q) \langle x\rangle\langle x\rangle^{T}(R-Q)\Big)\Big]\\
		&\qquad\qquad\qquad\qquad\qquad\qquad\qquad\qquad\qquad\:
		+ \frac{1}{2}\E\Big[\mathrm{Tr}\Big(\langle x \rangle\langle x \rangle^T(R-Q) \langle x \rangle \langle x \rangle^T(R-Q)\Big)\Big]\\
		&= \frac{1}{2}\E\,\big\Vert \langle x  x^T \rangle (R-Q)\big\Vert^2
		-\E\Big[\mathrm{Tr}\Big(\langle x  x^T \rangle(R-Q) \langle x\rangle\langle x\rangle^{T}(R-Q)\Big)\Big]
		+ \frac{1}{2}\E\,\big\Vert \langle x \rangle \langle x \rangle^T(R-Q)\big\Vert^{2}\\
		&= \frac{1}{2}\E\,\big\Vert \big(\langle x  x^T \rangle-\langle x \rangle \langle x \rangle^T\big)(R-Q)\big\Vert^2\,.
		\end{align*}
		To get the second equality, we once again used Gaussian integration by parts and the identity \eqref{eq:identity_t_derivative_sqrt_R}.
		The second-to-last equality follows from the Nishimori identity:
		$$
		\E\Big[\mathrm{Tr}\Big(XX^T(R-Q) \langle x\rangle\langle x\rangle^{T}(R-Q)\Big)\Big]
		= \E\Big[\mathrm{Tr}\Big(\langle xx^T\rangle(R-Q) \langle x\rangle\langle x\rangle^{T}(R-Q)\Big)\Big]\,.
		$$
		The convexity of $h$ now follows directly from the non-negativity of $h''$ on $[0,1]$.
		
		To prove the Lipschitz continuity of $\psi$, note that the derivative of $h$ satisfies $\forall t \in [0,1]$:
		\begin{equation}\label{eq:upperbound_derivative_h}
		\vert h^{\prime}(t) \vert
		= \frac{1}{2} \big\vert \E\big[X^T(R-Q) \langle x \rangle_{tR + (1-t)Q}\,\big]\big\vert
		\leq \frac{\Vert R - Q\Vert}{2}  \E\big[\Vert X \Vert \Vert \langle x \rangle_{tR + (1-t)Q}\Vert \big]
		\leq \frac{\mathrm{Tr}(\Sigma_X)}{2} \Vert R - Q\Vert\;.
		\end{equation}
		The mean value theorem then directly implies $\vert \psi(R) - \psi(Q)\vert = \vert h(1) - h(0) \vert \leq \frac{\mathrm{Tr}\,\Sigma_X}{2} \Vert R - Q\Vert$.
		The last inequality in \eqref{eq:upperbound_derivative_h} follows from Cauchy-Schwarz inequality, Jensen's inequality and the Nishimori identity:
		\begin{multline*}
			\E\big[\Vert X \Vert \Vert \langle x \rangle_{tR + (1-t)Q}\Vert \big]
		\leq \sqrt{\E\,\Vert X \Vert^2 \:\E\,\Vert \langle x \rangle_{tR + (1-t)Q}\Vert^2}\\
		\leq \sqrt{\E\,\Vert X \Vert^2 \:\E\,\langle \Vert x \Vert^2 \rangle_{tR + (1-t)Q}}
		= \E\,\Vert X \Vert^2 =\mathrm{Tr}(\Sigma_X)\;.
		\end{multline*}
	\end{proof}
\section{Divergence of the function \texorpdfstring{$G_n$}{G\_\{n\}}}\label{app:divergence}
In Proposition \ref{prop:ode_and_properties} we introduced the inference problem \eqref{inference_problem_t_R}.
The associated posterior distribution is
\begin{equation}\label{posterior_dist_t,R}
dP(\bx\,\vert\,\bY^{(t)},\widetilde{\bY}^{(t,R)}) \coloneqq \frac{1}{\cZ_{t,R}(\bY^{(t)},\widetilde{\bY}^{(t,R)})}
\, e^{-\cH_{t,R}(\bx ; \bY^{(t)},\widetilde{\bY}^{(t,R)})}\, \prod_{j=1}^n dP_{X}(x_j) \;,
\end{equation}
where $\cZ_{t,R}(\bY^{(t)},\widetilde{\bY}^{(t,R)}) \coloneqq \int e^{-\cH_{t,R}(\bx ; \bY^{(t)},\widetilde{\bY}^{(t,R)})}\, \prod_{j=1}^n dP_{X}(x_j)$
and $\forall \bx \in \mathbb{R}^{n \times K}$:
\begin{multline}
\cH_{t,R}(\bx ; \bY, \widetilde{\bY})
\coloneqq \sum_{i \in \mathcal{I}}
\frac{(1-t) (p-1)!}{2n^{p-1}} \Bigg(\sum_{k=1}^{K}\prod_{a=1}^p x_{i_a k}\Bigg)^{\! 2} -\sqrt{\frac{(1-t) (p-1)!}{n^{p-1}}} Y_{i}  \sum_{k=1}^{K} \prod_{a=1}^p x_{i_a k}\\
+\sum_{j=1}^{n} \frac{1}{2} x_j^T R x_j- \widetilde{Y}_j^T  \sqrt{R} x_j \,.
\end{multline}
Let $\langle - \rangle_{t,R}=\int - \,dP(\bx\,\vert\,\bY^{(t)},\widetilde{\bY}^{(t,R)})$ be the Gibbs brackets associated to the posterior distribution \eqref{posterior_dist_t,R}.
In this appendix we prove a formula for the divergence of the function
\begin{equation}\label{eq:def_Fn}
G_n:
\begin{array}{ccl}
[0,1] \times \mathcal{S}_K^+ &\to& \mathcal{S}_K^+\\
(t,R) &\mapsto& \E[\langle \bQ \rangle_{t,R}]^{\circ(p-1)}
\end{array}\quad.
\end{equation}
\begin{lemma}[Divergence of $G_n$]\label{lemma:divergence_Fn}
Let $\delta_{\ell\ell'} = 0$ if $\ell \neq \ell'$ and $\delta_{\ell\ell'} = 1$ otherwise.
$\forall (\ell,\ell') \in \{1,\dots,K\}^2$:
\begin{multline}\label{eq:partial_derivative_F}
	\frac{\partial (G_n)_{\ell \ell'}}{\partial R_{\ell\ell'}}\bigg\vert_{t,R}
	= \frac{n (p-1)}{1+\delta_{\ell \ell'}} \E[\langle Q_{\ell\ell'}\rangle_{t,R}\,]^{p-2}
	\bigg(
	\E\big[\big\langle \bQ \circ \big(\bQ + \bQ^T - \langle \bQ + \bQ^T\rangle_{t,R}\big)\big\rangle_{t,R}\,\big]\big\vert_{\ell\ell'}\\
	-\E\bigg[\langle \bQ^T \rangle_{t,R}\circ\bigg(\!\langle \bQ + \bQ^T\rangle_{t,R} - 2\frac{\langle \bx \rangle_{t,R}^T \langle \bx \rangle_{t,R}}{n}\bigg)\bigg]\bigg\vert_{\ell\ell'}
	\bigg)\,.
\end{multline}
Then, the divergence of $G_n$ is 
$$
\sum_{1\leq \ell \leq \ell'\leq K} \frac{\partial (G_n)_{\ell\ell'}}{\partial R_{\ell\ell'}}\big\vert_{t,R}
= n(p-1)\mathrm{Tr}(\E[\langle \bQ \rangle_{t,R}\,]^{\circ(p-2)}\pmb{\Delta})
$$
with
\begin{equation}\label{eq:def_Delta}
\pmb{\Delta} \coloneqq
	\E\Bigg[\bigg\langle \bigg(\frac{\bQ + \bQ^T}{2} - \bigg\langle \frac{\bQ + \bQ^T}{2} \bigg\rangle_{\!\! t,R}\,\bigg)^{\!\! \circ 2}\, \bigg\rangle_{\!\! t,R} \!
	- \bigg(\bigg\langle\frac{\bQ + \bQ^T}{2}\bigg\rangle_{\!\! t,R} - \frac{\langle\bx\rangle_{t,R}^T \langle\bx\rangle_{t,R}}{n} \bigg)^{\!\! \circ 2} \, \Bigg]\:.
\end{equation}
\end{lemma}
\begin{proof}
To lighten notations, we omit the subscripts of the Gibbs brackets $\langle - \rangle_{t,R}$.
Let $(\ell,\ell') \in \{1,\dots,K\}^2$.
The partial derivative of $R \mapsto \big(G_n(t,R)\big)_{\ell \ell'}$ with respect to $R_{\ell\ell'}$ reads:
\begin{equation}\label{partial_derivative_Gn}
\frac{\partial (G_n)_{\ell\ell'}}{\partial R_{\ell\ell'}}\bigg\vert_{t,R}
= \frac{\partial \E[\langle Q_{\ell\ell'} \rangle]^{p-1}}{\partial R_{\ell\ell'}}\bigg\vert_{t,R}
= (p-1)\E[\langle Q_{\ell\ell'} \rangle]^{p-2}\,
\E\bigg[\langle Q_{\ell\ell'} \rangle \bigg\langle\frac{\partial \cH_{t,R}}{\partial R_{\ell\ell'}}\bigg\rangle
-\bigg\langle Q_{\ell\ell'} \frac{\partial \cH_{t,R}}{\partial R_{\ell\ell'}}\bigg\rangle\bigg]\,,
\end{equation}
with
\begin{equation}\label{partial_derivative_HtR}
\frac{\partial \cH_{t,R}}{\partial R_{\ell\ell'}}
\equiv \frac{\partial \cH_{t,R}(\bx;\bY^{(t)},\widetilde{\bY}^{(t,R)})}{\partial R_{\ell\ell'}}
= \sum_{j=1}^{n} \frac{1}{2}x_j^T \frac{\partial R}{\partial R_{\ell\ell'}} x_j - X_j^T \frac{\partial R}{\partial R_{\ell\ell'}} x_j
- \widetilde{Z}_j^T \frac{\partial \sqrt{R}}{\partial R_{\ell\ell'}}x_j\,.
\end{equation}
Once the identity \eqref{partial_derivative_HtR} has been plugged in the right-hand side of \eqref{partial_derivative_Gn}, two expectations involving the Gaussian randon vectors $\widetilde{Z}_j$, $j=1\dots n$, appear.
An integration by parts with respect to the Gaussian random variables $\widetilde{Z}_{jk}$, $k \in \{1,\dots,K\}$, gives:
\begin{align*}
\E\bigg[\bigg\langle Q_{\ell\ell'} \widetilde{Z}_j^T \frac{\partial \sqrt{R}}{\partial R_{\ell\ell'}}x_j \bigg\rangle\bigg]
&=\sum_{k=1}^K \E\bigg[\widetilde{Z}_{jk}  \bigg\langle Q_{\ell\ell'} \bigg(\frac{\partial \sqrt{R}}{\partial R_{\ell\ell'}}x_j\bigg)_{\!\! k}\, \bigg\rangle\bigg]\\
&=\sum_{k=1}^K \E\bigg[\bigg\langle Q_{\ell\ell'} \big(\sqrt{R}x_j\big)_k \bigg(\frac{\partial \sqrt{R}}{\partial R_{\ell\ell'}}x_j\bigg)_{\!\! k}\, \bigg\rangle\bigg]
-\E\bigg[\bigg\langle Q_{\ell\ell'} \bigg(\frac{\partial \sqrt{R}}{\partial R_{\ell\ell'}}x_j\bigg)_{\!\! k} \,\bigg\rangle\big\langle \big(\sqrt{R}x_j\big)_k \big\rangle \bigg]\\
&=\E\bigg[\bigg\langle Q_{\ell\ell'} x_j^T \sqrt{R}\frac{\partial \sqrt{R}}{\partial R_{\ell\ell'}}x_j \bigg\rangle\bigg]
-\E\bigg[\big\langle Q_{\ell\ell'} x_j^T \big\rangle
\frac{\partial \sqrt{R}}{\partial R_{\ell\ell'}} \sqrt{R} \langle x_j \rangle \bigg]\\
&=\frac{1}{2}\E\bigg[\bigg\langle Q_{\ell\ell'} x_j^T \frac{\partial R}{\partial R_{\ell\ell'}}x_j \bigg\rangle\bigg]
-\E\bigg[\big\langle Q_{\ell\ell'} x_j^T \big\rangle
\frac{\partial \sqrt{R}}{\partial R_{\ell\ell'}} \sqrt{R} \langle x_j \rangle \bigg]\;;
\end{align*}
\begin{align*}
&\E\bigg[\langle Q_{\ell\ell'} \rangle \bigg\langle \widetilde{Z}_j^T \frac{\partial \sqrt{R}}{\partial R_{\ell\ell'}}x_j \bigg\rangle\bigg]\\
&\qquad\quad=\sum_{k=1}^K \E\bigg[ \widetilde{Z}_{jk}  \langle Q_{\ell\ell'} \rangle \bigg\langle\bigg(\frac{\partial \sqrt{R}}{\partial R_{\ell\ell'}}x_j\bigg)_{\!\! k}\, \bigg\rangle\bigg]\\
&\qquad\quad=\sum_{k=1}^K \E\bigg[\langle Q_{\ell\ell'} \rangle \bigg\langle \big(\sqrt{R}x_j\big)_k \bigg(\frac{\partial \sqrt{R}}{\partial R_{\ell\ell'}}x_j\bigg)_{\!\! k}\, \bigg\rangle\bigg]
-2\E\bigg[\langle Q_{\ell\ell'} \rangle \big\langle \big(\sqrt{R}x_j\big)_k\big\rangle\bigg\langle \bigg(\frac{\partial \sqrt{R}}{\partial R_{\ell\ell'}}x_j\bigg)_{\!\! k}\, \bigg\rangle\bigg]\\
&\qquad\qquad\qquad\qquad\qquad\qquad\qquad\qquad\qquad\qquad\quad\,
+\; \E\bigg[\big\langle Q_{\ell\ell'} \big(\sqrt{R}x_j\big)_k\big\rangle\bigg\langle \bigg(\frac{\partial \sqrt{R}}{\partial R_{\ell\ell'}}x_j\bigg)_{\!\! k}\, \bigg\rangle\bigg]\\
&\qquad\quad=\E\bigg[\langle Q_{\ell\ell'} \rangle \bigg\langle x_j^T \sqrt{R} \frac{\partial \sqrt{R}}{\partial R_{\ell\ell'}}x_j\bigg\rangle\bigg]
-2\E\bigg[\langle Q_{\ell\ell'} \rangle \big\langle x_j \rangle^T \sqrt{R}\frac{\partial \sqrt{R}}{\partial R_{\ell\ell'}} \langle x_j \rangle\bigg]
+ \E\bigg[\big\langle Q_{\ell\ell'} x_j^T \big\rangle \sqrt{R}\frac{\partial \sqrt{R}}{\partial R_{\ell\ell'}} \langle x_j \rangle\bigg]\\
&\qquad\quad= \frac{1}{2}\E\bigg[\langle Q_{\ell\ell'} \rangle \bigg\langle x_j^T \frac{\partial R}{\partial R_{\ell\ell'}}x_j\bigg\rangle\bigg]
-\E\bigg[\langle Q_{\ell\ell'} \rangle \big\langle x_j \rangle^T \frac{\partial R}{\partial R_{\ell\ell'}} \langle x_j \rangle\bigg]
+ \E\bigg[\big\langle Q_{\ell\ell'} x_j^T \big\rangle \sqrt{R}\frac{\partial \sqrt{R}}{\partial R_{\ell\ell'}} \langle x_j \rangle\bigg]\,.
\end{align*}
In both chains of equalities, the last one follows from an identity similar to \eqref{eq:prop3_derivative_R}, i.e.,
\begin{equation}\label{eq:identity_partial_diff_Rll'}
\forall v \in \mathbb{R}^K:
v^T \sqrt{R} \frac{\partial\sqrt{R}}{\partial R_{\ell\ell'}} v = \frac{1}{2} v^T
\bigg(\sqrt{R} \frac{\partial\sqrt{R}}{\partial R_{\ell\ell'}} + \frac{\partial \sqrt{R}}{\partial R_{\ell\ell'}} \sqrt{R}\bigg) v
= \frac{1}{2} v^T \frac{\partial R}{\partial R_{\ell\ell'}} v\,.
\end{equation}
Making use of the two identities yielded by the integration by parts, as well as \eqref{eq:identity_partial_diff_Rll'}, we find:
\begin{multline}
\E\bigg[\langle Q_{\ell\ell'} \rangle \bigg\langle\frac{\partial \cH_{t,R}}{\partial R_{\ell\ell'}}\bigg\rangle
-\bigg\langle Q_{\ell\ell'} \frac{\partial \cH_{t,R}}{\partial R_{\ell\ell'}}\bigg\rangle\bigg]
= \sum_{j=1}^n
\E\bigg[\bigg\langle Q_{\ell\ell'} X_j^T \frac{\partial R}{\partial R_{\ell\ell'}} x_j \bigg\rangle
-\langle Q_{\ell\ell'} \rangle X_j^T \frac{\partial R}{\partial R_{\ell\ell'}} \langle x_j \rangle\bigg]\\
+ \E\bigg[\langle Q_{\ell\ell'} \rangle \langle x_j \rangle^T \frac{\partial R}{\partial R_{\ell\ell'}} \langle x_j \rangle
- \big\langle Q_{\ell\ell'} x_j^T \big\rangle \frac{\partial R}{\partial R_{\ell\ell'}} \langle x_j \rangle\bigg].\label{eq:difference_partial_derivative_QH}
\end{multline}
Thanks to the Nishimori identity, we have
\begin{equation*}
\E\bigg[\langle Q_{\ell\ell'} \rangle \langle x_j \rangle^T \frac{\partial R}{\partial R_{\ell\ell'}} \langle x_j \rangle
- \big\langle Q_{\ell\ell'} x_j^T \big\rangle \frac{\partial R}{\partial R_{\ell\ell'}} \langle x_j \rangle\bigg]
= \E\bigg[\langle Q_{\ell'\ell} \rangle \langle x_j \rangle^T \frac{\partial R}{\partial R_{\ell\ell'}} \langle x_j \rangle
- \big\langle Q_{\ell'\ell} \big\rangle X_j^T\frac{\partial R}{\partial R_{\ell\ell'}} \langle x_j \rangle\bigg]\,,
\end{equation*}
and \eqref{eq:difference_partial_derivative_QH} further simplifies (the last equality uses the cyclic property of the trace):
\begin{align}
&\E\bigg[\langle Q_{\ell\ell'} \rangle \bigg\langle\frac{\partial \cH_{t,R}}{\partial R_{\ell\ell'}}\bigg\rangle
-\bigg\langle Q_{\ell\ell'} \frac{\partial \cH_{t,R}}{\partial R_{\ell\ell'}}\bigg\rangle\bigg]\nonumber\\
&\qquad\qquad=\sum_{j=1}^{n}\E\bigg[\bigg\langle Q_{\ell\ell'} X_j^T \frac{\partial R}{\partial R_{\ell\ell'}} \big(x_j-\langle x_j \rangle\big) \bigg\rangle\bigg]
-\E\bigg[\langle Q_{\ell'\ell} \rangle (X_j - \langle x_j \rangle)^T \frac{\partial R}{\partial R_{\ell\ell'}} \langle x_j \rangle\bigg]\nonumber\\
&\qquad\qquad= \E\bigg[\bigg\langle Q_{\ell\ell'} \mathrm{Tr}\bigg(\bX \frac{\partial R}{\partial R_{\ell\ell'}} (\bx-\langle \bx \rangle)^T\bigg) \bigg\rangle\bigg]
- \E\bigg[\langle Q_{\ell'\ell} \rangle 
\mathrm{Tr}\bigg((\bX-\langle \bx \rangle) \frac{\partial R}{\partial R_{\ell\ell'}} \langle \bx \rangle^T\bigg)\bigg]\nonumber\\
&\qquad\qquad= n\E\bigg[\bigg\langle Q_{\ell\ell'} \mathrm{Tr}\bigg(\frac{\partial R}{\partial R_{\ell\ell'}} (\bQ-\langle \bQ \rangle)\bigg) \bigg\rangle\bigg]
- n\E\bigg[\langle Q_{\ell'\ell} \rangle\,
\mathrm{Tr}\bigg(\!\frac{\partial R}{\partial R_{\ell\ell'}} \bigg(\langle \bQ \rangle - \frac{\langle \bx \rangle^T\langle \bx \rangle}{n}\bigg)\bigg)\bigg]\,.\label{intermediary_formula_delta}
\end{align}
Now consider the case $\ell \neq \ell'$.
All the entries of $\nicefrac{\partial R}{\partial R_{\ell\ell'}}$ are zeros save for the entries $(\ell,\ell')$ and $(\ell',\ell)$ which are both one. Equation \eqref{intermediary_formula_delta} then reads:
\begin{multline}\label{final_formula_delta}
\E\bigg[\langle Q_{\ell\ell'} \rangle \bigg\langle\frac{\partial \cH_{t,R}}{\partial R_{\ell\ell'}}\bigg\rangle
-\bigg\langle Q_{\ell\ell'} \frac{\partial \cH_{t,R}}{\partial R_{\ell\ell'}}\bigg\rangle\bigg]\\
= n\E\Big[\Big\langle Q_{\ell\ell'} \big(\bQ + \bQ^T - \big\langle \bQ + \bQ^T \big\rangle\big)_{ \ell\ell'}\,\Big\rangle\Big]
- n\E\bigg[\langle Q_{\ell'\ell} \rangle
\bigg(\langle \bQ + \bQ^T \rangle - 2\frac{\langle \bx \rangle^T\langle \bx \rangle}{n}\bigg)_{\!\! \ell\ell'} \,\bigg]\;.
\end{multline}
Combining \eqref{partial_derivative_Gn} and \eqref{final_formula_delta} gives the identity \eqref{eq:partial_derivative_F} when $\ell \neq \ell'$. The case $\ell=\ell'$ is obtained in a similar way except that now the entries of $\nicefrac{\partial R}{\partial R_{\ell\ell}}$ are zeros save for the entry $(\ell,\ell)$ which is one.\\

We can now prove the identity for the divergence of $G_n$. This divergence, denoted $\mathcal{D}$, satisfies:
\begin{equation}\label{eq:identities_divergence}
\mathcal{D} = \sum_{\ell \leq \ell'} \!\! \frac{\partial (G_n)_{\ell\ell'}}{\partial R_{\ell\ell'}}\bigg\vert_{t,R}
= \sum_{\ell \leq \ell'} \!\! \frac{\partial (G_n)_{\ell'\ell}}{\partial R_{\ell'\ell}}\bigg\vert_{t,R}
= \frac{1}{2}\sum_{\ell \leq \ell'} \!\! \frac{\partial (G_n)_{\ell\ell'}}{\partial R_{\ell\ell'}}\bigg\vert_{t,R}
+ \frac{1}{2}\sum_{\ell \leq \ell'} \!\! \frac{\partial (G_n)_{\ell'\ell}}{\partial R_{\ell'\ell}}\bigg\vert_{t,R} .
\end{equation}
In the last equality of \eqref{eq:identities_divergence}, replacing the summands by their formula \eqref{eq:partial_derivative_F} yields:
\begin{align}
\mathcal{D}
&= \frac{n(p-1)}{2} \sum_{\ell,\ell'=1}^K \E\big[\big\langle Q_{\ell\ell'} \big\rangle\,\big]^{p-2}\bigg(
\E\Big[\Big\langle\bQ \circ \Big(\bQ + \bQ^T - \big\langle \bQ + \bQ^T\big\rangle\Big)\Big\rangle\Big]\Big\vert_{\ell\ell'}\nonumber\\
&\qquad\qquad\qquad\qquad\qquad\qquad\qquad\qquad\;\;\:-\E\bigg[\big\langle \bQ^T \big\rangle\circ\bigg(\big\langle \bQ + \bQ^T\big\rangle - 2\frac{\langle \bx \rangle^T \langle \bx \rangle}{n}\bigg)\bigg]\bigg\vert_{\ell\ell'}
\bigg)\nonumber\\
&= \frac{n(p-1)}{2} \mathrm{Tr}\Big(\E[\langle \bQ \rangle]^{\circ(p-2)}\,
\E\Big[\Big\langle\bQ^T \circ \Big(\bQ + \bQ^T - \big\langle \bQ + \bQ^T\big\rangle\Big)\Big\rangle\Big]\Big)\nonumber\\
&\qquad\qquad\qquad\qquad\quad-\frac{n(p-1)}{2} \mathrm{Tr}\bigg(
\E[\langle \bQ \rangle]^{\circ(p-2)}\,
\E\bigg[\langle \bQ \rangle\circ\bigg(\big\langle \bQ + \bQ^T\big\rangle - 2\frac{\langle \bx \rangle^T \langle \bx \rangle}{n}\bigg)\bigg]\bigg)\,.\label{eq:div_diff_traces}
\end{align}
Remember that $\E[\langle \bQ \rangle_{t,R}\,]$, and therefore $\E[\langle \bQ \rangle_{t,R}\,]^{\circ (p-2)}$, is symmetric. Using that the trace is invariant by transposition and cyclic permutation, the two traces in~\eqref{eq:div_diff_traces} read:
\begin{align*}
&\mathrm{Tr}\Big(\E[\langle \bQ \rangle]^{\circ(p-2)}\:
\E\Big[\Big\langle\bQ^T \circ \Big(\bQ + \bQ^T - \big\langle \bQ + \bQ^T\big\rangle\Big)\Big\rangle\Big]\Big)\\
&\qquad\qquad\qquad\qquad\qquad\qquad\qquad=\frac{1}{2}\mathrm{Tr}\Big(\E[\langle \bQ \rangle]^{\circ(p-2)}\:
\E\Big[\Big\langle (\bQ+\bQ^T) \circ \Big(\bQ + \bQ^T - \big\langle \bQ + \bQ^T\big\rangle\Big)\Big\rangle\Big]\Big)\,;\\
&\mathrm{Tr}\bigg(
\E[\langle \bQ \rangle]^{\circ(p-2)}\:
\E\bigg[\langle \bQ \rangle\circ\bigg(\big\langle \bQ + \bQ^T\big\rangle -2\frac{\langle \bx \rangle^T \langle \bx \rangle}{n}\bigg)\bigg]\bigg)\\
&\qquad\qquad\qquad\qquad\qquad\qquad\qquad=\frac{1}{2}\mathrm{Tr}\bigg(
\E[\langle \bQ \rangle]^{\circ(p-2)}\:
\E\bigg[\big\langle \bQ+\bQ^T\big\rangle\circ\bigg(\big\langle \bQ + \bQ^T\big\rangle -2\frac{\langle \bx \rangle^T \langle \bx \rangle}{n}\bigg)\bigg]\bigg)\,.
\end{align*}
Clearly,
$\E\,\big\langle (\bQ+\bQ^T) \circ \big(\bQ + \bQ^T - \langle \bQ + \bQ^T\rangle\big)\big\rangle
= \E\,\big\langle \big(\bQ + \bQ^T - \langle \bQ + \bQ^T \rangle\big)^{\circ 2} \big\rangle$.
Similarly, we have
$$
\E\bigg[\big\langle \bQ+\bQ^T\big\rangle\circ\bigg(\big\langle \bQ + \bQ^T\big\rangle - 2\frac{\langle \bx \rangle^T \langle \bx \rangle}{n}\bigg)\bigg]
= \E\bigg[\bigg(\big\langle \bQ + \bQ^T\big\rangle - 2\frac{\langle \bx \rangle^T \langle \bx \rangle}{n}\bigg)^{\!\!\circ 2}\,\bigg]\;.
$$
For this last equality, we could complete the square thanks to the following term being zero:
\begin{align*}
\E\bigg[2\frac{\langle \bx \rangle^T \langle \bx \rangle}{n}\circ\bigg(\big\langle \bQ + \bQ^T\big\rangle - 2\frac{\langle \bx \rangle^T \langle \bx \rangle}{n}\bigg)\bigg]
&= 2\E\bigg[\frac{\langle \bx \rangle^T \langle \bx \rangle}{n} \circ \big\langle \bQ + \bQ^T\big\rangle \bigg]
-4\E\bigg[\bigg(\frac{\langle \bx \rangle^T \langle \bx \rangle}{n}\bigg)^{\!\! \circ 2}\,\bigg]\\
&= 2\E\bigg[\frac{\langle \bx \rangle^T \langle \bx \rangle}{n} \circ
\frac{\langle\bx\rangle^T \bX + \bX^T\langle\bx\rangle}{n} \bigg]
-4\E\bigg[\bigg(\frac{\langle \bx \rangle^T \langle \bx \rangle}{n}\bigg)^{\!\! \circ 2}\,\bigg]\\
&= 2\E\bigg[\bigg\langle\frac{\langle \bx \rangle^T \langle \bx \rangle}{n} \circ
\frac{\langle\bx\rangle^T \bx + \bx^T\langle\bx\rangle}{n} \bigg\rangle\bigg]
-4\E\bigg[\bigg(\frac{\langle \bx \rangle^T \langle \bx \rangle}{n}\bigg)^{\!\! \circ 2}\,\bigg]\\
&= 2\E\bigg[\frac{\langle \bx \rangle^T \langle \bx \rangle}{n} \circ
\frac{\langle\bx\rangle^T \langle\bx\rangle + \langle\bx\rangle^T\langle\bx\rangle}{n}\bigg]
-4\E\bigg[\bigg(\frac{\langle \bx \rangle^T \langle \bx \rangle}{n}\bigg)^{\!\! \circ 2}\,\bigg]\\
&=0\;.
\end{align*}
Plugging these identities back in~\eqref{eq:div_diff_traces}, we finally obtain:
\begin{align*}
\mathcal{D}
&= \frac{n(p-1)}{4}\mathrm{Tr}\bigg(\E[\langle \bQ \rangle_{t,R}\,]^{\circ(p-2)}\;
\E\bigg[\Big\langle \Big(\bQ + \bQ^T - \big\langle \bQ + \bQ^T\big\rangle\Big)^{\!\!\circ 2} \Big\rangle
-
\bigg(\big\langle \bQ + \bQ^T\big\rangle - 2\frac{\langle \bx \rangle^T \langle \bx \rangle}{n}\bigg)^{\!\!\circ 2}\,\bigg]\bigg)\:,
\end{align*}
where we recognize that the second expectation $\mathbb{E}[-]$ is equal to $4\pmb{\Delta}$ (see definition \eqref{eq:def_Delta}).
\end{proof}
\section{Concentration of the free entropy}\label{app:concentration_free_entropy}
Consider the inference problem~\eqref{inference_problem_t_R}.
The associated Hamiltonian reads
\begin{multline}
\cH_{t,R}(\bx ; \bY,\widetilde{\bY})
=
\sum_{i \in \mathcal{I}}
\frac{ (1-t) (p-1)!}{2n^{p-1}} \Bigg(\sum_{k=1}^{K}\prod_{a=1}^{p} x_{i_a k}\Bigg)^{\!\! 2}
- \sqrt{\frac{ (1-t) (p-1)!}{n^{p-1}}} Y_{i}  \sum_{k=1}^{K}\prod_{a=1}^{p} x_{i_a k}\\
+  \sum_{j=1}^{n} \frac{1}{2} x_j^T R x_j- \widetilde{Y}_j^T  \sqrt{R} x_j \,.
\end{multline}
In this section we show that the free entropy
\begin{equation}
\frac{1}{n} \ln \cZ_{t,R}\big(\bY^{(t)},\widetilde{\bY}^{(t,R)}\big)
= \frac{1}{n} \ln\Bigg( \int \prod_{i=1}^n dP_{X}(x_i) \: e^{-\cH_{t,R}(\bx; \bY^{(t)},\widetilde{\bY}^{(t,R)})}\Bigg)
\end{equation}
concentrates around its expectation.
We will sometimes write $\frac{1}{n} \ln \cZ_{t,R}$, omitting the arguments, to shorten notations.
\begin{theorem}[Concentration of the free entropy]\label{th:concentration_free_entropy}
	Assume $P_X$ has finite $(4p-4)$\textsuperscript{th} order moments.
	There exists a positive constant $C$ depending only on $P_X$, $K$, $p$ and $\Vert R \Vert$ such that
	\begin{equation}\label{bound_variance_free_entropy}
	\E \Bigg[\Bigg(\frac{ \ln \cZ_{t,R}}{n}
	- \E\bigg[\frac{\ln \cZ_{t,R}}{n} \bigg]
	\Bigg)^{\!\! 2}\:\Bigg]
	\leq \frac{C}{n} \:.
	\end{equation}
\end{theorem}
	\begin{proof}
		To lighten notations we drop the subscripts of the Gibbs brackets $\langle - \rangle_{t,R}$.
		First, we show that the free entropy concentrates on its conditional expectation given the Gaussian noise $\bZ$, $\widetilde{\bZ}$.
		Thus, $\nicefrac{\ln  \cZ_{t,R}}{n}$ is seen as a function of $X_1,\dots,X_n$ \textit{only} and we work conditionally to $\bZ,\widetilde{\bZ}$.
		Let $X'_1,\dots,X'_n$ be i.i.d.\ samples from $P_X$, independent of $\bX$. 
		For all $j \in \{1,\dots,n\}$, we define
		$$
		\cZ_{t,R}^{(j)}\big(\bY^{(j,t)},\widetilde{\bY}^{(j,t,R)}\big) = \int \prod_{i=1}^n dP_{X}(x_i) \: e^{-\cH_{t,R}(\bx; \bY^{(j,t)},\widetilde{\bY}^{(j,t,R)})} \:,
		$$
		where $\bY^{(j,t)}$, $\widetilde{\bY}^{(j,t,R)}$ are obtained from $\bY^{(t)}$, $\widetilde{\bY}^{(t,R)}$ by replacing $X_j$ by $X'_j$.
		We can consider an inference problem similar to \eqref{inference_problem_t_R} for which the observations are $\bY^{(j,t)}$, $\widetilde{\bY}^{(j,t,R)}$.
		Then the Gibbs brackets associated to the posterior distribution are
		$$
		\langle - \rangle_{(j)} = \int - \,\prod_{i=1}^n dP_{X}(x_i) \: e^{-\cH_{t,R}(\bx; \bY^{(j,t)},\widetilde{\bY}^{(j,t,R)})}\,.
		$$
		By the Efron-Stein inequality (see \cite[Theorem 3.1]{Boucheron2013Concentration}), we have:
		\begin{equation}\label{efron_stein}
		\E \Bigg[\Bigg(\frac{ \ln \cZ_{t,R}}{n}
		- \E\bigg[\frac{\ln \cZ_{t,R}}{n} \bigg\vert \bZ,\widetilde{\bZ} \bigg]
		\Bigg)^{\!\! 2}\:\Bigg]
		\leq \frac{1}{2} \sum_{j=1}^{n} \E\bigg[\bigg(\frac{\ln  \cZ_{t,R}}{n} - \frac{\ln  \cZ_{t,R}^{(j)}}{n}\bigg)^{\!\! 2}\,\bigg]\,.
		\end{equation}
		Fix $j \in \{1,\dots,n\}$. By Jensen's inequality, note that
		\begin{multline}\label{bounds_delta_free_entropy}
		\frac{1}{n}\big\langle \cH_{t,R}(\bx; \bY^{(j,t)},\widetilde{\bY}^{(j,t,R)}) - \cH_{t,R}(\bx; \bY^{(t)},\widetilde{\bY}^{(t,R)})  \big\rangle_{(j)}\\
		\leq \frac{\ln  \cZ_{t,R}}{n} - \frac{\ln  \cZ_{t,R}^{(j)}}{n}
		\leq \frac{1}{n}\big\langle \cH_{t,R}(\bx; \bY^{(j,t)},\widetilde{\bY}^{(j,t,R)}) - \cH_{t,R}(\bx; \bY^{(t)},\widetilde{\bY}^{(t,R)})  \big\rangle\:.
		\end{multline}
		Define $\mathcal{I}_j = \{i \in \mathcal{I}: \exists b \in \{1,\dots,p\}: i_b = j\}$ and
		$\forall i \in \mathcal{I}_j:c(i) = \big\vert\big\{a \in \{1,\dots,p\}: i_a=j\big\}\big\vert$.
		The quantity inbetween the Gibbs brackets in \eqref{bounds_delta_free_entropy} reads:
		\begin{multline}
		\cH_{t,R}(\bx; \bY^{(j,t)},\widetilde{\bY}^{(j,t,R)}) - \cH_{t,R}(\bx; \bY^{(t)},\widetilde{\bY}^{(t,R)})\\
		= \frac{(1-t) (p-1)!}{n^{p-1}}\sum_{i \in \mathcal{I}_j}
		\sum_{\ell,\ell'=1}^{K}(X_{j\ell}^{c(i)}-X_{j\ell}^{\prime c(i)})
		\prod_{\substack{a=1\\i_a \neq j}}^{p} X_{i_a \ell}\prod_{a=1}^{p} x_{i_a \ell'}
		+\big(X_j - X'_j\big)^T  R \,x_j \:.
		\end{multline}
		Using Jensen's inequality, we further obtain:
		\begin{multline}\label{expectation_delta_H}
		\E\big[\big\langle \cH_{t,R}(\bx; \bY^{(j,t)},\widetilde{\bY}^{(j,t,R)}) - \cH_{t,R}(\bx; \bY^{(t)},\widetilde{\bY}^{(t,R)})  \big\rangle^2\big]\\
		\leq
		\frac{2 ((p-1)!)^2K^2 \vert \mathcal{I}_j\vert}{n^{2p-2}}\sum_{i \in \mathcal{I}_j}
		\sum_{\ell,\ell'=1}^{K}
		\E\Bigg[(X_{j\ell}^{c(i)} - X_{j\ell}^{\prime c(i)})^2
		\prod_{\substack{a=1\\i_a \neq j}}^{p} X_{i_a \ell}^2\,\Bigg\langle\prod_{a=1}^{p} x_{i_a \ell'}\Bigg\rangle^{\!\! 2}\,\Bigg]\\
		+2\E\Big[\Big(\big(X_j-X'_j\big)^T  R \langle x_j \rangle\Big)^2\Big].
		\end{multline}
		We now bound each summand on the right-hand side of \eqref{expectation_delta_H} separately.
		For all $i \in \mathcal{I}_j$ and $(\ell,\ell') \in \{1,\dots,K\}^2$:
		\begin{align*}
		&\E\Bigg[(X_{j\ell}^{c(i)} - X_{j\ell}^{\prime c(i)})^2
		\prod_{\substack{a=1\\i_a \neq j}}^{p} X_{i_a \ell}^2\,\Bigg\langle\prod_{a=1}^{p} x_{i_a \ell'}\Bigg\rangle^{\!\!\! 2}\,\Bigg]\\
		&\qquad\qquad\qquad\qquad\qquad\qquad\qquad\qquad
		\leq \E\Bigg[(X_{j\ell}^{c(i)} - X_{j\ell}^{\prime c(i)})^4
		\prod_{\substack{a=1\\i_a \neq j}}^{p} X_{i_a \ell}^4\Bigg]^{\nicefrac{1}{2}}
		\E\Bigg[\Bigg\langle\prod_{a=1}^{p} x_{i_a \ell'}\Bigg\rangle^{\!\! 4}\,\Bigg]^{\nicefrac{1}{2}}\\
		&\qquad\qquad\qquad\qquad\qquad\qquad\qquad\qquad
		\leq \E\Bigg[(X_{j\ell}^{c(i)} - X_{j\ell}^{\prime c(i)})^4
		\prod_{\substack{a=1\\i_a \neq j}}^{p} X_{i_a \ell}^4\Bigg]^{\nicefrac{1}{2}}
		\E\Bigg[\Bigg\langle\prod_{a=1}^{p} x_{i_a \ell'}^4\Bigg\rangle\,\Bigg]^{\nicefrac{1}{2}}\\
		&\qquad\qquad\qquad\qquad\qquad\qquad\qquad\qquad
		= \E\Bigg[(X_{j\ell}^{c(i)} - X_{j\ell}^{\prime c(i)})^4
		\prod_{\substack{a=1\\i_a \neq j}}^{p} X_{i_a \ell}^4\Bigg]^{\nicefrac{1}{2}}
		\E\Bigg[\prod_{a=1}^{p} X_{i_a \ell'}^4\Bigg]^{\nicefrac{1}{2}}\\
		&\qquad\qquad\qquad\qquad\qquad\qquad\qquad\qquad
		= \sqrt{\E\big[\big(X_{j\ell}^{c(i)} - X_{j\ell}^{\prime c(i)}\big)^4\,\big]
		\E\big[X_{j\ell'}^{4c(i)}\big]
		\E\Bigg[\prod_{\substack{a=1\\i_a \neq j}}^{p} X_{i_a \ell}^4\Bigg]
		\E\Bigg[\prod_{\substack{a=1\\i_a \neq j}}^{p} X_{i_a \ell'}^4\Bigg]}\;.
		\end{align*}
		The first inequality follows from the Cauchy-Schwarz inequality, the second one from Jensen's inequality, and the first equality from the Nishimori identity. The final bound is finite given that $P_X$ has finite $(4p-4)$\textsuperscript{th} order moments.
		Hence, there exists a positive constant $C$ depending only on $P_X$, $K$ and $p$ such that the first term on the right-hand side of \eqref{expectation_delta_H} is bounded by $\nicefrac{C \vert\mathcal{I}_j\vert^2}{n^{2p-2}} \leq C$ (as $\vert \mathcal{I}_j\vert \leq n^{p-1}$).
		Regarding the second term on the right-hand side of \eqref{expectation_delta_H}, we easily get:
		\begin{equation*}
		\E\Big[\Big(\big(X'_j -X_j\big)^T  R \langle x_j \rangle\Big)^2\Big]
		\leq \E\big[\Vert X'_j -X_j\Vert^2 \Vert R\Vert^2 \Vert\langle x_j \rangle\Vert^2\big]
		\leq \Vert R \Vert^2 \E[\Vert X'_j -X_j \Vert^4 ]^{\nicefrac{1}{2}}
		\E[\Vert X_j \Vert^4 ]^{\nicefrac{1}{2}}\:.
		\end{equation*}
		We conclude that there exists a positive constant $C$ depending only on $P_X$, $K$, $p$ and $\Vert R \Vert$ such that
		\begin{equation}\label{bound_delta_hamiltonian}
		\forall j \in \{1,\dots,n\}: \E\big[\big\langle \cH_{t,R}(\bx; \bY^{(j,t)},\widetilde{\bY}^{(j,t,R)}) - \cH_{t,R}(\bx; \bY^{(t)},\widetilde{\bY}^{(t,R)})  \big\rangle^2\,\big]
		\leq C \,.
		\end{equation}
		A similar bound holds when the Gibbs brackets $\langle - \rangle$ are replaced by $\langle - \rangle_{(j)}$.
		Finally, combining \eqref{efron_stein}, \eqref{bounds_delta_free_entropy} and \eqref{bound_delta_hamiltonian}, we obtain the desired upper bound:
		\begin{equation}\label{bound_variance_EF}
		\E \Bigg[\Bigg(\frac{ \ln \cZ_{t,R}}{n}
		- \E\bigg[\frac{\ln \cZ_{t,R}}{n} \bigg\vert \bZ,\widetilde{\bZ} \bigg]
		\Bigg)^{\!\! 2}\:\Bigg]
		\leq \frac{C}{n}\;,
		\end{equation}
		where the positive constant $C$ is not necessarily the same than before but still only depends on $P_X$, $K$, $p$ and $\Vert R \Vert$.\\
		
		The second -- and final -- step is to show that the conditional expectation of the free entropy given $\bZ,\widetilde{\bZ} $ concentrates on its expectation.
		Let $g(\bZ,\widetilde{\bZ}) = \E[\nicefrac{\ln  \cZ_{t,R}}{n} \vert \bZ, \widetilde{\bZ}]$.
		By the Gaussian-Poincar\'{e} inequality (see \cite[Theorem 3.20]{Boucheron2013Concentration}), we have:
		\begin{equation}\label{gaussian_poincare}
		\E \Bigg[\Bigg(\E\bigg[\frac{ \ln \cZ_{t,R}}{n} \bigg\vert \bZ,\widetilde{\bZ}\bigg]
		- \E\bigg[\frac{\ln \cZ_{t,R}}{n}\bigg]\Bigg)^{\!\! 2}\,\Bigg]
		\leq \E\,\big\Vert \nabla g(\bZ,\widetilde{\bZ}) \big\Vert^2 \;.
		\end{equation}
		The squared norm of the gradient of $g$ reads
		$\Vert \nabla g\Vert^2 =
		\sum_{i \in \mathcal{I}} \vert\nicefrac{\partial g}{\partial Z_i}\vert^2
		+ \sum_{j=1}^{n}\sum_{\ell=1}^{K} \vert\nicefrac{\partial g}{\partial \widetilde{Z}_{j\ell}}\vert^2$.
		Each of these partial derivatives takes the form 
		$\partial g = -n^{-1} \big\langle \partial \mathcal{H}_{t,R}\big\rangle$. More precisely:
		\begin{equation*}
		\bigg\vert\frac{\partial g}{\partial Z_i}\bigg\vert
		= n^{-1}\Bigg\vert \Bigg\langle  \sqrt{\frac{(1-t) (p-1)!}{n^{p-1}}} \sum_{k=1}^{K}\prod_{a=1}^{p} x_{i_a k} \Bigg\rangle \Bigg\vert\quad ; \quad
		\bigg\vert \frac{\partial g}{\partial \widetilde{Z}_{j\ell}} \bigg\vert
		= n^{-1} \big\vert \big\langle \big(\sqrt{R}x_j\big)_{\ell} \big\rangle \big\vert \,.
		\end{equation*}
		On one hand, we have
		\begin{multline}\label{first_upperbound_GP}
		\sum_{i \in \mathcal{I}} \E\,\bigg\vert\frac{\partial g}{\partial Z_i}\bigg\vert^2
		\leq \frac{K (p-1)!}{n^{p+1}}\sum_{i \in \mathcal{I}}\sum_{k=1}^{K}
		\E\Bigg[\Bigg\langle \prod_{a=1}^{p} x_{i_a k} \Bigg\rangle^{\!\! 2}\,\Bigg]\\
		\leq \frac{K (p-1)!}{n^{p+1}}\sum_{i \in \mathcal{I}}\sum_{k=1}^{K}
		\E\Bigg[\Bigg\langle \prod_{a=1}^{p} x_{i_a k}^2 \Bigg\rangle\Bigg]
		= \frac{ K (p-1)!}{n^{p+1}}\sum_{i \in \mathcal{I}}\sum_{k=1}^{K}
		\E\Bigg[\prod_{a=1}^{p} X_{i_a k}^2 \Bigg],
		\end{multline}
		where the first two inequalities follow from Jensen's inequality and the equality from the Nishimori identity.
		On the other hand, we have
		\begin{align}\label{second_upperbound_GP}
		\sum_{j=1}^{n}\sum_{\ell=1}^{K} \E\,\bigg\vert\frac{\partial g}{\partial \widetilde{Z}_{j\ell}}\bigg\vert^2
		\leq \frac{1}{n^2} \sum_{j=1}^{n}\sum_{\ell=1}^{K} \E\big[\big\langle \big(\sqrt{R}x_j\big)_{\ell}^2\big\rangle\big]
		&= \frac{1}{n^2} \sum_{j=1}^{n} \E\big[ \big\Vert \sqrt{R}X_j\big\Vert^2\Big]
		\leq \frac{\Vert R \Vert}{n} \E_{X \sim P_X}\big[\Vert X \Vert^2 \big],
		\end{align}
		where the first inequality follows from Jensen's inequality and the equality from the Nishimori identity.
		Both upper bounds in \eqref{first_upperbound_GP} and \eqref{second_upperbound_GP} take the form $\nicefrac{C}{n}$ with $C$ a positive constant $C$ depending only on $P_X$, $K$, $p$ and $\Vert R \Vert$ (remember that $\vert \mathcal{I} \vert \leq n^p$).
		Plugging \eqref{first_upperbound_GP} and \eqref{second_upperbound_GP} in \eqref{gaussian_poincare}, we conclude that
		\begin{equation}\label{bound_variance_GP}
		\E \Bigg[\Bigg(\E\bigg[\frac{ \ln \cZ_{t,R}}{n} \bigg\vert \bZ,\widetilde{\bZ}\bigg]
		- \E\bigg[\frac{\ln \cZ_{t,R}}{n}\bigg]\Bigg)^{\!\! 2}\,\Bigg]
		\leq \frac{C}{n}\,,
		\end{equation}
		where $C$ depends only on $P_X$, $K$, $p$ and $\Vert R \Vert$.
		Combining \eqref{bound_variance_EF} and \eqref{bound_variance_GP} ends the proof of \eqref{bound_variance_free_entropy}.
	\end{proof}
\end{document}